\documentclass[12pt,final,letterpaper]{amsart}

\usepackage[T1]{fontenc}
\usepackage{lmodern}
\usepackage[utf8]{inputenc}
\usepackage[english]{babel}
\usepackage{amsmath}
\usepackage{amssymb}
\usepackage{amsfonts}
\usepackage{amsthm}
\usepackage{graphicx}
\usepackage{enumerate}
\usepackage[usenames]{color}
\usepackage[top=1.25in, bottom=1.25in, left=1.25in, right=1.25in]{geometry}
\usepackage{subcaption}
\usepackage[]{todonotes}
\usepackage{natbib}
\usepackage[final]{showlabels}
\usepackage{hyperref}
\hypersetup{hidelinks}
\usepackage{setspace}
\usepackage{caption}
\usepackage[normalem]{ulem}

\selectfont
\def\red#1{{#1}}

\def\gre#1{{#1}}

\newtheoremstyle{mytheorem}
{6pt}{6pt}
{\itshape}
{0pt}
{}
{}
{1em}
{\underline{\scshape #1~#2}  \ifx\relax#3\relax\else\ (\normalfont #3)\fi}

\newtheoremstyle{myremark}
{6pt}
{10pt}
{\rm}
{-0pt}
{\scshape}
{}
{1em}
{}
\theoremstyle{mytheorem}
\newtheorem{theorem}{Theorem}[section]

\newtheorem{Proposition}[theorem]{Proposition}
\newtheorem{Lemma}[theorem]{Lemma}
\newtheorem{Corollary}[theorem]{Corollary}

\theoremstyle{myremark}
\newtheorem{Remark}[theorem]{Remark}

\newlength{\bibitemsep}
\newlength{\bibparskip}
\let\oldthebibliography\thebibliography
\renewcommand\thebibliography[1]{
\oldthebibliography{#1}
\setlength{\parskip}{\bibitemsep}
\setlength{\itemsep}{\bibparskip}
}
\makeatletter
\renewcommand\@makefntext[1]{
  \setlength{\parindent}{0pt}
  \noindent\@makefnmark\ #1
}
\makeatother

\makeatletter
\renewcommand\@makefntext[1]{
  \setlength{\parindent}{0pt}
  \hangindent=1.8em \hangafter=1
  \noindent\hb@xt@1.8em{\@makefnmark\hss}
  \begin{spacing}{1.0}#1\end{spacing}
}
\makeatother

\makeatletter
\renewcommand\@makefntext[1]{
  \setlength{\parindent}{0pt}
  \noindent\@makefnmark\ #1
}
\makeatother

\selectfont

\renewcommand{\phi}{\varphi}

\newcommand{\RR}{\mathbb{R}}

\newcommand\abs[1]{\left| #1 \right|}
\newcommand\norm[1]{\left|\left| #1 \right|\right|}

\newcommand{\di}{\,\mathrm{d}i}

\newcommand{\gA}{\gamma_A}
\newcommand{\gB}{\gamma_B}

\title[The invisible hand as an emergent property]{The invisible hand as an emergent property
}

\thanks{Giorgio Fabbri, Univ. Grenoble Alpes, CNRS, INRA, Grenoble INP, GAEL, Grenoble, France. Email: \emph{giorgio.fabbri@univ-grenoble-alpes.fr}.}
\thanks{
Davide Fiaschi, University of Pisa, Dipartimento di Economia e Management, Pisa, Italy. davide.fiaschi@unipi.it}
\thanks{
Cristiano Ricci, University of Pisa, Dipartimento di Economia e Management, Pisa, Italy.
cristiano.ricci@unipi.it}

\author[G.Fabbri]{Giorgio Fabbri}
\address{Giorgio Fabbri, Univ. Grenoble Alpes, CNRS, INRA, Grenoble INP, GAEL, Grenoble, France. The work of Giorgio Fabbri is partially supported by the French National Research Agency in the framework of the ``Investissements d'avenir'' program (ANR-15-IDEX-02).}
\email{giorgio.fabbri@univ-grenoble-alpes.fr}

\author[D.Fiaschi]{Davide Fiaschi}
\address{Davide Fiaschi, University of Pisa, Dipartimento di Economia e Management, Pisa, Italy. The work of Davide Fiaschi is partially supported by the University of Pisa, in the framework of the PRA Project PRA202286 (Mobilità di persone e merci nell’Europa)}
\email{davide.fiaschi@unipi.it}

\author[C.Ricci]{Cristiano Ricci}
\address{Cristiano Ricci, University of Pisa, Dipartimento di Economia e Management, Pisa, Italy. The work of Cristiano Ricci is partially supported by the University of Pisa, in the framework of the PRIN project 2017FKHBA8001 (The Time-Space Evolution of Economic Activities: Mathematical Models and Empirical Applications) and PRA Project PRA202286 (Mobilità di persone e merci nell’Europa)}
\email{cristiano.ricci@unipi.it}

\address{\textbf{Acknowledgements:}
We thank the coeditor, Pablo Kurlat, and two anonymous referees for their
helpful comments and suggestions. We also thank Pierre Dehez, Vincenzo
Denicolò, Massimo Giannini, Faruk Gül, Ted Loch-Temzelides, Antonio Minniti, Pier Mario Pacini, Piero Peretto, John Seater, and participants at seminars in Grenoble, Pisa, Rimini, Besançon, \'Evry, and Vienna.}
\address{The dataset used in this paper, as well as the associated R and Julia code to
reproduce the analysis, are publicly available in a Zenodo repository \url{https://doi.org/10.5281/zenodo.16908598}. The original ORBIS data have been obtained through a restricted data use agreement and therefore not available for public dissemination.}

\date{\today}

\newcommand{\rev}[1]{\ifmmode\begingroup\else\leavevmode\begingroup\fi#1\endgroup}
\newcommand{\revisionblockstart}{\begingroup}
\newcommand{\revisionend}{\endgroup}
\newcommand{\revm}[1]{\ifmmode\begingroup\else\leavevmode\begingroup\fi#1\endgroup}
\newcommand{\revb}[1]{{#1}}

\begin{document}
\color{black}

\begin{abstract}

\singlespacing{
We develop a multi-sector competitive economy where firms reallocate across sectors under myopic profit-seeking behaviour and quadratic reallocation costs. The dynamic path, formalised as a gradient flow in Wasserstein space, yields two emergent properties: (i) short-run equilibria can be seen as solutions to sequential aggregate optimisation problems; and (ii) the long-run equilibrium is globally stable and efficient. \red{Linearizing around the common long-run equilibrium, we show that convergence can be faster or slower under a planner who accounts for reallocation costs than under decentralized adjustment.}
}

\end{abstract}

\maketitle



 \clearpage
\tableofcontents
\clearpage

\section{Introduction}

The Walrasian notion of general competitive equilibrium, in its original form, describes a situation in which prices clear all markets and no forces of disequilibrium remain \citep{walras1900elements}. The question of how decentralized agents coordinate to reach such an outcome has challenged economists since Walras’s time (see \citealp{edgeworth1881mathematical}; and, just after the Second World War, \citealp{hayek1945use, arrow1959stability}), and in the 1980s inspired a rich literature on the dynamics of economic adjustment (e.g., \citealp{sonnenschein1982price,fisher1983disequilibrium,mas1986notes, artzner1986convergence}). The work of \cite{tatonnementLorenzoniWerning2026}
 represents the most recent contribution in this line of research.

As emphasized by \citet{stiglitz2018modern}, this strand of literature argues that economic theories should account for a multisectoral economy characterized by heterogeneous profit rates and wages across sectors, along with a continuous reallocation of firms and workers. Similar concerns are central to the theory of \emph{medium run}, which promotes a framework in which prices, wages, and stocks adjust at different speeds \citep{blanchard1997medium, solow2000toward}. In this respect, \emph{common-sense} microfoundations may require a careful reconsideration of how agents form expectations and make decisions in complex environments \citep{sargent1993bounded,de2019behavioural, moll2024trouble}.

A compelling theoretical description of such a reallocation process is provided by \citet{sonnenschein1982price}, who characterizes a capitalist economy as a sequence of short-run competitive equilibria. In the short run, “firms are `in place', and the returns to factors that are temporarily immobile are not necessarily equalized. The prices that clear market at each moment reflect only the relative scarcity of factors that are instantaneously variable''. Over time, all factors become free to move and, guided by simple behavioural rules, “flow toward that branch of production'' where returns are higher. This theoretical narrative has rich empirical support and strong predictive power (see, e.g., \citealp{simon1997empirically, foster2008reallocation}).

To \rev{study} these mechanisms, we develop a multi-sector competitive economy in which firms can reallocate across sectors in response to differential returns (Section \ref{sec:model}). At each instant, a short-run competitive equilibrium emerges, reflecting sectoral heterogeneity in technology, factor allocation, and consumer preferences (Section \ref{sec:shortRunEquilibrium}). Sectoral reallocation is governed by myopic profit-seeking behaviour, generating a continuous-time macroeconomic evolution that can be described by a partial differential equation (PDE) (Section \ref{sub:dynamics}).
{Theorem~\ref{th:riscritturaGradientFlow} gives this equation a variational
interpretation: as firms move towards more profitable sectors, their
distribution follows the direction of steepest ascent of log aggregate
consumption in the Wasserstein metric induced by quadratic reallocation costs.} This reformulation shows how macroeconomic behaviour emerges from micro-level decisions, revealing an invisible hand mechanism: the uncoordinated pursuit of profit steers the economy along a trajectory of monotonically increasing production and consumption, ultimately converging to an efficient (and Pareto-optimal) allocation (Theorems \ref{teo:convergence} and \ref{th:mu-e-ottimale}).

\red{Extensions involving intra-sectoral externalities, fixed reallocation costs, non-symmetric preferences, and immobile labour are studied in Appendices~\ref{sub:spillovers}, \ref{sub:fixedcost}, \ref{sec:nonsymmetric preferences}, and~\ref{sec:immobilework}, respectively.}

To motivate our framework empirically, Figure \ref{fig:motivatingEvidenceI} plots the cross-sectional distributions, at the sectoral level, of return on equity (ROE) and labour productivity for a balanced panel of EU firms from 2018 to 2023.\footnote{The sample draws on ORBIS data for more than 750,000 firms in 14 EU countries, classified into 680 four-digit NACE sectors.} \red{Appendix~\ref{app:quantitative} interprets} the convergence in ROE, along with the mild (if any) divergence in labour productivity, as the result of significant firm/capital reallocation across sectors, in contrast to much slower labour reallocation, and it traces the efficiency consequences of this imbalance. Although the ORBIS dataset does not allow us to track sectoral capital reallocation, extensive empirical evidence suggests that firms frequently reallocate resources across activities through changes in their product mix. For instance, \citet{bernard2010multiple} report that ``one-half of U.S. manufacturing firms alter their mix of five-digit SIC products every five years.''
\begin{figure}[!htbp]
\caption{Distribution of (normalized) return on equity (ROE) and labour productivity for 680 four-digit NACE sectors (from 01 to 82) in 2018 and 2023.}
\label{fig:motivatingEvidenceI}
\centering
\begin{subfigure}{0.45\textwidth}
\vspace{-0.5cm}
\includegraphics[width=1\textwidth]{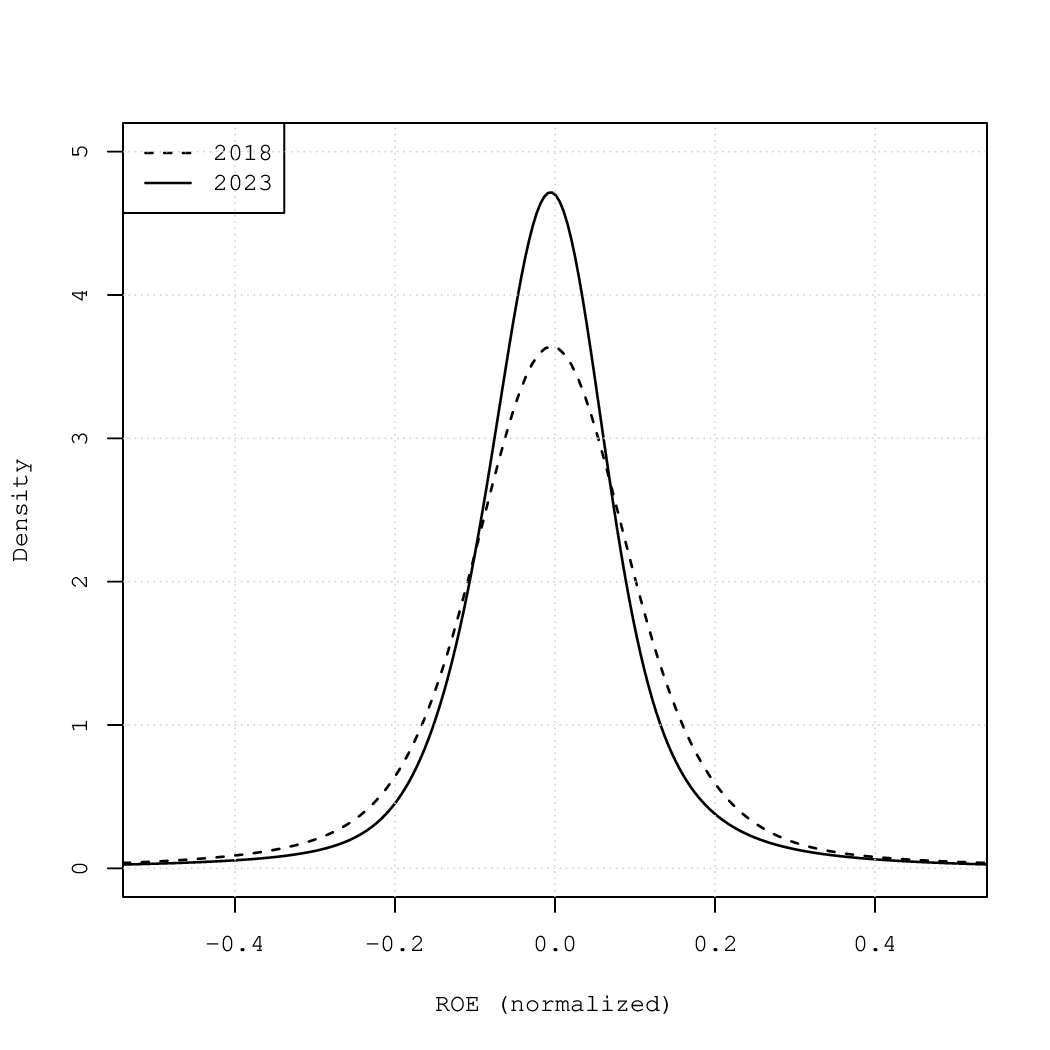}
\caption{The sectoral distribution of ROE (each value divided by the sample mean)}
\label{fig:productionGrowthVsROE}
\end{subfigure}
\hfill
\begin{subfigure}{0.45\textwidth}
\vspace{-0.5cm}
\includegraphics[width=1\textwidth]{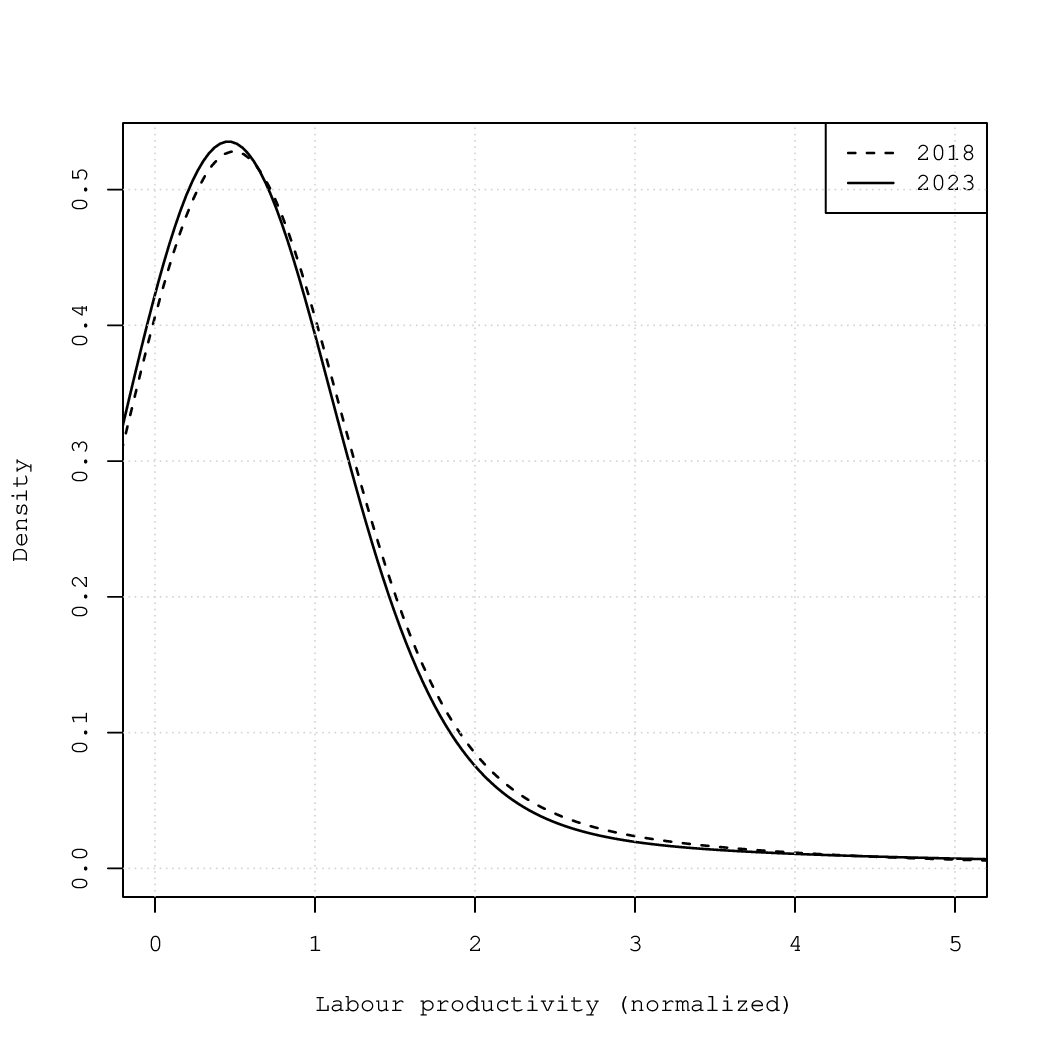}
\caption{The sectoral distribution of labour productivity (each value divided by the sample mean).}
\label{fig:labourProductivitycrossSectoralDistribution}
\end{subfigure}
\begin{flushleft}
{\small \textit{Source: our elaboration on ORBIS data.}}
\end{flushleft}
\end{figure}
Section~\ref{sec:backData} provides additional evidence on the convergence of sectoral profit rates (but not labour productivity), the degree of substitutability among goods, the magnitude of intra-sectoral externalities, and the limited role of fixed reallocation costs.


{
 Our methodological contribution is to identify the Wasserstein gradient-flow structure of firm reallocation under the myopic adjustment rule, building on the framework of \citet{artzner1986convergence}. Firms' profits coincide with the first variation of the potential $\log X(\mu)$, while quadratic reallocation costs determine the associated transport metric. This structure provides a variational characterization of decentralized adjustment through the JKO scheme and supports our analysis of monotonicity, convergence, and long-run efficiency.
}
\revm{\revb{We relate our results to the literature at the relevant points of the paper: firm reallocation, the medium run, behavioural decision rules and expectations in Section~\ref{sub:dynamics}; the positive theory of general equilibrium, potential games, evolutionary dynamics and learning in Section~\ref{sec:LongRUnEquilibrium}.}}

The paper is organized as follows. Section \ref{sec:model} introduces the model under the assumption of perfect labour mobility. Section \ref{sec:emergentmacrpproerties} presents its gradient flow representation and analyses the properties of the long-run competitive equilibrium. \rev{Section \ref{sec:planner} compares decentralized and optimal transitions.} Section \ref{sec:conclusions} concludes and outlines directions for future research. \red{Appendices~\ref{app:proofs} and~\ref{app:control-representation} collect the proofs and an exact control representation of the decentralized dynamics. Appendices~\ref{sub:spillovers}--\ref{sec:immobilework} present the extensions, and Appendix~\ref{app:quantitative} contains the empirical analysis and numerical experiments.}

\section{The model\label{sec:model}}
This section presents the structure of the model. We begin by establishing the \emph{short-run competitive equilibrium} through the analysis of consumer and firm decisions, assuming perfect labour mobility across sectors and instantaneous adjustment of prices and wages (Section \ref{sec:shortRunEquilibrium}). We then introduce firms’ dynamics, namely their reallocation across sectors in response to profit differentials (Section \ref{sub:dynamics}).

\subsection{The short-run competitive equilibrium}
\label{sec:shortRunEquilibrium}

We consider a competitive economy with a unit mass of identical
consumers, each supplying one unit of labour inelastically and
earning income from wages and firm ownership. Aggregate labour
supply is $L=1$, and individual and aggregate income coincide:
$y=Y$. Goods are indexed by $i\in[0,n]$. Time indices are omitted
in this section.

Consumers have CES preferences over goods, with elasticity of
substitution $\sigma$ \citep{dixit1977monopolistic}. Instantaneous utility is $u(X)$,
with $u$ increasing in $X$. The consumption aggregate and its price index are
\begin{equation}
\label{eq:preferences}
X=
\left(
    \int_0^n x(i)^{\frac{\sigma-1}{\sigma}}\,di
\right)^{\frac{\sigma}{\sigma-1}},
\qquad
P=
\left(
    \int_0^n p(i)^{1-\sigma}\,di
\right)^{\frac{1}{1-\sigma}},
\end{equation}
where $\sigma>0$, $\sigma\ne1$, and $x(i)$ and $p(i)$ denote
consumption and the price of good $i$. Consumers maximize utility
subject to $\int_0^n p(i)x(i)\,di\leq y$.

A unit mass of firms is distributed across sectors with density
$\mu$, where $\int_0^n\mu(i)\,di=1$.\footnote{
We use $\mu$ for both the density and its associated measure
(with total mass equal to $1$), writing $d\mu(i)=\mu(i)\,di$.
}
Each firm in sector $i$ employs $l(i)$ units of labour and produces
\begin{equation}
\label{eq:productionFunctionGoodi}
q(i)=A(i)l(i)^\beta,
\qquad 0<\beta<1,
\end{equation}
where $A(i)>0$ is sectoral productivity. Labour is perfectly
mobile, so all firms pay the same wage $w$. Firms take prices
as given and maximize profits $\pi(i)=p(i)q(i)-wl(i)$, which
can be interpreted as the return to one fixed unit of capital
per firm.

Goods and labour markets clear instantaneously for a given
distribution $\mu$, while firms reallocate across sectors over
time. Goods and labour market clearing require, respectively,
\begin{equation}
\label{eq:marketClearingEachSector}
x(i)=\mu(i)q(i),
\qquad
\int_0^n\mu(i)l(i)\,di=1.
\end{equation}

\begin{Proposition}[The short-run competitive equilibrium]
\label{pr:shor-run-equilibrium}
Let $A$ and $\mu$ be continuous and strictly positive, with
$\int_0^n\mu(i)\,di=1$, and normalize aggregate nominal income
to $Y=1$.\footnote{
This normalization follows \citet{grossman_helpman_1991};
see their equation (3.8), p.~48.
}
The equilibrium wage is $w=\beta$, and aggregate profits are
$\Pi:=\int_0^n\pi(i)\mu(i)\,di=1-\beta$.
Profits per firm are
\begin{equation}
\label{eq:profitti_dopo}
\pi(i)=
\frac{
    (1-\beta)
    A(i)^{\frac{\sigma-1}{\sigma(1-\beta)+\beta}}
    \mu(i)^{-\frac{1}{\sigma(1-\beta)+\beta}}
}{
    \displaystyle
    \int_0^n
    \left[A(j)\mu(j)^{1-\beta}\right]^{
        \frac{\sigma-1}{\sigma(1-\beta)+\beta}
    }\,dj
}.
\end{equation}
Aggregate consumption satisfies
\begin{equation}
\label{eq:X-P}
X(\mu)=\frac{1}{P}
=
\left(
    \int_0^n
    \left[A(i)\mu(i)^{1-\beta}\right]^{
        \frac{\sigma-1}{\sigma(1-\beta)+\beta}
    }\,di
\right)^{\frac{\sigma(1-\beta)+\beta}{\sigma-1}},
\end{equation}
and goods prices are
\begin{equation}
\label{eq:value-of-pi}
p(i)=
P^{\frac{(\sigma-1)(1-\beta)}{\sigma(1-\beta)+\beta}}
\left[A(i)\mu(i)^{1-\beta}\right]^{
    -\frac{1}{\sigma(1-\beta)+\beta}
}.
\end{equation}
Sectoral revenue and employment satisfy
\begin{equation}
Y(i):=p(i)q(i)\mu(i)
=
\frac{\mu(i)\pi(i)}{1-\beta}
=
l(i)\mu(i)=:L(i).
\end{equation}
\end{Proposition}

\begin{proof}
See Appendix~\ref{app:proofs}.
\end{proof}
As $\sigma>0$ (concavity of preferences) and $\beta \in (0,1)$ (diminishing returns to labour),
profits per firm $\pi(i)$ are positive and decreasing in the mass of firms in the sector $\mu(i)$. We note that decreasing returns are a necessary condition for positive profits: as $\beta$ approaches 1, profits approach zero even though goods remain differentiated.

Real wages and aggregate profits are $\beta X(\mu)$ and
$(1-\beta)X(\mu)$, respectively, and instantaneous welfare is
\begin{equation}
\label{eq:U-P}
U=u(X(\mu))=u\left(\frac{1}{P}\right).
\end{equation}
For each fixed firm distribution, competitive labour allocation
maximizes aggregate consumption.

\subsection{The dynamics of firms across sectors}
\label{sub:dynamics}

The firms’ search for higher profit rates, i.e. the reallocation of capital across sectors, is the key source of the dynamics.

\revm{\revb{Since goods and labour markets clear instantaneously while firms reallocate over time, the model introduces} a new dimension to the medium-run literature, which has traditionally focused on the rigidity of prices and wages, by highlighting the importance of considering the heterogeneity in the timescale associated with the adjustment of different productive factors \citep{blanchard1997medium, solow2000toward}. The literature on firm reallocation, including works by \citet{dixit1977monopolistic}, \cite{MaksimovicVojislav2001}, \citet{melitz2003impact}, \citet{foster2008reallocation}, \citet{bertoletti2017monopolistic}, \citet{parenti2017toward}, and \citet{osotimehin2019aggregate}, examines how firms adjust across sectors in response to productivity and demand shocks, often under monopolistic competition. \revb{Here firms are instead price takers, and we} provide insights into how sectoral heterogeneity in technology shapes both short- and long-run competitive equilibria.}

To avoid complications arising from boundary effects in the firm dynamics, assume that the set of goods is distributed along a circle of length $n$, i.e.:
\begin{equation}
\label{eq:defS}
S := \left \{ x \in \mathbb{R}^2 \, : \, \norm{x} = \frac{n}{2\pi} \right \}.
\end{equation}
The set $S$ corresponds to the interval $[0,n]$ with identified endpoints. For every $i,j\in[0,n]$ their distance, denoted by $\abs{i-j}$ is the arc-length distance on $S$. Functions defined on $S$ can be seen as functions of one real variable that are periodic with period $n$.

Following \cite{artzner1986convergence}, firms move across sectors in response to differences in current profit rates. In particular, a firm producing a good $i$ moves where profit rates are higher, with velocity $v$, but this reallocation is costly. The nature of these reallocation costs can vary, and their specific form can have important implications for general equilibrium.\footnote{\citet[p. 385]{eisfeldt2006capital} refer to “capital illiquidity'' as the main source of reallocation cost, and they assume that such cost decreases the stock of capital.} In our framework, these costs are implicit, not reflecting a specific waste or loss of capital, such as psychological costs borne by the entrepreneur due to changing the nature of her business or informational and contractual
specificities of the sector where capital is allocated. Instead, we retain the quadratic shape of capital reallocation cost used in the literature \citep{eisfeldt2006capital}, i.e. the cost function is assumed to be $c(v)=\frac{\rev{\chi}}{2}v^2$\rev{, where $\chi>0$ measures reallocation frictions}. The firm in sector $i$ therefore solves the following maximization problem:
\begin{equation}
\max_{v} \left [ \left( \frac{\partial \pi(t,i)}{\partial i} \right) v - \frac{\rev{\chi}}{2}v^2 \right ].
\label{eq:ruleGoverningRiallocationFirms}
\end{equation}
Problem \eqref{eq:ruleGoverningRiallocationFirms} is  the continuous-time (or “infinitesimal”) analogue of the solution of a problem in discrete time in which a firm in sector $i$ at time $t$ chooses a new sector $j$ to maximize:
\begin{equation} \label{eq:slide_discrete_max}
\max_{j \in S} \left [  \pi(t, j) - \frac{\gre \chi }{2} |i-j|^2 \right ],
\end{equation}
where $\pi(t,j)$ denotes the current profits in the sector $j$, and $\frac{1}{2}|i-j|^2$ is the quadratic cost of reallocating production from $i$ to $j$.
Problem \eqref{eq:slide_discrete_max} highlights the two key assumptions underlying Problem \eqref{eq:ruleGoverningRiallocationFirms}: 1) current profits are used as a proxy for future ones, that is, firms have \textit{static expectations}; and 2) only next-period profits are considered, which can be interpreted as a very high intertemporal discount rate for firms relative to their rate of reallocation.\footnote{\citet{krugman1993number} makes a similar assumption in a spatial context. As early as \citet{arrow1986rationality}, it was noted that requiring agents to solve full intertemporal problems with rational expectations places unrealistic informational and computational demands on them. More broadly, several authors argue that excessive sophistication in expectation formation is unrealistic in complex environments \citep{sargent1993bounded, de2019behavioural, blanchard2025convergence}. The recent heterogeneous-agent macro literature emphasizes possible criteria for bounded rationality in complex environments, such as the computational tractability of the model, the use of decision rules with some empirical support, and the requirement that agents' endogenous expectations do not systematically differ from the actual dynamics of the model \citep{moll2024trouble}. Our behavioral rule in \eqref{eq:ruleGoverningRiallocationFirms} aligns with these considerations.}

\revm{\revb{Our rule belongs to the tradition} advocating simple, empirically grounded decision rules, as discussed by \citet{nelson1985evolutionary}, \citet{simon1997empirically}, \citet{kirman2016ants}, \citet{stiglitz2018modern}, \revb{which} views the economy as a complex evolving system \citep{EconomyAsAComplexEvolvingSystemII1997} where a multiplicity of agents interact on the basis of quite simple, empirically based, rules. The same spirit, driven by the need to describe a positive theory based on “plausible'' behavioural rules for economic agents, underlies the behavioural approach in macroeconomic theory advocated in \citet{sargent1993bounded}, \citet{de2019behavioural} and \citet{moll2024trouble}. In these models, agents interact locally and “collective outcomes are `emergent properties' which cannot be imputed to the intention of any single agent'' \citep{Dosi_Roventini_2025}.}

\revm{\revb{Static expectations also connect the model} to the literature on \revb{what} rational expectations require in a decentralized economy\revb{: agents} must correctly forecast future prices, possibly inferring from them the information of others \revb{\citep{radner1972existence, radner1979rational, allen1981generic, jordan1982overview}}. When agents must learn the relevant structure, learning may \revb{stop before the truth is discovered \citep{mclennan1984price} or} converge only under stability conditions \revb{\citep{bray1982learning, marcet1989convergence}}. \revb{Our rule is closer in spirit to the internal rationality of \citet{adam2011internal}, with beliefs only about the payoff-relevant sectoral profit rates.}}

The solution to Problem (\ref{eq:ruleGoverningRiallocationFirms}) is
\[
v = \rev{\frac{1}{\chi}}\frac{\partial \pi(t,i)}{\partial i},
\]
which implies that the optimal velocity for a firm in a sector $i$ equals the marginal gain in profits \rev{divided by} the \rev{cost coefficient $\chi$}.


{
To translate the micro-level behaviour into an equation for the distribution of firms we use the continuity equation associated with the reallocation velocity field $v(t,i)$:
}
\begin{equation}
\label{eq:optimal-v}
\partial_t \mu(t,i) = -\partial_i \left( \mu(t,i) v(t,i) \right).
\end{equation}

{
This equation expresses conservation of firm mass through sectoral inflows and outflows. Substituting the firms' optimal velocity gives
}
\begin{equation}
\label{eq:continuity}
\partial_t \mu(t,i) = -\rev{\frac{1}{\chi}}\partial_i \left( \mu(t,i) \partial_i \pi(t,i) \right).
\end{equation}

Taking the expression for profits from Equation \eqref{eq:profitti_dopo}, given the initial distribution of firms $\mu(0, \cdot)$, and considering the boundary condition $\mu(t,0) = \mu(t,n)$ for every $t \geq 0$, yields
\begin{equation}
\label{eq:stato-di-nuovo}
\left \{
\begin{array}{l}
\displaystyle \partial_t \mu(t,i) = -\frac{1-\beta}{\rev{\chi}} \left[ \frac{\partial_i \left( \mu(t,i) \partial_i \left( A(i)^{\frac{\sigma -1}{\sigma(1-\beta) + \beta}} \mu(t,i)^{-\frac{1}{\sigma(1-\beta) + \beta}} \right) \right)}{\displaystyle  \int_0^n \left[ A(j) \mu(t,j)^{1-\beta} \right]^{\frac{\sigma -1}{\sigma(1-\beta) + \beta}} \, dj } \right],\,\,\, (t,i) \in \RR^+ \times S\\[7pt]
\displaystyle \mu(t,0) = \mu(t,n)\\[7pt]
\displaystyle \mu(0,i) = \mu_0(i).
\end{array}
\right .
\end{equation}

Theorem \ref{th:esistenzaeunicita} shows that the dynamics of $\mu$ described by partial differential equation (PDE) \eqref{eq:stato-di-nuovo} is well posed.
\begin{theorem}[Existence, uniqueness, and boundedness of the firm distribution]
\label{th:esistenzaeunicita}
Let $\mu_0, A\colon S \to \mathbb{R}$ be strictly positive and in $C^2(S)$\footnote{$S$ is defined in (\ref{eq:defS}), $C^2(S)$ are real continuous functions $v$ defined on $S$ (with $v(0)=v(n)$, with continuous first and second derivatives).}. Then PDE \eqref{eq:stato-di-nuovo} has a unique solution $\mu(t, \cdot)\in C([0, +\infty), C^2(S))$ and there exists $M>0$ such that, for all $t\geq 0$
\[
\frac{1}{M} \leq \mu(t,i) \leq M, \qquad \forall i\in S.
\]
\end{theorem}
\begin{proof}
See Appendix \ref{app:proofs}.
\end{proof}

\section{Emergent macro properties}
\label{sec:emergentmacrpproerties}

This section provides the main result of the paper using a gradient flow representation of the aggregate dynamics. We show that the dynamics induced by competition
among firms generate a monotonic improvement in welfare and converge in the long run to the economy’s social optimum.

\subsection{Gradient flow representation of the competitive equilibrium dynamics}
\label{sec:emergentMacroProperties}

We start by recasting the dynamics in PDE \eqref{eq:stato-di-nuovo} as a \textit{gradient flow} in the \textit{Wasserstein space} $\mathcal{W}_2(S)$\footnote{\revm{Observe that, although the use of optimal transport methods in economics is now well established (see, e.g., \citealp{galichon2018optimal}), to the best of our knowledge only two recent contributions have employed gradient flows in Wasserstein spaces to study economic phenomena. \citet{cohen2025gradient} use a related mathematical framework in Econophysics, while \citet{fiaschi2025spatial} apply this approach to workers' migration. By contrast, the present paper employs this formalism to model firms' dynamics.}}, which enables a direct understanding of how firm-level behaviour influences aggregate variables. In particular, Section \ref{sec:WassersteinSpace} recalls the \textit{Wasserstein distance} of index 2 in $S$, providing a suitable metric structure to analyse the evolution of $\mu(t)$, highlighting how the reallocation costs of firms are related to the structure of $\mathcal{W}_2(S)$. Section \ref{sec:gradientFlowRepresentation} reformulates the PDE \eqref{eq:stato-di-nuovo} as a gradient flow within $\mathcal{W}_2(S)$, that is, we identify a functional $\mathcal{F}(\mu)$ such that the evolution of $\mu(t)$ follows the direction of steepest \rev{ascent} of $\mathcal{F}$ in $\mathcal{W}_2(S)$.
This formulation corresponds to a sequence of period-by-period optimizations in the continuous-time limit, as formalized by the Jordan-Kinderlehrer-Otto (JKO) scheme \citep{jordan1998variational}. Then, in Section \ref{sec:implicationGradientFlowRepresentation}, we examine the key emergent properties of dynamics. Finally, Section \ref{sec:PDEGradientFlow} provides a brief overview of reformulating PDEs as gradient flows and of the JKO scheme.

\subsubsection{The Wasserstein space}
\label{sec:WassersteinSpace}

The Wasserstein space with index 2 on $S$, denoted by $\mathcal{W}_2(S)$, is defined as the set of probability measures $\mu \in \mathcal{P}(S)$
endowed with the distance\footnote{In general, it is required that the measures have a finite second moment; however, since the space $S$ is bounded in our analysis, this requirement does not restrict the set $\mathcal{P}(S)$. The Wasserstein-2 distance can also be defined using the \textit{Kantorovich formulation}, which considers all possible transport plans between the two measures. The distance is given by $d_{W_2}(\mu, \nu) = \left( \inf_{\pi \in \Pi(\mu, \nu)} \int_{S \times S} |i-j|^2 \, d\pi(i,j) \right)^{1/2}$, where $\Pi(\mu, \nu)$ is the set of all joint probability measures on $S \times S$ whose first and second marginals are $\mu$ and $\nu$, respectively.}
\begin{equation}
d_{W_2}(\mu, \nu) = \left( \inf_{T \, : \, T_\# \mu = \nu} \int_S \abs{i -T(i)}^2 \, d\mu(i) \right)^{1/2},
\label{def:wasserstein}
\end{equation}
where the infimum is taken over all \textit{transport maps} $T : S \to S$ such that $T_\# \mu = \nu$, meaning that the push forward of $\mu$ by $T$ equals $\nu$.\footnote{The push forward measure $T_{\#}\mu$ of a measure $\mu$ on $S$ by a measurable map $T: S \to S$ is a measure on $S$ defined for any measurable set $A \subseteq S$ by $(T_\# \mu)(A) = \mu(T^{-1}(A))$.} An important relationship exists between the Wasserstein-2 distance $d_{W_2}(\mu, \nu)$ and the quadratic cost of reallocating firms across sectors. In this respect, consider the solution of Problem \eqref{eq:slide_discrete_max}, i.e.
\begin{equation}
\label{eq:totalcost-discrete}
\tilde T(i) = \arg\max_{j \in S} \left [ \pi(t,j) - \frac{1}{2} |i-j|^2 \right ].
\end{equation}
$\tilde T$ is the map of how the firms' reallocation drives the firms' distribution. The total reallocation cost associated with this reallocation, $RC$, is equal to:
\[
RC = \frac{\gre \chi}{2} \int_S \abs{i-\tilde T(i)}^2 \, d\mu(i),
\]
which, apart from $1/2$, has exactly the same form as the cost to be minimized in the definition of $d_{W_2}(\mu, \nu)$ in \eqref{def:wasserstein}; hence
\[
RC_m(\mu, \nu) = \frac{\gre \chi}{2} \left(d_{W_2}(\mu, \nu) \right)^2
\]
is exactly the \textit{minimum reallocation cost} in the model of taking the firm distribution $\nu$ to $\mu$.

\subsubsection{The gradient flow representation}
\label{sec:gradientFlowRepresentation}

We now use the notions introduced above to rewrite the dynamics of $\mu$ in PDE \eqref{eq:stato-di-nuovo}. In this new formalism, the distribution of firms across sectors at a given time $t$, $\mu(t,.)$, is regarded as a point in the space $\mathcal{W}_2(S)$, and its evolution can be described as a gradient flow, that is, a dynamics in which, at each instant, the system moves in the direction of the steepest increase of a suitable functional.\footnote{In the mathematical literature (see, e.g., \citealp{santambrogio2015optimal}), it is customary to consider the negative of the functional in Equation \eqref{eq:functionalGradientFlow}, which results in a minus sign in the gradient flow expression \eqref{eq:gradientflow}. This is purely a notational convention; we adopt the positive form to remain consistent with the economic interpretation.}
The gradient flow framework also allows the dynamics to be expressed as a sequence of maximization problems over time, namely, via the JKO scheme. Theorem \ref{th:riscritturaGradientFlow} expresses rigorously this formalism.

\begin{theorem}[The gradient flow representation and JKO scheme]
\label{th:riscritturaGradientFlow}
Let $\mu_0, A\colon S \to \mathbb{R}$ be strictly positive and in $C^2(S)$; then the unique solution of PDE \eqref{eq:stato-di-nuovo} is the unique gradient flow in $\mathcal{W}_2(S)$
of $\frac{1}{\chi}\mathcal{F}$, where
\begin{equation}
\label{eq:functionalGradientFlow}
    \mathcal{F}(\mu)
    =  \left\{ \frac{(1-\beta)[\sigma(1-\beta) + \beta]}{[\sigma(1-\beta) + \beta]-1} \right\}\log\left(\int_{0}^{n} \left[A(i)^{\frac{1}{1-\beta} } \mu(t,i)\right]^{1-\frac{1}{\sigma(1-\beta) + \beta}} \, di\right)
    = \log X (\mu(t)),
\end{equation}
i.e.
\begin{equation}
\label{eq:gradientflow}
\frac{d\mu}{dt} =\rev{\frac{1}{\chi}}\nabla_{W_{2}} \mathcal{F}(\mu(t)) =  \dfrac{\nabla_{W_2} {X}(\mu(t))}{\rev{\chi} X(\mu(t))}, \qquad \mu(0,\cdot) = \mu_0.
\end{equation}
The trajectory of $\mu(t)$ in Equation  \eqref{eq:gradientflow} can be obtained as the limit for $\Delta t \to 0$ of the JKO scheme, where
\begin{multline}
\mu(t+\Delta t) = \arg \max_{\mu} \left\{\mathcal{F}(\mu) - \left(\frac{\rev{\chi}}{2\Delta t}\right) d_{W_2}^2(\mu, \mu(t)) \right\} \\
       = \arg \max_{\mu} \left\{\log X(\mu) -\left(\frac{1}{\Delta t}\right) RC_m(\mu, \mu(t)) \right\}.
    \label{eq:JKOScheme}
\end{multline}
\end{theorem}
\begin{proof}
See Appendix \ref{app:proofs}.
\end{proof}

\revm{The continuous structure of the goods space allows us to reformulate the problem as a gradient flow in Wasserstein space. Instead, the finite-goods counterpart would not be technically simpler. The corresponding optimization problems generally lack interior solutions and continuous adjustment dynamics, while introducing indivisibilities akin to “small menu costs” that make convergence less smooth. Moreover, they tend to generate multiple equilibria arising from the discreteness of the framework itself rather than from underlying economic mechanisms; such effects vanish as the number of goods becomes large.
To the best of our knowledge, the closest contributions adopting a related framework in a finite setting are recent studies in spatial economics, where sectors are replaced by locations (e.g., \citealp{CaliendoDvorkinParro2019, KleinmanLiuRedding2023}).}

\subsubsection{The emergent aggregate efficiency}
\label{sec:implicationGradientFlowRepresentation}


{
The identification of $\mathcal F(\mu)=\log X(\mu)$ as the potential driving decentralized reallocation yields a direct efficiency implication. By Theorem~\ref{th:riscritturaGradientFlow}, $\mathcal F$ is non-decreasing along the dynamics; hence aggregate consumption is non-decreasing, as stated in Corollary~\ref{cor:crescitaOttimal}.
}
\begin{Corollary}[Non-decreasing aggregate consumption]
\label{cor:crescitaOttimal}
Assume the hypotheses of Theorem \ref{th:riscritturaGradientFlow} hold; the functional $\mathcal{F}$ is increasing along the trajectory of the solution of PDE \eqref{eq:stato-di-nuovo}.
\end{Corollary}
\begin{proof}
See Appendix \ref{app:proofs}.
\end{proof}

In our model, $\mathcal{F}$ corresponds to the natural logarithm of aggregate consumption, but also reflects aggregate production, profits, and wages, the latter being a constant share of total production. Hence, based on Equation \eqref{eq:functionalGradientFlow}, the economy evolves along a path which, when measured in $\mathcal{W}_2(S)$, locally maximizes the instantaneous rate of change of aggregate consumption (and, by extension, production, profits, and wages) at each point in time.
This result can be intuitively understood by noting that each firm aims to increase its net profits, seeking the highest possible relative gain given its current profit level and the associated reallocation costs.
The JKO scheme represents decentralized reallocation as a sequence of one-step optimization problems. It therefore does not establish
intertemporal efficiency once the cost of reallocation is included.  Whether the resulting
trajectory also maximizes intertemporal welfare net of reallocation
costs is examined in Section~\ref{sec:planner} (\red{see also Proposition~\ref{pr:control-representation}}).

Since a key step in the gradient flow approach is the identification of the functional $\mathcal{F}$, a natural question arises: to what extent is $\mathcal{F}$ constrained by the specific functional forms adopted in this paper? In our case, the identification of $\mathcal{F}$ from the original PDE \eqref{eq:stato-di-nuovo} is obtained starting from the results of \citet{iacobelli2019weighted}.
We argue that the structure of $\mathcal{F}$ is sufficiently robust to allow for its computation under various extensions of the model, as illustrated \red{in Appendices~\ref{sub:spillovers}, \ref{sec:nonsymmetric preferences}, and~\ref{sec:immobilework}}.

 \subsubsection{From gradient flow to JKO scheme}
 \label{sec:PDEGradientFlow}

 In a more general perspective, there exist several classes of PDEs that can be reformulated as gradient flows. For some of these equations, such as the \textit{heat equation} and the \textit{p-Laplacian equation} (\citealp{ambrosio2021} Chapter 13), one can find a gradient flow formulation in $L^2(\Omega)$ spaces. For other classes of equations (e.g., the Fokker-Planck equation, porous medium equation, advection equation, aggregation equation, Keller-Segel system, and crowd motion equations, see \citealp[Chapter 8]{santambrogio2017}, and \citealp[Chapter 18]{ambrosio2021}), including the PDE studied here, which belongs to the class of (ultra) \textit{fast diffusion equations}, the natural formulation is in $\mathcal{W}_2(S)$.
 More specifically, the PDEs that can be formulated as gradient flows in $\mathcal{W}_2(S)$ are typically expressed as continuity equations of the form (see \citealp[Chapter 5]{santambrogio2017}, and also Example 11.1.2, page 281, in \citealp{ambrosio2005gradient}):
 \[
 \frac{\partial \mu}{\partial t} = - \nabla \cdot (\mu v),
 \]
 where $\mu$ is the probability density, and the velocity field $v$ is given by:
 \[
 v = \rev{\frac{1}{\chi}}\nabla \left( \frac{\delta \mathcal F}{\delta \mu} \right),
 \]
 \rev{using our ascent convention (the standard minimizing formulation uses $-\mathcal F$). Here} $\mathcal F$ \rev{is} a functional on probability measures, and $\delta \mathcal F / \delta \mu$ its first variation.

 As regards the possibility of using JKO scheme, it is known to converge for equations formulated as gradient flows in $\mathcal{W}_2(S)$ driven by functionals that are bounded from below, lower semicontinuous, and exhibit some form of convexity (see \citealp[Chapter 11.1.3]{ambrosio2005gradient}). Standard examples include the heat equation, linear Fokker-Planck equations, and the porous medium equation with $m>1$.
 For the \emph{fast diffusion equation} ($0<m<1$), a JKO-scheme-based construction of weak solutions via non-local approximations has recently appeared in \cite{CarrilloEspositoSkrzeczkowskiWu2024}, while existence results via the JKO scheme for weighted ultra-fast diffusion equations ($m<0$) are obtained in \cite{iacobelli2019weighted}, employing weighted $L^1$ estimates and careful regularizations. For aggregation-diffusion PDEs with regular interaction potentials, JKO convergence is explicitly established in \cite{HuangMaininiVazquezVolzone2022}.
 On the other hand, there exist several classes of equations, such as certain Keller-Segel systems or models with hard state constraints, that fail to satisfy standard JKO convergence conditions due to the loss of geodesic convexity or insufficient regularity of the driving functional. In these cases, the gradient flow solution is typically defined as a curve of maximal slope or via regularizations (see, e.g.,\citealp[Chapter 11]{ambrosio2005gradient}).

\subsection{The long-run competitive equilibrium}
\label{sec:LongRUnEquilibrium}

In this section, we characterize the long-run competitive equilibrium, an equilibrium in which all markets clear, consumers and firms maximize utility and profits, and neither firms nor workers have an incentive to reallocate. Theorem \ref{teo:convergence} establishes the existence, uniqueness, and global convergence of the economy toward this equilibrium. As shown in its proof, identifying a long-run competitive equilibrium amounts to finding a stationary distribution of the dynamics in PDE (\ref{eq:stato-di-nuovo}). Finally, Theorem \ref{th:mu-e-ottimale} demonstrates that the long-run competitive equilibrium is efficient, in the sense that aggregate consumption, production, profits, and wages are maximized.

\begin{theorem}[The long-run competitive equilibrium]
\label{teo:convergence}
Assume that the hypotheses of Theorem \ref{th:esistenzaeunicita} hold. Then the economy admits a unique long-run competitive equilibrium, at which the firms' distribution is given by
\begin{equation}\label{eq:steadystatemu}
\mu^{EQ}(i) = \frac{A(i)^{\sigma -1}}{\displaystyle\int_0^n A(j)^{\sigma - 1}\,dj} \quad \forall i \in S,
\end{equation}

\noindent profit rates by $\pi^{EQ}(i) = 1-\beta$, and production, employment and consumption, respectively, by
\[
Y^{EQ}(i) = L^{EQ}(i)  =\frac{A(i)^{\sigma - 1}}{\displaystyle\int_0^n A(j)^{\sigma -  1} dj} \quad \forall i \in S \text{ and } \qquad X^{EQ} = \left( \int_0^n A(j)^{\sigma-1} \, dj \right)^{\frac{1}{\sigma-1}}.
\]
\noindent Finally, for any initial condition $\mu_0$, the unique solution of PDE \eqref{eq:stato-di-nuovo} converges to $\mu^{EQ}(\cdot)$ exponentially fast in $L^2(S)$, i.e.
\begin{equation}
\norm{\mu(t,\cdot)-\mu^{EQ}(\cdot)}_{L^2(S)} \leq C_1 e^{-C_2 t}
\end{equation}
for some $C_1,C_2 > 0$.
\end{theorem}
\begin{proof}
See Appendix \ref{app:proofs}.
\end{proof}

In the long-run competitive equilibrium, profit rates are equalized across sectors at the level $1-\beta$, and the mass of firms is allocated across sectors according to technological progress and consumer preferences. Specifically, more productive sectors attract a larger share of firms only if the elasticity of substitution among goods exceeds one (Equation \eqref{eq:steadystatemu}), that is, when the substitution effect dominates the income effect.
In this equilibrium, the distributions of production $Y$ and employment $L$ mirror that of the firms’ allocation; that is, labour productivity is equal to one in each sector. This latter finding is not surprising, given the assumptions of symmetric preferences over goods and perfect labour mobility.

Theorem \ref{th:mu-e-ottimale} states that, in the long-run competitive equilibrium, consumption (as well as production, wages, and profits) is maximized, implying that the long-run competitive equilibrium is efficient.
\begin{theorem}[The efficiency of long-run competitive equilibrium]
\label{th:mu-e-ottimale}
Assume the hypotheses of Theorem \ref{th:esistenzaeunicita} hold. Then \eqref{eq:steadystatemu} is the unique distribution that maximizes $X$ in the set of non-negative measurable functions $\mu(i)$ on $S$ such that $\int \mu(i) di = 1$, i.e. $X^{EQ}$ is the maximum attainable consumption.
\end{theorem}
\begin{proof}
See Appendix \ref{app:proofs}.
\end{proof}

To summarize, Theorems \ref{teo:convergence} and \ref{th:mu-e-ottimale} together provide a comprehensive characterization of the dynamics and their efficiency properties. Starting from any initial distribution, the economy converges to a unique long-run competitive equilibrium distribution $\mu^{EQ}$, which arises from decentralized, myopic, profit-seeking behaviour by firms facing reallocation costs.
This equilibrium distribution coincides with the one that maximizes aggregate consumption $X$, and hence welfare $U(X)$, production, profits, and wages. \gre{The decentralized dynamics also admit an exact intertemporal
control representation, given in
Proposition~\ref{pr:control-representation}
in Appendix~\ref{app:control-representation}.}

\revm{\revb{
		 These results connect to the} ``positive'' theory of general equilibrium, as advanced by \citet{sonnenschein1981price}, \citet{sonnenschein1982price}, and \citet{artzner1986convergence}, \revb{which} explores economies where firms operate under perfect competition and adjust across sectors based on myopic profit-maximizing rules. This literature grapples with the stability of equilibria, particularly in light of the “impossibility'' results of \citet{Sonnenschein1972}. Our gradient flow representation provides a formal invisible hand result \emph{à la} Smith, where beneficial social and economic outcomes emerge from the self-interested actions of individuals: we demonstrate that the sequence of short-run competitive equilibria \revb{admits an instantaneous variational representation with potential $\log X$, which is distinct from intertemporal welfare maximization.} This allows \revb{us} to establish, \revb{for the preferences and technology maintained here,} the existence, global stability, and the efficiency (maximum production and consumption) of the long-run competitive equilibrium even in the absence of coordinated actions \citep{artzner1986convergence, howitt2006coordination} and sophisticated formation of expectations \citep{simon1997empirically}.
		 \revb{
		 	For the behavioural literature discussed in Section~\ref{sub:dynamics}, the framework provides an explicit link:} the aggregation of myopic, profit-seeking behaviour is \revb{represented by the gradient ascent of $\log X$, while the intertemporal planner studied in Section~\ref{sec:planner} generally chooses a different transition.} This demonstrates how decentralized, “plausibly'' behavioural rules can lead to an efficient allocation, providing a formal bridge between behavioural micro-foundations and aggregate efficiency.}

\revm{\revb{The economy is also a potential game in the sense of \citet{monderer1996potential} and \citet{sandholm2001potential}, in the continuous-strategy version of \citet{cheung2018nonatomic}.
Economic applications of potential games have so far concerned a few specific settings, such as Cournot oligopoly \citep{monderer1996potential, cheung2018nonatomic}, congestion \citep{beckmann1956studies, sandholm2002evolutionary} and a variant of the core--periphery model \citep{oyama2009agglomeration}. To our knowledge, the property had not been established for the microfounded multi-sector CES general-equilibrium model; it is not imposed, but emerges from competitive pricing.
}}

\revm{Our paper shares with the \textit{evolutionary literature} initiated by \citet{kandori1993learning} and \citet{young1993evolution} the description of the economy as a large population of interacting agents following simple adjustment rules, and the focus on the evolution of the population state over time. The central question of that literature, however, is different. Its games, typically coordination games, have multiple strict equilibria; the deterministic adjustment process reaches different equilibria depending on initial conditions, and rare mutations are introduced to select among them, possibly in favour of an inefficient (risk-dominant) one \citep{kandori1995evolution}. 
}
\revm{Finally, our paper is also related to the literature on how \revb{sophisticated players must be to learn} to play a Nash equilibrium \citep{kalai1993rational}. \revb{Unlike the Bayesian players of \citet{kalai1993rational}, our firms} form no beliefs about other agents' strategies; they react to current profit differentials. In our model, therefore, convergence follows from the structure of the competitive economy rather than from the sophistication of agents\revb{, which offers a} complementary route to equilibrium.}

\revisionblockstart

\subsection{Local dynamics around the long-run equilibrium}

We now describe decentralized adjustment at first order around
the long-run equilibrium. This representation will provide the
benchmark for comparing its speed of convergence with the
planner's in Section~\ref{sec:planner}.
{
To obtain this, we represent the solution as a series using a suitable basis
$(\phi_k)_{k\geq1}$ of $L^2(S)$ (with $\phi_1=\mu^{EQ}$). These functions are introduced in
Appendix~\ref{app:local-spectral-notation} as eigenfunctions of an
operator $\mathcal L$ defined there, and are orthonormal in a suitable
 weighted $L^2$ space. We refer the reader to that appendix for these
definitions, the associated eigenvalues $(\kappa_k)_{k\geq1}$, and
the coefficients $(a_k)_{k\geq1}$ of the initial density. This
notation will be used for both the decentralized (see the proposition just below) and planner dynamics.
}

\begin{Proposition}[Local decentralized dynamics]
\label{pr:local-decentralized}
Assume the hypotheses of Theorem~\ref{th:esistenzaeunicita},
and let $\mu^{EQ}$ be the equilibrium density
in~\eqref{eq:steadystatemu}.
{
The linearization of \eqref{eq:stato-di-nuovo} around $\mu^{EQ}$,
with initial condition $\mu_D(0)=\mu_0$, has the unique solution
}
\begin{equation}
\label{eq:local-decentralized}
\mu_D(t)
=\mu^{EQ}+\sum_{k\geq2}a_k e^{\lambda_{D,k}t}\phi_k, \quad \text{where}\;
\lambda_{D,k}
=-\frac{(1-\beta)\kappa_k}
       {\chi[\sigma(1-\beta)+\beta]},
\end{equation}
{
The series converges in $L^2(S)$ uniformly for $t\geq0$, and
$\mu_D(t)$ converges to $\mu^{EQ}$ in $L^2(S)$ as $t\to+\infty$.
} Generically,
\[
\|\mu_D(t)-\mu^{EQ}\|_{L^2(S)}
\sim C e^{\lambda_{D,2}t}
\qquad (t\to+\infty),
\]
for some $C>0$.
\end{Proposition}

\begin{proof}
See Appendix~\ref{app:proofs}.
\end{proof}

\gre{The convergence rate of the linearized decentralized dynamics
provides a measure of the speed of adjustment. Section~\ref{sec:planner} examines this quantity from the perspective of an intertemporal planner and compares its local adjustment
with the decentralized dynamics.}

\begingroup

\section{Intertemporal efficiency of reallocation}
\label{sec:planner}

In this section we introduce a planner who chooses the velocity at which firms reallocate across sectors.
This optimization problem provides a welfare benchmark for assessing whether decentralized reallocation is intertemporally efficient.

The planner maximizes discounted intertemporal welfare associated with the intratemporal allocation described in the previous sections. Instantaneous
welfare is utility from consumption, in the same form used by the representative agent in those sections,\footnote{Since the population has unit mass and firm ownership is equally distributed across households, individual, average, and aggregate consumption and utility coincide. Thus Benthamite and Millian aggregation criteria coincide with the social utility used here for the planner.} minus aggregate reallocation costs:
\begin{equation}
\label{eq:social-planner}
\sup_{(\mu,v)}
\int_0^\infty e^{-rt}
\left[
u(X(\mu(t)))
-\frac{\chi}{2}\int_S v(t,i)^2\mu(t,i)\,di
\right]dt,
\end{equation}
for some $r>0$. The planner's state equation is the continuity equation: sectoral inflows and outflows change the distribution of firms
while preserving their total mass, so admissible pairs consist of strictly positive probability
densities
\[
\mu\in C([0,\infty);C^2(S))
       \cap C^1([0,\infty);C(S))
\]
and velocities $v\in C([0,\infty);C^1(S))$ satisfying
\[
\partial_t\mu+\partial_i(\mu v)=0,
\qquad \mu(0)=\mu_0.
\]
The payoff may equal $-\infty$; it is well-defined because
$X(\mu)\leq X^{EQ}$.

When $\mu_0=\mu^{EQ}$, the stationary path $(\mu^{EQ},0)$
attains the upper bound $u(X^{EQ})/r$ and is therefore optimal. Consequently, the optimal feedback $K$, which expresses the velocity field as a function of the current distribution $\mu$, must satisfy $K(\mu^{EQ})=0$.


{
To compare the convergence speeds of decentralized and optimal
adjustment, we linearize the state equation around $\mu^{EQ}$ after
substitution of the planner's optimal feedback. The following
proposition gives the resulting series representation and the
convergence rate of this linear approximation.
}

{
The representation below uses the same basis $(\phi_k)_{k\geq1}$
of $L^2(S)$, eigenvalues $(\kappa_k)_{k\geq1}$, and initial
coefficients $(a_k)_{k\geq1}$, all defined in
Appendix~\ref{app:local-spectral-notation}.
}

\begin{Proposition}[Local planner dynamics]
\label{pr:local-planner}
{
Assume the hypotheses of Proposition~\ref{pr:local-decentralized},
and let $u$ be concave and strictly increasing.
For $A$, $u$ and $\mu_0$ sufficiently regular, and $\mu_0$
sufficiently close to $\mu^{EQ}$, as specified at the beginning
of the proof, the planner admits a continuously differentiable
optimal feedback $K$.
The linearization around $\mu^{EQ}$ of the state equation
$\partial_t\mu=-\partial_i(\mu K(\mu))$,
with initial condition $\mu_P(0)=\mu_0$,
has the unique solution
\begin{equation}
\label{eq:local-planner}
\mu_P(t)=\mu^{EQ}+\sum_{k\geq2}a_k e^{\lambda_{P,k}t}\phi_k,
\qquad
\lambda_{P,k}
=\frac{r-\sqrt{r^2+
\frac{4X^{EQ}u'(X^{EQ})(1-\beta)\kappa_k}
{\chi[\sigma(1-\beta)+\beta]}}}{2}.
\end{equation}
}
{
The series converges in $L^2(S)$ uniformly for $t\geq0$, and
$\mu_P(t)$ converges to $\mu^{EQ}$ in $L^2(S)$ as $t\to+\infty$.
} Generically,
\[
\|\mu_P(t)-\mu^{EQ}\|_{L^2(S)}
\sim C e^{\lambda_{P,2}t}
\qquad (t\to+\infty),
\]
for some $C>0$.
\end{Proposition}
\begin{proof}
See Appendix~\ref{app:proofs}.
\end{proof}
\subsection{Comparison of convergence speeds}

In both cases we linearize the state equation after substituting
the relevant feedback. The two linearized trajectories have the
same initial coefficients and spatial eigenfunctions, so their
adjustment speeds are determined by their temporal exponents.

\begin{Proposition}[Comparison of local convergence speeds]
\label{pr:planner-comparison}
Under the hypotheses of Proposition~\ref{pr:local-planner},
set $\lambda_D=\lambda_{D,2}$ and
$\lambda_P=\lambda_{P,2}$.
For generic common initial densities, these exponents
determine the asymptotic convergence speeds of the
decentralized and planner linearizations, and
\begin{equation}
\label{eq:planner-comparison}
\lambda_P\lesseqgtr\lambda_D
\quad\Longleftrightarrow\quad
X^{EQ}u'(X^{EQ})
\gtreqless
r+\frac{(1-\beta)\kappa_2}
        {\chi[\sigma(1-\beta)+\beta]}.
\end{equation}
Thus the planner's linearization converges faster, at the same exponential
rate, or more slowly according as the left-hand side
of the latter comparison is greater than, equal to,
or smaller than its right-hand side.
\end{Proposition}

\begin{proof}
See Appendix~\ref{app:proofs}.
\end{proof}

\begin{Corollary}[Constant productivity and CRRA utility]
\label{cor:planner-constant-A}
Under the hypotheses of Proposition~\ref{pr:planner-comparison},
suppose that $A(i)\equiv A_0>0$ and
$u(X)=(X^{1-\vartheta}-1)/(1-\vartheta)$
for $\vartheta\geq0$, $\vartheta\neq1$,
with $u(X)=\log X$ when $\vartheta=1$.
Then
\begin{equation}
\label{eq:planner-crra}
\lambda_P\lesseqgtr\lambda_D
\quad\Longleftrightarrow\quad
A_0^{1-\vartheta}
n^{\frac{1-\vartheta}{\sigma-1}}
\gtreqless
r+\frac{4\pi^2(1-\beta)}
        {\chi n^2[\sigma(1-\beta)+\beta]}.
\end{equation}
\end{Corollary}

\begin{proof}
See Appendix~\ref{app:proofs}.
\end{proof}


At the linearized level, a higher discount rate $r$ slows the planner's adjustment
without affecting decentralized adjustment.
A higher reallocation cost $\chi$ slows both linearized
paths, but increases the ratio of the planner's convergence
speed to the decentralized speed. A higher $\beta$ has
the same effect.
Productivity $A_0$ leaves the decentralized rate unchanged;
under CRRA utility, it accelerates the planner when
$\vartheta<1$, has no effect when $\vartheta=1$,
and slows it when $\vartheta>1$.
Changes in $n$ and $\sigma$ also affect equilibrium
consumption, so their effects need not have a uniform sign
for general $u$.

\endgroup

\revisionend \section{Conclusions \label{sec:conclusions}}

We have developed a framework for analysing the sectoral reallocation of firms (capital) in the presence of quadratic reallocation costs, using a gradient flow approach and the JKO scheme to characterize both short- and long-run competitive equilibria. We prove convergence to the long-run competitive equilibrium and establish its efficiency properties. \rev{A planner with general utility shares that stationary allocation. Under the stated regularity assumptions, its linearized convergence rate can be higher or lower than the decentralized rate once movement costs are included.}
Our analysis highlights the critical role of the elasticity of substitution between goods and \red{sectoral productivity} in shaping the distribution of firm sizes across sectors.

\red{Appendices~\ref{sub:spillovers}--\ref{sec:immobilework} study extensions with intra-sectoral externalities, fixed reallocation costs, non-symmetric preferences, and labour immobility, while Appendix~\ref{app:quantitative} provides an empirical application to a large sample of EU firms.}

A first extension of our work would be to relax the assumptions of homogeneous space and within-sector technology, following the insights of \citet{melitz2003impact}. Another important direction is relaxing the assumption of myopic firms and, more specifically, of static expectations, particularly when analysing business cycle dynamics. \citet{sargent1993bounded} suggests possible ways to address this issue, for example through the introduction of more sophisticated adaptive expectations regarding sectoral profits. Our conjecture is that such extensions would affect the properties of the transitional dynamics, while leaving the main characteristics of the long-run competitive equilibrium largely unchanged.
Another natural extension involves jointly considering capital and labour mobility in the presence of frictions. This joint reallocation introduces non-trivial technical challenges but enables the exploration of richer and more complex dynamics. In particular, complementarities in production between these two factors may lead to the concentration of resources in a few sectors, where returns to individual inputs (wages and profit rates) are dynamically enhanced by continuous inflows of both capital and labour from other sectors \citep{krugman1991increasing, krugman1994complex}.

Our work isolates the role of sectoral reallocation. A richer framework would allow for positive saving and capital accumulation, raising the question of which sectors benefit from a growing capital stock. A key challenge in this context is defining the entry process for new firms and understanding how such dynamics shape sectoral reallocation \citep{melitz2003impact}.
Related to this is the possibility of introducing a financial sector, which could facilitate capital mobility across sectors without incurring quadratic reallocation costs. In this setting, it would be important to revisit the nature of reallocation costs, potentially allowing for capital losses during intersectoral transfers \citep{eisfeldt2006capital} and introducing an opportunity cost of investing in the production of goods.
\noindent Finally, incorporating some form of nominal rigidity in the short-run adjustment of prices and sectoral wages could enhance the model’s relevance for studying medium-run fiscal and monetary policy effects. Such rigidities may generate sluggish adjustment dynamics, making them particularly important for macroeconomic analysis \citep{blanchard2025convergence}.

In summary, our framework provides a solid foundation for analysing sectoral dynamics and offers a versatile tool for studying complex processes of factor reallocation. We are confident that it can serve as a basis for future extensions, further enhancing the applicability and explanatory power of our approach to understanding the dynamics of a capitalist economy.

\begin{footnotesize}

\nocite{BoydVandenberghe2004,Brezis2011,CarrilloEspositoSkrzeczkowskiWu2024,CostabelDaugeNicaise2010,Deimling1985,Griliches1992,Jovanovic1989,LionsMagenes-2-EN,Romer1986,Romer1990a,Taylor1991,allen2014trade,ambrosio2005gradient,eisfeldt2006capital,fujita_thisse_2002,hadamard1923lectures,iacobelli2019weighted,lee2006intersectoral,lunardi15analytic,mas1995microeconomic,efron2011tweedie,gupta2014efficiency}
\bibliographystyle{apalike}
\bibliography{biblio}

\appendix

\section{Proofs for the baseline model}
\label{app:proofs}

\begin{proof}[Proof of Proposition~\ref{pr:shor-run-equilibrium}]
We only sketch the proof, since the arguments are standard.
Profit maximization gives $wl(i)=\beta p(i)q(i)$ and hence
$\pi(i)=(1-\beta)p(i)q(i)$. Integrating these identities over firms (i.e. with respect to $\mu(i)\,di$) and using labour market clearing and $Y=1$ yields $w=\beta$ and $\Pi=1-\beta$.

Under these normalizations, CES demand and firms' labour demand
take the form
\begin{equation}
\label{eq:demandGoodi}
x(i)=P^{\sigma-1}p(i)^{-\sigma},
\qquad
l(i)=
\left[\frac{\beta p(i)A(i)}{w}\right]^{\frac{1}{1-\beta}}
=
\left[p(i)A(i)\right]^{\frac{1}{1-\beta}}.
\end{equation}
Combining these expressions with goods market clearing,
$x(i)=\mu(i)A(i)l(i)^\beta$, gives
\eqref{eq:value-of-pi}. Substitution into the CES price index
and rearrangement yield
\[
P^{\frac{1-\sigma}{\sigma(1-\beta)+\beta}}
=
\int_0^n
\left[A(i)\mu(i)^{1-\beta}\right]^{
    \frac{\sigma-1}{\sigma(1-\beta)+\beta}
}\,di,
\]
which gives \eqref{eq:X-P}, since $X=Y/P=1/P$.

Finally, $\pi(i)=(1-\beta)p(i)q(i)$ gives
\eqref{eq:profitti_dopo}, while
$Y(i)=p(i)x(i)$ and $L(i)=\beta Y(i)/w=Y(i)$
yield the stated expressions for sectoral revenue and employment.
Since $\int_0^nY(i)\,di=1$, the resulting allocation also
satisfies labour market clearing.
\end{proof}

\begin{proof}[Proof of Theorem \ref{th:esistenzaeunicita}]

In this proof, to lighten the notation, we will define  $\rho \equiv \sigma(1-\beta) + \beta$. Notice that $\rho \in (\sigma,1)$ or $\rho \in (1,\sigma)$ depending on whether $\sigma \lessgtr 1$.

We start by looking at the following PDE:
\begin{equation}
\label{eq:stato-di-santambrogio}
\left \{
\begin{array}{l}
\displaystyle \partial_t V(t,i) = -\frac{1-\beta}{{\chi}} \partial_i \left( V(t,i) \partial_i \left( A(i)^{\frac{1}{1-\beta} \left( 1 - \frac{1}{\rho} \right)} V(t,i)^{-1/\rho} \right) \right)\\[7pt]
\displaystyle V(t,0) = V(t,n)\\[7pt]
\displaystyle V(0,i) = \mu_0(i).
\end{array}
\right .
\end{equation}
We develop the right-hand side of the (\ref{eq:stato-di-santambrogio}) in more explicit terms for later reference. If we avoid writing the dependences $(t,i)$ and we denote $z = \frac{1}{1-\beta} \left( 1 - \frac{1}{\rho} \right)$  We have that
\begin{multline}
\label{eq:def-f}
-\frac{1-\beta}{{\chi}} \partial_i \left( V(t,i) \partial_i \left( A(i)^{\frac{1}{1-\beta} \left( 1 - \frac{1}{\rho} \right)} V(t,i)^{-1/\rho} \right) \right) \\
= -\frac{1-\beta}{{\chi}} V^{-1/\rho} \bigg [ (1-2/\rho) \left ( z A^{z-1} \right ) (\partial_i A) (\partial_i V) + \frac{1}{\rho^2} A^z \frac{(\partial_i V)^2}{{V}}
\\+ z (z-1) A^{z-2}(\partial_i A)^2 V + z A^{z-1}(\partial^2_i A) V
- \frac{1}{\rho} A^z (\partial^2_i V) \bigg ]=: f(V, \partial_i V, \partial^2_i V).
\end{multline}

Thanks to Theorem 1.2 of \cite{iacobelli2019weighted} (the results hold for Neumann and also periodic boundary conditions so they include our case) it has a unique solution in  $L^2_{loc}([0, +\infty), H^1(S))$ and thanks to  their Theorem 1.3 there exists $M>0$ such that, for all $t\geq 0$
\begin{equation}
\label{eq:boundeSantambrogio}
\frac{1}{M} < V(t,i) < M, \qquad \text{for all $i\in S$.}
\end{equation}

This fact guarantees further regularity of $V$:
\begin{itemize}
\item[(i)] The positivity result (\ref{eq:boundeSantambrogio}) ensure that the operator $f$ defined in (\ref{eq:def-f}) is uniformly elliptic and then one can use the standard results for parabolic equations for recursively get (using for instance Theorem 15.2 page 81 and then 6.1 page 33 of \citealp{LionsMagenes-2-EN} or Krylov-Safonov and then Schauder estimates) the regularity in the interior and that $V\in C([0, +\infty), C^2(S))$.

\item[(ii)] For positive and regular initial data as the $\mu_0$ considered in the text of the theorem one can gain some additional regularity at $t=0$. In particular, we can for instance apply Theorem 8.1.3 of \cite{lunardi15analytic}\footnote{We use Theorem 8.1.3 of \cite{lunardi15analytic} specifying the abstract setting as follows:
$X = C(S)$, $D = C^2(S)$, $\mathcal{O} = \left \{ v\in D \; : \; \frac{1}{M} < v(i) < M, \text{for all $i\in S$.}  \right \}$, $F (t, v)(\xi) = f (v(\xi), v'(\xi), v''(\xi))$ where $f$ is defined in (\ref{eq:def-f}).
Remark that, since $S$ is compact, the condition $A\in C^2(S)$ implies in particular that the first and the second derivatives of $A$ are bounded. As remarked below equation (8.0.3) page 287 in \cite{lunardi15analytic} in our specification, the condition (8.0.3) is verified as long as $\frac{\partial f}{\partial q} (u, p, q) > 0$ which is assured by the positivity of $A$ and of $V$. The Lipschitz and the Holder conditions (8.1.1) page 290 are verified thanks to the bounds in the definition of $\mathcal{O}$, the continuity and the Fréchet differentiability are assured near $\mu_0$ (the regularity result is local in time) thanks to its regularity and its boundedness. Again the regularity of $\mu_0$ ensures that the condition (8.1.18), page 295 is verified.} to be sure that there exists $\alpha$ such that, for all\footnote{The result of Theorem 8.1.3 of \cite{lunardi15analytic} only ensures this fact until to a small $\delta>0$ but, since all the boundedness are homogeneous in time, the result extends to $[0,T]$.} $T>0$,
\begin{equation}
\label{eq:resulta-Lunardi}
V(t)\in C^{1+\alpha}([0,T), C(S)).
\end{equation}
\end{itemize}
We want to use $V$ to obtain a solution of (\ref{eq:stato-di-nuovo}). If $\mu$ (we will prove its existence at the end)  is a $C^{1,2}$ solution of (\ref{eq:stato-di-nuovo}) and there exist two regular functions $\phi\colon \mathbb{R}^+ \to \mathbb{R}^+$ and $V\colon \mathbb{R}^+ \times S \to \mathbb{R}$ such that $\mu(t,i) = V(\phi(t),i)$, then $\mu$ satisfies
\begin{equation*}
\partial_{t}\mu(t,i) = \phi'(t) \partial_{t} V(\phi(t),i).
\end{equation*}
Moreover, if $\phi:\RR^{+}\to \RR^{+}$ is monotonously increasing, $\phi(0) = 0$ and
\begin{equation}\label{eq:timeChange}
\phi'(t) = \frac{1}{\left( \int_0^n \left[ A(j)^{1/(1-\beta)} \mu(t,j) \right]^{1-\frac{1}{\rho}} dj \right)} = \frac{1}{\left( \int_0^n \left[ A(j)^{1/(1-\beta)} V(\phi(t),j) \right]^{1-\frac{1}{\rho}} dj \right)},
\end{equation}
then $V(\phi(t),i)$ is a solution to Equation \eqref{eq:stato-di-nuovo}.
Hence, to prove existence and uniqueness to \eqref{eq:stato-di-nuovo} it is enough, starting from that of \eqref{eq:stato-di-santambrogio}, to show the existence and uniqueness of the time change $\phi(t)$ i.e. the existence of a unique $\phi$ such that
\begin{equation*}
\begin{cases}
\displaystyle \phi'(t) = G(\phi(t)) :=  \frac{1}{ \int_0^n \left[ A(j)^{1/(1-\beta)} V(\phi(t),j) \right]^{1-\frac{1}{\rho}} dj }.\\
\phi(0) = 0
\end{cases}
\end{equation*}
To prove that this equation has a unique solution, it is enough to prove that $G$ is Lipschitz-continuous. So we need to estimate
\begin{multline*}
\abs{G(x) - G(y)} = \abs{\frac{1}{ \int_0^n \left[ A(j)^{1/(1-\beta)} V(x,j) \right]^{1-\frac{1}{\rho}} dj } - \frac{1}{ \int_0^n \left[ A(j)^{1/(1-\beta)} V(y,j) \right]^{1-\frac{1}{\rho}} dj }} \\
= \abs{\frac{\int_0^n \left[ A(j)^{1/(1-\beta)} V(x,j) \right]^{1-\frac{1}{\rho}} dj-\int_0^n \left[ A(j)^{1/(1-\beta)} V(y,j) \right]^{1-\frac{1}{\rho}} dj}{\left(\int_0^n \left[ A(j)^{1/(1-\beta)} V(x,j) \right]^{1-\frac{1}{\rho}} dj\right)\left(\int_0^n \left[ A(j)^{1/(1-\beta)} V(y,j) \right]^{1-\frac{1}{\rho}} dj\right)}}.
\end{multline*}
The denominator of this expression, thanks to (\ref{eq:boundeSantambrogio}) and the strict positivity and boundedness of $A$, is always greater than some constant $1/C>0$ so the previous expression is smaller than
\begin{multline*}
C \cdot \abs{	\int_0^n \left[ A(j)^{1/(1-\beta)} V(x,j) \right]^{1-\frac{1}{\rho}} dj-\int_0^n \left[ A(j)^{1/(1-\beta)} V(y,j) \right]^{1-\frac{1}{\rho}} dj} = \\
= {C \cdot}\abs{\int_0^n \left[ A(j)^{1/(1-\beta)} V(x,j) \right]^{1-\frac{1}{\rho}} - \left[ A(j)^{1/(1-\beta)} V(y,j) \right]^{1-\frac{1}{\rho}} dj}.
\end{multline*}

Again, thanks to (\ref{eq:boundeSantambrogio}) and to the strict positivity and boundedness of $A$, we know that the quantity $\left[ A(j)^{1/(1-\beta)} V(x,j) \right]$ stays uniformly away from zero. Since the map $z \mapsto z^{1 - \frac{1}{\rho}}$ is Lipschitz-continuous on bounded subsets of $(0, +\infty)$, it follows that applying this transformation yields a Lipschitz functional. Therefore, possibly adjusting the constant $C$ (also depending on the bounds on $A$), the previous expression is smaller than
\begin{equation}
\label{eq:ultimastima}
\leq C \int_{0}^{n}\abs{V(x,j)-V(y,j)} \,dj.
\end{equation}
Since (\ref{eq:resulta-Lunardi}) holds we know that
\[
\sup_{i\in S} \sup_{t\in [x,y]} \abs{\frac{\partial V}{\partial t}}(t,i) : N <+\infty
\]
So (\ref{eq:ultimastima}) is smaller than
\[
C \int_{0}^{n}N\abs{x-y} \, dj = CNn \abs{x-y}
\]
and then $G$ is Lipschitz continuous on all intervals $[0, T]$.Therefore, the time change exists, and it is unique and so is the solution $\mu$ of (\ref{eq:stato-di-nuovo}).
\end{proof}

\begin{proof}[Proof of Theorem \ref{th:riscritturaGradientFlow}]
We restrict the proof to the case $\sigma < 1$. In the case
$\sigma > 1$ our proof in analogous, since the theory developed
in \cite{iacobelli2019weighted} falls back in the classical theory
of fast diffusion equation, see
\cite{CarrilloEspositoSkrzeczkowskiWu2024}.

Introduce $\rho \equiv \sigma(1-\beta) + \beta$.
Consider the equation
\begin{equation*}
\left\{
\begin{array}{l}
\displaystyle
\partial_t \nu(t,i)
=
-\dfrac{1-\rho}{\rev{\chi}\rho}
\partial_i \left(
\nu(t,i)\,
\partial_i \left(
A(i)^{\frac{1}{1-\beta}\left(1-\frac{1}{\rho}\right)}
\nu(i)^{-1/\rho}
\right)
\right)
\\[7pt]
\displaystyle \nu(t,0)=\nu(t,n)
\\[7pt]
\displaystyle \nu(0,i)=\nu_0(i).
\end{array}
\right.
\end{equation*}

The results of \cite{iacobelli2019weighted}
(see in particular page 1171 and subsequent%
\footnote{To adapt our notation to theirs
(see (1.1) in \citealt{iacobelli2019weighted}) we can set
\begin{equation*}
\rho(x)
=
A(x)^{\frac{1}{1-\beta}\rev{(}1-\frac{1}{\rho}\rev{)}},
\quad
\text{and }
r=\frac{1}{\rho}-1.
\end{equation*}
})
ensure that its (unique, weak, see Theorem 1.2 page 1170)
solution can be characterized as the gradient flow
(defined as the limit of the Jordan-Kinderlehrer-Otto scheme,
see Lemma 2.2 page 1174 of \citet{iacobelli2019weighted})
of the functional \rev{$-\mathcal G/\chi$, where}
\begin{equation}
\mathcal{G}(\nu)
:=
\int_0^n
A(i)^{\frac{1}{1-\beta}\left(1-\frac{1}{\rho}\right)}
\nu(i)^{\frac{\rho-1}{\rho}}\,di
\end{equation}
i.e.
\[
\frac{d\nu}{dt}
=
\rev{-\frac{1}{\chi}}
\nabla_{W_2}\mathcal{G}(\nu(t)).
\]

If we consider
\[
\mathcal{F}(\nu)
=
\frac{\rho}{\rev{\rho-1}}(1-\beta)
\ln\left(\mathcal{G}(\nu)\right)
\]
we have
\begin{equation*}
\rev{\frac{1}{\chi}}
\nabla_{W_2}\mathcal{F}(\mu(t))
=
\rev{\frac{1}{\chi}}
\frac{\rho}{\rev{\rho-1}}(1-\beta)
\frac{
\nabla_{W_2}\mathcal{G}(\mu(t))
}{
\mathcal{G}(\mu(t))
}.
\end{equation*}
which corresponds to the expression appearing in the right
hand side of (\ref{eq:stato-di-nuovo}) so that
\[
\frac{d\mu}{dt}
=
\rev{\frac{1}{\chi}}
\nabla_{W_2}\mathcal{F}(\rev{\mu}(t)).
\]
and this concludes the proof.
\end{proof}

\begin{proof}[Proof of Corollary \ref{cor:crescitaOttimal}]
The fact that the gradient flow evolves along the curves of
maximal slope is a very general result; see, for example,
Theorem 11.1.3, page 283 of \cite{ambrosio2005gradient}%
\rev{, applied here to the energy $-\mathcal F/\chi$}.
Therefore, in particular,
\(t\mapsto\mathcal F(\mu(t))\) is increasing.
Moreover, we have
\begin{equation*}
\begin{aligned}
P(t)
&=
\left(
\int_0^n
\left[A(i)^{1/(1-\beta)}\mu(t,i)\right]^{1-1/\rho}
\,di
\right)^{\frac{\rho}{1-\sigma}}
\\
&=
\exp\left(
\frac{\rev{\rho-1}}{\rho}
\frac{1}{1-\beta}\mathcal F(\mu(t))
\right)^{\frac{\rho}{1-\sigma}}.
\end{aligned}
\end{equation*}
Thus, \(t\mapsto P(t)\) decreases along the trajectories.
The behaviours of \(X(t)\) and \(U(t)\) follow from this
fact and from the relations (\ref{eq:X-P}) and
(\ref{eq:U-P}).
\end{proof}
\begin{proof}[Proof of Theorem
\ref{teo:convergence}]

Finding a long-run competitive equilibrium amounts to identifying a stationary state of the dynamics in (\ref{eq:stato-di-nuovo}). Indeed, any long-run competitive equilibrium must necessarily also be a short-run competitive equilibrium, otherwise profit and utility maximization would be violated or markets would fail to clear—so that profits are given by (\ref{eq:profitti_dopo}). If, as required in the definition of a long-run competitive equilibrium, no firm has an incentive to switch sectors, i.e. it chooses $v=0$ in Problem (\ref{eq:ruleGoverningRiallocationFirms}), then $\pi$ must be constant in $i$. This means that there exists some positive real number $C_{\infty}$ such that
\begin{equation*}
\mu(i)^{-\frac{1}{\rho}} A(i)^{\frac{1}{1-\beta}\frac{\rho-1}{\rho}} = C_{\infty},
\end{equation*}
and, using the fact that $\mu$ integrates to one, (\ref{eq:steadystatemu}) distribution can be derived.

\smallskip

On the other hand, any stationary state $\mu$ of (\ref{eq:stato-di-nuovo}) must satisfy that the term
\[
\Psi(i)\;:=\;(1-\beta)\frac{A(i)^{\frac{1}{1-\beta}\frac{\rho-1}{\rho}}\;\mu(i)^{-\frac1{\rho}}}{\displaystyle  \int_0^n \left[ A(j)^{\frac{1}{1-\beta}} \mu(j) \right]^{\frac{\rho -1}{\rho}} \, dj }
\]
is constant on $S$. Indeed, we can write (\ref{eq:stato-di-nuovo}) as a continuity equation:
\[
\partial_t \mu \;=\; -\,\partial_i J, \qquad J(i):=\rev{\chi^{-1}}\mu(i)\,\partial_i\Psi(i).
\]
At stationarity $\partial_t \mu=0$, hence $\partial_i J=0$ and $J(i)\equiv c$ for some constant $c$.
Integrating over $S$ we have
\[
0=\int_S \partial_i \Psi\,di=\rev{\chi}\int_S \frac{J}{\mu}\,di=\rev{\chi} c\int_S \frac{1}{\mu}\,di,
\]
(the first equality holds since $\Psi$ is periodic) and since $\mu>0$ the integral is strictly positive, forcing $c=0$. Hence $J\equiv0$, so with $\mu>0$ we get $\partial_i\Psi=0$ everywhere, i.e. $\Psi(i)\equiv C_\infty$. Equivalently,
\[
\mu(i)^{-\frac{1}{\rho}}\,A(i)^{\frac{1}{1-\beta}\frac{\rho-1}{\rho}} \;=\; C_\infty.
\]
This is exactly the condition of no incentive to
reallocate for firms and, again using the fact that $\mu$ integrates to one, we get (\ref{eq:steadystatemu}). Once we couple this condition with the prices and the quantity found in Proposition \ref{pr:shor-run-equilibrium}, in (\ref{eq:demandGoodi}) and (\ref{eq:demandGoodi}) we get we have, by the short-run competitive equilibrium characterization, that all markets clear, consumers and
firms are maximizing their utility and workers have no incentive to
reallocate. So we have a long-run competitive equilibrium.

\bigskip

To show the convergence part of the Theorem, first of all, consider the time change $\phi(t)$ introduced in the proof of Theorem \ref{th:esistenzaeunicita}, and defined in \eqref{eq:timeChange}. Notice that, by equation \eqref{eq:boundeSantambrogio}, such a $\phi$ will automatically be strictly increasing and with bounded first derivative, i.e.
\begin{equation*}
c_{\phi} \leq \phi'(t) \leq C_{\phi}
\end{equation*}
for some constants $c_{\phi}>0,C_{\phi}>0$,  so that $\phi \in C(\RR^+)$ and $\lim_{t\to\infty}\phi(t) = +\infty$. This property is fundamental since it guarantees the uniqueness of the solution of \eqref{eq:stato-di-nuovo} from that of $\phi(t)$. Moreover, since $\phi(0) = 0$ it also implies that $\phi(t) \geq K_\phi \cdot t$ for some $K_{\phi} > 0$.

Consider Equation \eqref{eq:stato-di-santambrogio}. Thanks to Theorem 1.4 of \cite{iacobelli2019weighted} for the function $V$, solution to \eqref{eq:stato-di-santambrogio} it holds\footnote{Notice that, the steady state of \eqref{eq:stato-di-santambrogio} is the same of \eqref{eq:stato-di-nuovo}.}
\begin{equation}
\norm{V(t,\cdot) - \mu^{EQ}(\cdot)}_{L^2(S)} \leq C_1 e^{-C_2 t},
\end{equation}
(which also implies the uniqueness of the steady state). Therefore, since $\mu(t,i) = V(\phi(t),i)$ one immediately has
\begin{equation}
\norm{\mu(t,\cdot) - \mu^{EQ}(\cdot)}_{L^2(S)} \leq C_1 e^{-C_2 \phi(t)}.
\end{equation}
However, this immediately translates to
\begin{equation}
\norm{\mu(t,\cdot) - \mu^{EQ}(\cdot)}_{L^2(S)} \leq C_1 e^{-C_2 K_{\phi} t}
\end{equation}
which ends the proof.

To derive the expressions for $\pi(i)$, $Y(i)$, and $L(i)$, it suffices to substitute the expression from (\ref{eq:steadystatemu}) into the corresponding formulas given in Proposition \ref{pr:shor-run-equilibrium}.
\end{proof}

\begin{proof}[Proof of Theorem \ref{th:mu-e-ottimale}]
We define $\rho := \sigma(1-\beta) + \beta$ and notice that
$\frac{\rho-1}{\sigma-1} > 0$ for all $\sigma>0$ with $\sigma \neq 1$.

Consider for the case $\sigma > 1$, i.e. $\rho > 1$ (the case $0<\sigma < 1$ follows from similar arguments) and consider the quantity
\begin{equation*}
 X(\mu) = \left( \int_0^n \left[ A(i)^{\frac{1}{1-\beta}} \mu(i) \right]^{\frac{\rho-1}{\rho}} \, di \right)^{\frac{\rho}{\sigma-1}} = \left( \int_0^n \left[ A(i)^{\frac{1}{1-\beta}} \mu(i) \right]^{\frac{\rho-1}{\rho}} \, di \right)^{\frac{\rho}{\rho-1}\frac{\rho-1}{\sigma-1}}.
\end{equation*}
Since the exponent outside of the integral $X$ is maximized if and only if the value of the integral is maximized, so consider $\mu$, such that $\int \mu = 1$. We define $\bar{A}(i):= A(i)^{\frac{1}{1-\beta}\frac{\rho-1}{\rho}}$ and $\gamma = \frac{\rho-1}{\rho}$. Since $0 < \gamma < 1$ the $\gamma$-power function is concave and then one
\begin{equation*}
\int_0^n \bar{A}(i) \mu(i)^\gamma  - \bar{A}(i)\mu^{EQ}(i)^\gamma\,di \leq \int_0^n \bar{A}(i)\gamma \mu^{EQ}(i)^{\gamma-1}\left[\mu(i)-\mu^{EQ}(i) \right]\,di.
\end{equation*}
On the other hand
\begin{equation}
\label{eq:optimBaseCase1}
\int_0^n \bar{A}(i)\gamma \mu^{EQ}(i)^{\gamma-1}\left[\mu(i)-\mu^{EQ}(i) \right]\,di = \frac{\displaystyle \int_0^n \bar{A}(i)\gamma A(i)^{\frac{\rho-1}{1-\beta}(\gamma-1)}\left[\mu(i)-\mu^{EQ}(i) \right]\,di}{\displaystyle \left(\int_0^n A(j)^{\frac{\sigma(1-\beta) + \beta-1}{1-\beta}} \, dj\right)^{\gamma-1}}.
\end{equation}
and
\begin{equation*}
A(i)^{\frac{\rho-1}{1-\beta}(\gamma-1)} = A(i)^{\frac{\rho-1}{1-\beta}\frac{-1}{\rho}} = \bar{A}(i)^{-1}
\end{equation*}
so that \eqref{eq:optimBaseCase1} becomes
\begin{equation}
= \frac{\displaystyle \int_0^n \gamma \left[\mu(i)-\mu^{EQ}(i) \right]\,di}{\displaystyle \left(\int_0^n A(j)^{\frac{\sigma(1-\beta) + \beta-1}{1-\beta}} \, dj\right)^{\gamma-1}} = 0.
\end{equation}
Therefore for each $\mu$ such that $\int \mu = 1$ one has
\begin{equation}
\int_0^n \bar{A}(i) \mu(i)^\gamma\rev{\,di} \leq \rev{\int_0^n} \bar{A}(i)\mu^{EQ}(i)^\gamma\,di
\end{equation}
\bigskip
which ends the proof.

A similar argument holds for $0 < \sigma < 1$. In this case, $\gamma = (\rho-1)/\rho < 0$ and the outer exponent $\rho/(\sigma-1) < 0$. Thus, $x \mapsto x^\gamma$ is convex, and maximizing $X(\mu) = ( \int \bar{A} \mu^\gamma )^{\rev{\rho/(\sigma-1)}}$ requires minimizing the integral. The convexity inequality, combined with the previously shown zero value of the first-order term, confirms that $\mu^{EQ}$ minimizes the integral, thus maximizing $X(\mu)$.
\end{proof}

\begingroup

\subsection{Spectral notation and proofs for the local dynamics}
\label{app:local-spectral-notation}

We first introduce the spectral notation shared by
Propositions~\ref{pr:local-decentralized} and~\ref{pr:local-planner}.
Throughout this subsection, $\rho$ retains the meaning introduced
in the proof of Theorem~\ref{th:esistenzaeunicita}.
We use the periodic Sobolev spaces $H^s(S)$; for integer $s$, their
elements have weak derivatives up to order $s$ in $L^2(S)$.

Under the hypotheses of Theorem~\ref{th:esistenzaeunicita}, the
equilibrium density $\mu^{EQ}$ in~\eqref{eq:steadystatemu} is in
$C^2(S)$ and is bounded away from zero. Define the real Hilbert space
\[
H=L^2(S,di/\mu^{EQ}),\qquad
\langle f,g\rangle_H
=\int_S\frac{f(i)g(i)}{\mu^{EQ}(i)}\,di.
\]
Its norm is equivalent to the usual $L^2(S)$ norm. On $H$, consider
\[
\mathcal Lf
=-\partial_i\!\left[
\mu^{EQ}\partial_i\!\left(\frac{f}{\mu^{EQ}}\right)
\right],\qquad D(\mathcal L)=H^2(S).
\]
The associated form on $H^1(S)$ is
\[
\mathfrak q(f,g)
=\int_S\mu^{EQ}
\partial_i\!\left(\frac f{\mu^{EQ}}\right)
\partial_i\!\left(\frac g{\mu^{EQ}}\right)\,di,
\]
which is symmetric and nonnegative.
Standard periodic elliptic theory implies that $\mathcal L$ is nonnegative and self-adjoint and that $(I+\mathcal L)^{-1}:H\to H$ is compact.
Applying the spectral theorem \citep[Theorem~6.11]{Brezis2011}
to this compact self-adjoint operator yields a real orthonormal eigenbasis $(\phi_k)_{k\geq1}$ of $H$ for $\mathcal L$, with nonnegative eigenvalues. Since $\mathcal L(f)=0$ exactly when $f$ is a multiple of $\mu^{EQ}$, and $\|\mu^{EQ}\|_H=1$, we can choose the numbering so that
\[
\mathcal L\phi_k=\kappa_k\phi_k,\qquad
0=\kappa_1<\kappa_2\leq\kappa_3\leq\cdots\longrightarrow+\infty,
\qquad \phi_1=\mu^{EQ}.
\]
Since the norms of $L^2$ and of $H$ are equivalent, the functions $(\phi_k)$ are also a basis of $L^2(S)$.

For an initial probability density $\mu_0\in L^2(S)$, we set
\[
a_k=\langle\mu_0,\phi_k\rangle_H
=\int_S\frac{\mu_0(i)\phi_k(i)}{\mu^{EQ}(i)}\,di.
\]
Then $a_1=1$ and
\[
\mu_0=\mu^{EQ}+\sum_{k\geq2}a_k\phi_k.
\]

\begin{proof}[Proof of Proposition~\ref{pr:local-decentralized}]
For a perturbation $h$ with $\int_S h\,di=0$, differentiating
\eqref{eq:profitti_dopo} at equilibrium gives
\[
D\pi(\mu^{EQ})[h]
=-\frac{1-\beta}{\rho}\frac{h}{\mu^{EQ}}.
\]
As $\pi(\mu^{EQ})$ is constant across sectors, the linearization
of~\eqref{eq:stato-di-nuovo} is
\[
\partial_t(\mu_D-\mu^{EQ})
=-\frac{1-\beta}{\chi\rho}\mathcal L(\mu_D-\mu^{EQ}).
\]
Expansion in the eigenbasis therefore yields
\[
\mu_D(t)=\mu^{EQ}+\sum_{k\geq2}a_k e^{\lambda_{D,k}t}\phi_k,
\qquad
\lambda_{D,k}=-\frac{(1-\beta)\kappa_k}{\chi\rho},
\]
which is~\eqref{eq:local-decentralized}. The series converges in $L^2(S)$ uniformly for $t\geq0$ by Parseval's identity \citep[Corollary~5.10]{Brezis2011}. The remaining claims follow.
\end{proof}

\begin{proof}[Proof of Proposition~\ref{pr:local-planner}]
We assume in addition that $A\in C^\infty(S)$,
$u\in C^\infty((0,\infty))$, and $\mu_0\in H^4(S)$, with
\[
\int_S\mu_0\,di=1,\qquad
\|\mu_0-\mu^{EQ}\|_{H^4}<\varepsilon
\]
for $\varepsilon>0$ sufficiently small. All neighborhoods of
densities below are taken in the affine space of $H^4(S)$
functions with unit integral. Smallness in $H^4$ ensures
strict positivity. Concavity and strict monotonicity of $u$
give $u'(X^{EQ})>0$.

Here $\mu$ and $v$ are the density and velocity in the planner's
problem~\eqref{eq:social-planner}. We first construct a sufficiently
regular pair $(\mu,v)$, close to $(\mu^{EQ},0)$, satisfying the
continuity equation with initial condition $\mu(0)=\mu_0$,
together with\footnote{This is, formally, the Euler equation of the planner. Indeed let $\psi$ be the current-value costate
associated with the continuity equation. After periodic
integration by parts, the current-value Hamiltonian is
\[
\mathcal H(\mu,v,\psi)
=u(X(\mu))
+\int_S\left(
\mu v\,\partial_i\psi-\frac{\chi}{2}\mu v^2
\right)di.
\]
Using
$D(u\circ X)(\mu)[h]
=X(\mu)u'(X(\mu))\int_S\pi(\mu)h\,di$,
stationarity with respect to $v$ and the costate equation
give
\[
\begin{array}{l}
\partial_i\psi=\chi v,\\
\partial_t\psi-r\psi+\frac{\chi}{2}v^2
+X(\mu)u'(X(\mu))\pi(\mu)=c(t),
\end{array}
\]
where the first identity has been substituted into
the second, and $c(t)$ is independent of the sector
and accounts for the unit-mass constraint.
Differentiating the second identity with respect to $i$
and using the first gives~\eqref{eq:planner-euler}.}
\begin{equation}
\label{eq:planner-euler}
\chi(\partial_t v+v\,\partial_i v-rv)
+X(\mu)u'(X(\mu))\partial_i\pi(\mu)=0.
\end{equation}
The proof proceeds in four steps. Step 1 constructs a solution
of the continuity and Euler equations near equilibrium. We express
the density and velocity in terms of a primitive of the density
perturbation, obtain an elliptic equation, and apply the inverse
function theorem. This also gives differentiable dependence on
the initial density.

Step 2 verifies that the constructed pair is the unique optimizer.
Concavity allows comparison with every admissible competitor,
and a uniform bound makes the discounted terminal term vanish.

Step 3 defines the optimal feedback as the initial velocity of
the optimal trajectory. Uniqueness then shows that this same
feedback describes the velocity along the entire trajectory.

Finally, Step 4 differentiates the state equation after
substitution of this feedback. The spectral decomposition
determines the temporal exponents, while integrability in time
excludes the growing modes. This gives the stated solution and
convergence of the linearized dynamics.

\medskip\noindent\emph{Step 1: construction of the solution.}

\noindent\emph{A primitive of the density perturbation.}
Write the density perturbation as
\begin{equation}
\label{eq:planner-primitive}
\partial_i z=\mu-\mu^{EQ},\qquad
\int_S\frac{z}{\mu^{EQ}}\,di=0.
\end{equation}
The initial trace is the normalized periodic primitive
of $\mu_0-\mu^{EQ}$ and belongs to $H^5(S)$. For a prescribed density path, the continuity equation
and~\eqref{eq:planner-primitive} imply
$\partial_i(\partial_t z+\mu v)=0$. Hence
\[
\mu v=c(t)-\partial_t z,
\]
where $c(t)$ is independent of the sector.
Different choices of $c(t)$ generate the same density
path but different kinetic costs. At each time, the
least costly choice therefore minimizes
\[
\frac{\chi}{2}\int_S
\frac{(c(t)-\partial_t z)^2}{\mu}\,di.
\]
This expression is strictly convex in $c(t)$, and its
first-order condition is
\[
\int_S\frac{c(t)-\partial_t z}{\mu}\,di
=\int_Sv\,di=0.
\]
Solving for $c(t)$ and substituting into the expression
for the flux gives
$\int_Sv\,di=0$ and
\[
v=\frac{1}{\mu}\left[
\frac{\displaystyle\int_S
       \frac{\partial_t z}{\mu}\,di}
     {\displaystyle\int_S\frac{1}{\mu}\,di}
-\partial_t z
\right].
\]
Thus $z$ determines both the density and the velocity.

Let $\mathcal P$ remove the weighted spatial mean:
\[
\mathcal P f
=f-\frac{\int_S f/\mu^{EQ}\,di}
          {\int_S1/\mu^{EQ}\,di}.
\]
Substituting the density and velocity determined by $z$
into~\eqref{eq:planner-euler} gives the equation
\begin{equation}
\label{eq:planner-elliptic}
-\mathcal P\!\left[
\mu(\partial_t v+v\,\partial_i v-rv)
+\frac{\mu X(\mu)u'(X(\mu))}{\chi}
\,\partial_i\pi(\mu)
\right]=0.
\end{equation}
This is equivalent to~\eqref{eq:planner-euler}:
the projected equation makes $\mu$ times the Euler
residual spatially constant, while the residual itself
has zero spatial integral because $\int_Sv\,di=0$.
Positivity of $\mu$ therefore forces that constant
to vanish.

\noindent\emph{Invertibility of the linearization.}
Let $\Omega=(0,\infty)\times S$.
A subscript $*$ on a Sobolev space denotes the condition
\[
\int_S\frac{f}{\mu^{EQ}}\,di=0,
\]
imposed at almost every time for functions on $\Omega$.

{
Set
\[
a=\frac{(1-\beta)X^{EQ}u'(X^{EQ})}
{\chi[\sigma(1-\beta)+\beta]}>0.
\]
}

The linearization at $z=0$ of the left-hand side
of~\eqref{eq:planner-elliptic}, together with the initial
trace, is the map\footnote{For the definition of Sobolev
spaces of fractional order, see
\citet[Sections~1.1 and~2.1.b]{CostabelDaugeNicaise2010}.}
\begin{equation}
\label{eq:planner-elliptic-isomorphism}
\begin{aligned}
H_*^{11/2}(\Omega)&\longrightarrow
H_*^{7/2}(\Omega)\times H_*^5(S),\\
z&\longmapsto
\left(
\partial_{tt}z-r\partial_t z
+{a}\mu^{EQ}\partial_i\!\left(
\frac{\partial_i z}{\mu^{EQ}}\right),
z(0)
\right).
\end{aligned}
\end{equation}
{
For $f\in L_*^2(\Omega)$, consider the closed Hilbert subspace
$\mathcal V=\{z\in H_*^1(\Omega):z(0)=0\}$ with the $H^1$ norm.
The weak formulation of $Lz=f$, where $L$ is the differential
operator in~\eqref{eq:planner-elliptic-isomorphism}, is
$\mathfrak b(z,\varphi)=\ell_f(\varphi)$ for all
$\varphi\in\mathcal V$, with
\[
\mathfrak b(z,\varphi)
=\int_\Omega\frac{z_t\varphi_t+a z_i\varphi_i+r z_t\varphi}
{\mu^{EQ}}\,dt\,di,
\qquad
\ell_f(\varphi)=-\int_\Omega\frac{f\varphi}{\mu^{EQ}}\,dt\,di.
\]
Both are continuous by Cauchy--Schwarz and the positive bounds
on $\mu^{EQ}$. Integration by parts and density of time-compact
functions with zero initial trace give
$\int_\Omega z_tz/\mu^{EQ}=0$.
The weighted spatial Poincar\'e inequality on the zero-mean
subspace therefore yields
\[
\mathfrak b(z,z)
=\int_\Omega\frac{|z_t|^2+a|z_i|^2}{\mu^{EQ}}\,dt\,di
\geq c\|z\|_{H^1(\Omega)}^2,\qquad c>0.
\]
Thus \citet[Corollary~5.8]{Brezis2011} gives a unique weak
solution $z\in\mathcal V$, with
$\|z\|_{H^1}\leq C\|f\|_{L^2}$.
}

The regularity and trace-lifting results of \citet[Theorem~3.4.1, Corollary~2.3.3(i), and
Remark~2.1.2]{CostabelDaugeNicaise2010}
then show that the displayed map is an isomorphism.
The local estimates apply uniformly on successive unit
strips because the coefficients are independent of time;
combining them with the coercive estimate gives the
required global bounds.

\noindent\emph{Construction and dependence on initial data.}
We locally invert the map assigning to $z$ the left-hand
side of~\eqref{eq:planner-elliptic} and the initial trace
$z(0)$. Our target is the pair consisting of zero and
the normalized primitive of $\mu_0-\mu^{EQ}$:
a preimage therefore solves both the nonlinear equation
and the prescribed initial condition.

For $z$ near zero in $H_*^{11/2}(\Omega)$, the density
is uniformly positive, and standard Sobolev product
and composition estimates
\citep[Section~3.1]{Taylor1991}
show that this map is $C^1$.
Its derivative at zero is the isomorphism~\eqref{eq:planner-elliptic-isomorphism}.
The inverse function theorem
\citep[Section~15.1]{Deimling1985}
therefore gives, for $\mu_0$ sufficiently close to
$\mu^{EQ}$, a unique small solution
$z\in H_*^{11/2}(\Omega)$ depending $C^1$ on $\mu_0$. Thanks to trace estimate in
\citet[Remark~2.1.2, p.~67]{CostabelDaugeNicaise2010},
applied to $z$ and $\partial_t z$ on the translated
half-cylinders $(t,\infty)\times S$, one also get
\begin{equation}
\label{eq:planner-uniform-bound}
\sup_{t\geq0}
\left(
\|\mu(t)-\mu^{EQ}\|_{H^4}
+\|v(t)\|_{H^4}
\right)
\leq C\|z\|_{H^{11/2}(\Omega)}
\leq C\|\mu_0-\mu^{EQ}\|_{H^4}.
\end{equation}

\medskip\noindent\emph{Step 2: verification of optimality.}
Let $(\mu,v)$ denote the pair constructed in Step~1
from the initial density $\mu_0$. We now show that
this pair is the unique optimizer of
\eqref{eq:social-planner}.

{
Writing $\pi=\pi(\mu)$, differentiation of $X$ and
\eqref{eq:profitti_dopo} gives
\[
\begin{aligned}
DX(\mu)[h]&=X(\mu)\int_S\pi h\,di,\\
D\pi(\mu)[h]&=-\frac{\pi}{\sigma(1-\beta)+\beta}
\left(\frac{h}{\mu}+(\sigma-1)\int_S\pi h\,di\right).
\end{aligned}
\]
Hence
\[
\frac{D^2X(\mu)[h,h]}{X(\mu)}
=\left(1-\frac{\sigma-1}{\sigma(1-\beta)+\beta}\right)
\left(\int_S\pi h\,di\right)^2
-\frac{1}{\sigma(1-\beta)+\beta}
\int_S\frac{\pi}{\mu}h^2\,di.
\]
Cauchy--Schwarz and $\int_S\pi\mu\,di=1-\beta$ imply
\[
\left(\int_S\pi h\,di\right)^2
\leq\left(\int_S\pi\mu\,di\right)
\left(\int_S\frac{\pi}{\mu}h^2\,di\right)
=(1-\beta)\int_S\frac{\pi}{\mu}h^2\,di.
\]
If the coefficient of the squared integral is nonnegative,
substitution gives the bound $-\beta\int_S(\pi/\mu)h^2\,di$
for $D^2X/X$; otherwise, discard that nonpositive term. Thus
\[
D^2X(\mu)[h,h]
\leq-X(\mu)\min\left\{\frac{1}{\sigma(1-\beta)+\beta},\beta\right\}
\int_S\frac{\pi}{\mu}h^2\,di<0
\qquad(h\ne0).
\]
Therefore $X$ is strictly concave, and so is $u\circ X$
because $u$ is concave and strictly increasing.
}

The kinetic cost can be written as
\[
\frac{\chi}{2}\int_S\mu v^2\,di
=\frac{\chi}{2}\int_S\frac{(\mu v)^2}{\mu}\,di.
\]
Since $(x,y)\mapsto x^2/y$ is convex for $y>0$
\citep[Example~3.18, p.~89]{BoydVandenberghe2004},
this cost is jointly convex in the density $\mu$ and
the flux $\mu v$. At fixed positive density, it is
strictly convex in the flux, being a quadratic function
with a positive coefficient.

Let $\psi$ be the periodic potential satisfying
$\partial_i\psi=\chi v$ and $\int_S\psi\,di=0$;
it exists because $\int_Sv\,di=0$.
Equation~\eqref{eq:planner-euler} implies that
\[
\partial_t\psi-r\psi+\frac{\chi}{2}v^2
+X(\mu)u'(X(\mu))\pi(\mu)
\]
is spatially constant.

Consider any other admissible pair $(\nu,w)$ with
$\nu(0)=\mu_0$. Concavity of $u\circ X$ gives
\[
u(X(\nu))-u(X(\mu))
\leq X(\mu)u'(X(\mu))
\int_S\pi(\mu)(\nu-\mu)\,di.
\]
Similarly, convexity of the kinetic cost in density
and flux gives
\[
\frac{\chi}{2}\int_S(\nu w^2-\mu v^2)\,di
\geq
-\frac{\chi}{2}\int_S v^2(\nu-\mu)\,di
+\chi\int_S v(\nu w-\mu v)\,di.
\]
These are the supporting inequalities at the constructed
pair $(\mu,v)$.

Subtracting the second inequality from the first,
and using the spatial constancy above,
$\int_S(\nu-\mu)\,di=0$ and $\partial_i\psi=\chi v$,
we obtain
\[
u(X(\nu))-u(X(\mu))
-\frac{\chi}{2}\int_S(\nu w^2-\mu v^2)\,di
\leq-\int_S\left[
(\partial_t\psi-r\psi)(\nu-\mu)
+\partial_i\psi\,(\nu w-\mu v)
\right]di.
\]
The continuity equations imply
\[
\frac{d}{dt}\left(
e^{-rt}\int_S\psi(\nu-\mu)\,di
\right)
=e^{-rt}\int_S\left[
(\partial_t\psi-r\psi)(\nu-\mu)
+\partial_i\psi\,(\nu w-\mu v)
\right]di.
\]
Multiplying the preceding inequality by $e^{-rt}$,
integrating over $[0,T]$, and using the common initial
density therefore yields
\[
\int_0^T e^{-rt}\left[
u(X(\nu))-u(X(\mu))
-\frac{\chi}{2}\int_S(\nu w^2-\mu v^2)\,di
\right]dt
\leq-e^{-rT}
\int_S\psi(T)(\nu(T)-\mu(T))\,di.
\]
The absolute value of the right-hand side is at most
$2e^{-rT}\|\psi(T)\|_\infty$, which tends to zero by
\eqref{eq:planner-uniform-bound}.
Letting $T\to\infty$ proves optimality.
Strict concavity in density and strict convexity of
the kinetic cost in flux at fixed density give uniqueness.

\medskip\noindent\emph{Step 3: the optimal feedback.}
Define $K(\mu_0)$ as the initial velocity of the constructed
optimal pair. Differentiable dependence on the initial density
and the bounded trace map imply that $K$ is $C^1$ into $H^4(S)$,
with $K(\mu^{EQ})=0$.
For initial densities in a smaller neighborhood,
\eqref{eq:planner-uniform-bound} keeps the trajectory in the
domain of $K$. Every time translate is optimal for its current
initial density by the same verification argument.
Uniqueness therefore gives
\[
v(t)=K(\mu(t)).
\]

\medskip\noindent\emph{Step 4: linearization and convergence.}
Differentiate the constructed trajectory with initial condition
$\mu^{EQ}+\epsilon h$, where $h\in H^4(S)$ and $\int_S h\,di=0$.
Denote its derivative at $\epsilon=0$ by $(\eta,w)$.
Using
\[
DX(\mu^{EQ})[h]=0,\qquad
D\pi(\mu^{EQ})[h]
=-\frac{1-\beta}{\rho}\frac{h}{\mu^{EQ}},
\]
we obtain
\[
\begin{cases}
\partial_t\eta=-\partial_i(\mu^{EQ}w),\\
\partial_t w-rw-a\partial_i(\eta/\mu^{EQ})=0,\\
\eta(0)=h,
\end{cases}
\qquad w(t)=DK(\mu^{EQ})[\eta(t)].
\]
Eliminating $w$ gives
\[
\partial_{tt}\eta-r\partial_t\eta-a\mathcal L\eta=0.
\]
For $k\geq2$, the coefficient along $\phi_k$ has characteristic
roots
\[
\lambda_{P,k}=\frac{r-\sqrt{r^2+4a\kappa_k}}2<0,
\qquad r-\lambda_{P,k}>0.
\]
The derivative of $z$ belongs to $H_*^{11/2}(\Omega)$, so
$\eta\in L^2((0,\infty);H)$.
This excludes the growing exponential in every mode and yields
\[
\eta(t)=\sum_{k\geq2}\langle h,\phi_k\rangle_H
e^{\lambda_{P,k}t}\phi_k.
\]

Evaluating the first variational equation at $t=0$ therefore gives
\[
\left\langle
-\partial_i\!\left(\mu^{EQ}DK(\mu^{EQ})[h]\right),\phi_k
\right\rangle_H
=\lambda_{P,k}\langle h,\phi_k\rangle_H,
\qquad k\geq2.
\]
Thus the linearized state equation after substitution of the
feedback is diagonal in $(\phi_k)$ and has the unique solution
given by the preceding series. Taking $h=\mu_0-\mu^{EQ}$ proves
\eqref{eq:local-planner}.
Parseval's identity \citep[Corollary~5.10]{Brezis2011} gives uniform $L^2(S)$ convergence of the
series for $t\geq0$.
Since $\lambda_{P,k}<0$ decreases strictly with $\kappa_k$,
convergence to equilibrium and the stated generic asymptotic follow.
\end{proof}

\begin{proof}[Proof of Proposition~\ref{pr:planner-comparison}]
For generic common initial densities, the eigenspace associated
with $\kappa_2$ determines both asymptotic rates. For this mode, the planner's characteristic equation is
\[
\lambda^2-r\lambda
-\frac{(1-\beta)X^{EQ}u'(X^{EQ})}{\chi\rho}\kappa_2=0,
\]
whose negative root is $\lambda_P$. Moreover, the decentralized
rate satisfies
\[
\lambda_D=-\frac{(1-\beta)\kappa_2}{\chi\rho}.
\]
Substituting this expression into the characteristic equation
at $\lambda=\lambda_P$ gives
\[
(-\lambda_P)^2+r(-\lambda_P)
=X^{EQ}u'(X^{EQ})(-\lambda_D).
\]
Since $s\mapsto s^2+rs$ is strictly increasing for $s\geq0$,
comparison with its value at $s=-\lambda_D>0$ yields
\[
-\lambda_P\gtreqless-\lambda_D
\quad\Longleftrightarrow\quad
X^{EQ}u'(X^{EQ})\gtreqless r-\lambda_D.
\]
Substituting the expression for $\lambda_D$ proves
\eqref{eq:planner-comparison}.
\end{proof}

\begin{proof}[Proof of Corollary~\ref{cor:planner-constant-A}]
For $A\equiv A_0$, we have
$\mu^{EQ}=1/n$ and $X^{EQ}=A_0n^{1/(\sigma-1)}$.
Consequently, $\mathcal L=-\partial_{ii}$ on the circle of length
$n$, and $\kappa_2=\kappa_3=4\pi^2/n^2$.
CRRA utility is smooth, strictly increasing and concave on
$(0,\infty)$ for every $\vartheta\geq0$, including the logarithmic
case, and
\[
X^{EQ}u'(X^{EQ})
=A_0^{1-\vartheta}n^{(1-\vartheta)/(\sigma-1)}.
\]
Substitution into~\eqref{eq:planner-comparison} proves
\eqref{eq:planner-crra}. The decentralized rate is independent of $A_0$, while $|\lambda_P|$ is strictly increasing in
$X^{EQ}u'(X^{EQ}) =A_0^{1-\vartheta}n^{(1-\vartheta)/(\sigma-1)}$.
Thus higher productivity accelerates the planner's
linearized adjustment when $\vartheta<1$, leaves it
unchanged when $\vartheta=1$, and slows it when
$\vartheta>1$.
\end{proof}

\endgroup

\section{An exact control representation}
\label{app:control-representation}

This appendix provides an exact intertemporal control representation
of the decentralized dynamics, valid beyond the local approximation
around the long-run equilibrium. The objective combines a benefit
from aggregate consumption, the cost of reallocation, and a penalty
on sectoral profit gradients. This additional penalty distinguishes
the resulting optimization problem from the social planner's
welfare problem studied in Section~\ref{sec:planner}.

\begin{Proposition}[An exact control representation]
\label{pr:control-representation}
Assume the hypotheses of Theorem~\ref{th:esistenzaeunicita}
and fix $r>0$. Admissible pairs $(\mu,v)$ consist of strictly
positive probability densities
\[
\mu\in C([0,+\infty);C^2(S))
       \cap C^1([0,+\infty);C(S))
\]
and velocities $v\in C([0,+\infty);C^1(S))$ satisfying
\[
\partial_t\mu+\partial_i(\mu v)=0,\qquad
\mu(0)=\mu_0,\qquad
\lim_{t\to+\infty}e^{-rt}\log X(\mu(t))=0.
\]
Then
\begin{equation}
\label{eq:inverse-control}
\sup_{(\mu,v)\,\mathrm{admissible}}
\int_0^\infty e^{-rt}
\left[
r\log X(\mu(t))
-\frac{1}{2\chi}\int_S
 |\partial_i\pi(\mu(t))|^2\mu(t,i)\,di
-\frac{\chi}{2}\int_S v(t,i)^2\mu(t,i)\,di
\right]dt
=\log X(\mu_0).
\end{equation}
The payoff is understood as the limit of its finite-horizon
integrals, possibly equal to $-\infty$.
The unique optimal density path is the decentralized solution
of~\eqref{eq:stato-di-nuovo}, with optimal velocity
\[
v(t,i)=\frac{1}{\chi}\partial_i\pi(\mu(t))(i).
\]
\end{Proposition}

\begin{proof}
Differentiating $\mathcal F=\log X$ gives
$D\mathcal F(\mu)[\eta]=\int_S\pi(\mu)\eta\,di$.
The continuity equation and periodic integration by parts
therefore imply
\[
\frac{d}{dt}\mathcal F(\mu(t))
=\int_S\partial_i\pi(\mu(t))\,v(t,i)\mu(t,i)\,di.
\]
Writing $J_T(\mu,v)$ for the objective
in~\eqref{eq:inverse-control} truncated at $T$, integration
by parts in time and completion of the square give
\[
J_T(\mu,v)
=\mathcal F(\mu_0)-e^{-rT}\mathcal F(\mu(T))
-\frac{\chi}{2}\int_0^T e^{-rt}\int_S
\left|v-\frac{\partial_i\pi(\mu)}{\chi}\right|^2
\mu\,di\,dt.
\]
The terminal condition thus yields
\[
\lim_{T\to+\infty}J_T(\mu,v)
=\mathcal F(\mu_0)
-\frac{\chi}{2}\int_0^\infty e^{-rt}\int_S
\left|v-\frac{\partial_i\pi(\mu)}{\chi}\right|^2
\mu\,di\,dt
\leq\mathcal F(\mu_0).
\]

Along the decentralized solution,
\[
\frac{d}{dt}\mathcal F(\mu(t))
=\frac{1}{\chi}\int_S
 |\partial_i\pi(\mu(t))|^2\mu(t,i)\,di\geq0.
\]
Moreover, Theorem~\ref{th:mu-e-ottimale} gives
$\mathcal F(\mu(t))\leq\mathcal F(\mu^{EQ})$.
Hence $\mathcal F(\mu(t))$ is bounded and the terminal
condition holds. This path is admissible and attains
the upper bound because the square vanishes.

Conversely, equality requires
$v=\partial_i\pi(\mu)/\chi$, so the continuity equation
reduces to~\eqref{eq:stato-di-nuovo}.
Uniqueness follows from Theorem~\ref{th:esistenzaeunicita}.
\end{proof}


\setcounter{section}{2}

\section{Extension: Intrasectoral externalities}
\label{sub:spillovers}

In this section, we consider the presence of intrasectoral technological externalities. In particular, sectoral productivity $A(i,\mu(i))$ takes the form
\begin{equation}
\label{eq:A-spillovers}
A(i,\mu(i)) = A_0(i) \mu(i)^\eta
\end{equation}
for some real $\eta$ and an exogenous function $A_0$ depending only on $i$. At the sectoral level, $\eta>0$ in Equation \eqref{eq:A-spillovers} implies that a firm’s productivity (and profitability) benefits from the mass of firms operating within its sector. This is consistent with models of knowledge spillovers in production, as in \citet{Romer1986} (see also \citealp{Jovanovic1989,Griliches1992,Romer1990a}).
Applied in a spatial setting, this mechanism underpins the emergence of industrial districts (Marshallian externalities) in the new economic geography literature (see, e.g., \citealp{allen2014trade}).
Conversely, $\eta<0$ implies that a greater sectoral mass of firms negatively affects individual productivity, which is generally explained by the presence of congestion effects and/or negative pecuniary externalities at the sectoral level \citep{fujita_thisse_2002}.
At the aggregate level, $\eta>0$ (resp. $\eta<0$) implies increasing (resp. decreasing) returns to scale in production.

The short-run competitive equilibrium described by Proposition \ref{pr:shor-run-equilibrium} remains the same, except that $A(i)$ is replaced by ${A}_0(i) \mu(i)^\eta$ (Subsection~\ref{app:casospillovers} contains the technical details); 
profits are given by:
\begin{equation}
\pi(i) = \left(1 - \beta \right) \frac{ A_0(i)^{ \frac{\sigma-1}{\sigma(1-\beta) + \beta} } \mu(i)^{ \frac{ \eta(\sigma-1) -1 }{\sigma (1-\beta) + \beta} } }{ \displaystyle \int_0^n \left[ A_0(j) \mu(j) ^{1-\beta + \eta} \right]^{ \frac{\sigma - 1}{\sigma(1-\beta) + \beta} } dj }
\label{eq:profitsPresenceExternalities}
\end{equation}
while production and employment are therefore given by: 
\begin{equation}
\label{eq:productionEmploymentExternalities}
Y(i) \equiv L(i) = \frac{\left [ A_0(i)  \mu(i)^{1-\beta +\eta} \right ] ^\frac{\sigma-1}{\sigma\left(1-\beta\right)+\beta} }{ \displaystyle \int_0^n \left[ A_0(j) \mu(j) ^{1-\beta + \eta} \right]^{ \frac{\sigma - 1}{\sigma(1-\beta) + \beta} } dj  };
\end{equation}
with labour productivity equal to one in every sector. Externalities either exacerbate or smooth production/employment differentials across sectors, depending on their sign, that is, the sign of $\eta$, and on the intensity of goods substitutability, i.e. whether $\sigma$ is greater or less than 1.

In the dynamics of these extensions we use time units in which the quadratic reallocation-cost coefficient is $\chi=1$. The PDE describing the dynamics of firms' distribution is then:
\begin{equation}
\label{eq:PDEintrasectoralExternalities}
\begin{cases}
\partial_t\mu(t,i)=-\partial_i\bigl(\mu(t,i)\partial_i\pi(t,i)\bigr),
 &(t,i)\in\mathbb R^+\times S,\\
\mu(t,0)=\mu(t,n),&t\geq0,\\
\mu(0,i)=\mu_0(i),&i\in S,
\end{cases}
\end{equation}
where $\pi(t,i)$ is given by Equation~\eqref{eq:profitsPresenceExternalities} evaluated at $\mu(t,\cdot)$.

In the presence of sectoral externalities, in order to obtain results on the existence, uniqueness, and boundedness of the firms’ distribution as in Theorem \ref{th:esistenzaeunicita} and, consequently, to express PDE \eqref{eq:PDEintrasectoralExternalities} as a gradient flow, the exponent of $\mu(i)$ must be negative, i.e.:
\begin{equation}
\label{eq:condizione-spillovers}
\eta \left(\sigma -1\right) -1 < 0;
\end{equation}
Therefore, if $\sigma>1$, Condition \eqref{eq:condizione-spillovers} requires that $\eta \in \left(-\infty,1/(\sigma-1)\right)$; that is, there exists an upper bound on the intensity of positive spatial spillovers. Conversely, if $\sigma \in (0,1)$, then $\eta \in \left(-1/(1-\sigma),\infty\right)$; in this case, there is an upper bound on the strength of negative externalities. Intuitively, in economies with a high (respectively, low) elasticity of substitution between goods, i.e. $\sigma \gg 1$ (respectively, $\sigma$ close to 0), and strong, positive (respectively, negative) externalities, i.e. $\eta \gg 0$ (respectively, $\eta << 0$), the inflow of firms attracted by higher profits reinforces existing differentials, leading to an increase rather than a decrease in profit dispersion (see Equation \eqref{eq:profitsPresenceExternalities}).
If the inequality is reversed, the principal diffusion term changes sign and the forward-parabolic well-posedness arguments no longer apply. This sign change alone is not a proof that every initial condition develops a finite-time singularity.

Assume Condition~\eqref{eq:condizione-spillovers} and first let $1-\beta+\eta\ne0$. The functional $\mathcal{F}(\mu)$ is 
\begin{multline}
\label{eq:espressionediF-coneta}
\mathcal{F}(\mu)   
= \left( \frac{1-\beta}{1 - \beta + \eta} \right) \log \left( \left\{ \int_0^n \left[ {A}_0(j) \mu(t,j)^{1-\beta+\eta} \right]^{ \frac{\sigma- 1}{\sigma(1-\beta) + \beta} } dj \right\}^{\frac{\sigma(1-\beta) + \beta}{\sigma- 1}} \right) = \\
= \left( \frac{1-\beta}{1 - \beta + \eta} \right) \log X (\mu(t)),
\end{multline}
while the rest of Theorem \ref{th:riscritturaGradientFlow} still holds. In the same respect, we can extend Theorem \ref{teo:convergence} to the presence of intrasectoral externalities, where now:
\begin{equation}
\label{eq:longRunEQintraSpillovers}
\mu^{EQ}(i) = Y^{EQ}(i) = L^{EQ}(i) =\frac{A_0(i)^{\frac{\sigma - 1}{1 - \eta (\sigma - 1)}
}}{\int_0^n {A}_0(j)^{\frac{\sigma - 1}{1 - \eta (\sigma - 1)}
} \, dj}
\end{equation}
and
\begin{equation}
X^{EQ} = \left( \int_0^n {A}_0(j)^{\frac{\sigma - 1}{1 - \eta (\sigma - 1)}
} \, dj \right)^ {\frac{1- \eta(\sigma - 1)}{ \sigma -1}}.
\label{longRunConsumptionExternality}
\end{equation}
Equations \eqref{eq:longRunEQintraSpillovers} and \eqref{longRunConsumptionExternality} illustrate how the presence of externalities affects both the long-run distribution of factors across sectors and the equilibrium levels of consumption. The comparative effect of $\eta$ depends on the elasticity of substitution and the productivity profile; these formulas do not imply a universal increase in consumption with $\eta$. More importantly, under Condition \eqref{eq:condizione-spillovers}, the gradient flow associated with the functional (\ref{eq:espressionediF-coneta}) highlights a result that is not standard in the literature. Specifically, if $\eta>\beta -1$, that is, if externalities are positive or at least not excessively negative, $X(t)$ increases along the trajectory toward the unique long-run equilibrium $\mu^{EQ}$ defined by Equation \eqref{eq:longRunEQintraSpillovers}. Furthermore, it can be shown that $\mu^{EQ}$ maximizes $X(t)$ among all feasible distributions, i.e. \textit{First Welfare Theorem} holds. In other words, the decentralized choices made by firms and workers enable the economy to achieve maximum efficiency in the long run, even in the presence of positive or sufficiently small negative externalities. This finding contrasts with the conventional view in the literature, according to which the presence of externalities undermines the validity of the First Welfare Theorem (see, e.g., \citealp[Ch. 11]{mas1995microeconomic}). This result crucially depends on the firms' simple decision rule, which implicitly allows them to ``internalize'' externalities by using the observed profit rate as their primary decision-making variable. On the opposite, if $\eta < \beta -1 < 0$, that is, if externalities are sufficiently negative, $X(t)$ decreases along the trajectory toward the unique long-run equilibrium and therefore $\mu^{EQ}$ and the corresponding $X^{EQ}$ are not socially optimal. The decentralized equilibrium in this scenario is therefore not socially optimal. The key intuition behind these opposite results is the relationship between $\mu(i)$ and $Y(i)$, which crucially depends on the sign of $1-\beta+\eta$ (Equation \eqref{eq:productionEmploymentExternalities}).

\subsection{Technical results and proofs}
\label{app:casospillovers}

This section contains Proposition \ref{pr:shor-run-equilibrium-spillovers}, which collects the properties of the short-run equilibrium in presence of intrasectoral externalities. The equilibrium is similar to the case without externalities, except that \( A(i) \) is replaced by \( A_0(i) \mu(i)^\eta \).
\begin{Proposition}[The short-run equilibrium with externalities]
\label{pr:shor-run-equilibrium-spillovers}
In the short-run equilibrium with intrasectoral externalities:
\[
w = \beta,
\]
\begin{equation}
\displaystyle
p(i) = \frac{\left[   A_0(i) \mu(i)^{1-\beta+\eta}\right]^{-\frac{1}{\sigma(1-\beta)+\beta}}}{\displaystyle \left( \int_0^n \left[   A_0(j) \mu(j)^{1-\beta+\eta}\right]^{\frac{1}{1-\beta}\left(1-\frac{1}{\sigma(1-\beta)+\beta}\right)} dj \right)^{(1-\beta)}},
\end{equation}
\begin{equation}
P = \left( \int_0^n \left[ A_0(i) \, \mu(i)^{1 - \beta + \eta} \right]^{ \frac{\sigma -1 }{\sigma(1 - \beta) + \beta} } \, di \right)^{ \frac{ \sigma(1 - \beta) + \beta }{ 1 - \sigma } },    \end{equation}
\begin{equation}
\pi(i) = \left(1 - \beta \right) \frac{ A_0(i)^{ \frac{\sigma-1}{\sigma(1-\beta) + \beta} } \mu(i)^{ \frac{ \eta(\sigma-1) -1 }{\sigma (1-\beta) + \beta} } }{ \displaystyle \int_0^n \left[ A_0(j) \mu(j) ^{1-\beta + \eta} \right]^{ \frac{\sigma - 1}{\sigma(1-\beta) + \beta} } dj }
\end{equation}
the production in sector $i$ is 
\begin{equation}
Q(i) = \frac{\left [ A_0(i)  \mu(i)^{1-\beta +\eta} \right ] ^\frac{\sigma}{\sigma\left(1-\beta\right)+\beta} }{\left ( \displaystyle \int_0^n \left[ A_0(j) \mu(j) ^{1-\beta + \eta} \right]^{ \frac{\sigma - 1}{\sigma(1-\beta) + \beta} } dj  \right )^{\beta}},
\end{equation}
while the value of production and employment
\begin{equation}
Y(i) \equiv L(i) = \frac{\left [ A_0(i)  \mu(i)^{1-\beta +\eta} \right ] ^\frac{\sigma-1}{\sigma\left(1-\beta\right)+\beta}}{ \displaystyle \int_0^n \left[ A_0(j) \mu(j) ^{1-\beta + \eta} \right]^{ \frac{\sigma - 1}{\sigma(1-\beta) + \beta} } dj  }.
\end{equation}
\end{Proposition}
\begin{proof}
The proof follows directly from Proposition \ref{pr:shor-run-equilibrium} by using that $A(i) = A_0(i) \mu(i)^\eta$.
\end{proof}

As before $\Pi = (1-\beta)$.

\begin{Lemma}[Power-profit dynamics]
\label{lem:appendix-power-profit}
Let $a\in C^2(S)$ be strictly positive, let $m<1$, and let the initial density be strictly positive, of unit mass, and in $C^2(S)$. Define
\[
J(\mu):=\int_S a(i)\mu(i)^m\,di,
\qquad
\pi(\mu,i):=\frac{(1-\beta)a(i)\mu(i)^{m-1}}{J(\mu)}.
\]
Then the periodic equation $\partial_t\mu=-\partial_i(\mu\partial_i\pi)$ has the same global positive-solution properties as the baseline parabolic equation. Its stationary unit-mass density is
\[
\omega(i)=\frac{a(i)^{1/(1-m)}}{\int_S a(j)^{1/(1-m)}\,dj},
\]
and the solution converges to $\omega$. Its increasing potential is
\begin{equation}
\label{eq:appendix-power-potential}
\mathcal F(\mu)=
\begin{cases}
\displaystyle\frac{1-\beta}{m}\log J(\mu),&m\ne0,\\[5pt]
\displaystyle\frac{1-\beta}{\int_S a(j)\,dj}
                   \int_S a(i)\log\mu(i)\,di,&m=0.
\end{cases}
\end{equation}
The density $\omega$ uniquely maximizes this potential among positive unit-mass densities for which it is finite.
\end{Lemma}
\begin{proof}
Direct differentiation, including the $m=0$ case, gives
$\delta\mathcal F/\delta\mu=\pi$. Thus the equation is the Wasserstein gradient ascent of $\mathcal F$, with the sign convention of the main text, and
\[
\frac{d}{dt}\mathcal F(\mu(t))
 =\int_S\mu(t,i)|\partial_i\pi(t,i)|^2\,di\geq0.
\]
To check the bounds needed for global continuation, put $h=\mu/\omega$ and
$H_m(t)=\int_S\omega(i)h(t,i)^m\,di$. Substitution yields
\begin{equation}
\label{eq:appendix-common-flow}
\omega\partial_t h
 =\frac{(1-\beta)(1-m)}{H_m(t)}
   \partial_i\bigl(\omega h^{m-1}\partial_i h\bigr).
\end{equation}
The maximum principle bounds $h$ between the strictly positive minimum and the finite maximum of $h(0,\cdot)$. Hence $H_m$ and all diffusion coefficients stay bounded away from zero and infinity. These bounds allow the local parabolic solution to be continued uniquely for all time, by the same continuation argument used for Theorem~\ref{th:esistenzaeunicita}.

For convergence, mass conservation gives $\int_S\omega h=1$, and periodic integration by parts gives
\[
\frac{d}{dt}\frac12\int_S\omega(h-1)^2\,di
 =-\frac{(1-\beta)(1-m)}{H_m(t)}
       \int_S\omega h^{m-1}|\partial_i h|^2\,di.
\]
The uniform positive bounds and the weighted Poincare inequality imply exponential decay of this weighted $L^2$ distance; parabolic regularity gives convergence of the density to $\omega$.
Finally, for $m\ne0$,
\[
\mathcal F(\mu)-\mathcal F(\omega)
 =\frac{1-\beta}{m}\log\!\left(\int_S\omega h^m\,di\right)\leq0.
\]
For $0<m<1$ this follows from concavity, and for $m<0$ from convexity and division by a negative number. At $m=0$, Jensen's inequality gives
$(1-\beta)\int_S\omega\log h\,di\leq0$.
Strictness in each case identifies $h=1$ as the unique maximizer.
\end{proof}

For the extension in this appendix, apply Lemma~\ref{lem:appendix-power-profit} with
\[
a(i)=A_0(i)^{(\sigma-1)/\rho},\qquad
m=\frac{(1-\beta+\eta)(\sigma-1)}{\rho},\qquad
\rho=\sigma(1-\beta)+\beta.
\]
The condition $m<1$ is exactly Equation~\eqref{eq:condizione-spillovers}. This proves the gradient-flow and long-run assertions above. When $1-\beta+\eta=0$, Equation~\eqref{eq:espressionediF-coneta} is replaced by the $m=0$ branch of Equation~\eqref{eq:appendix-power-potential}; consumption is then independent of $\mu$, although firm reallocation need not cease. For $1-\beta+\eta>0$, maximizing $\mathcal F$ maximizes consumption; for $1-\beta+\eta<0$, it minimizes consumption.


\section{Extension: Fixed costs of firm reallocation}
\label{sub:fixedcost}

It is plausible that in the real world, a portion of firms' reallocation costs is independent of the ``intensity'' of reallocation for example due to significant barriers or indivisibilities in sector entry or exit.\footnote{"\cite[p. 385]{eisfeldt2006capital} argue that the empirical observation, namely, that in any given year a substantial share of U.S. firms engage in no capital reallocation, can be explained by the presence of fixed adjustment costs.}
In presence of a fixed cost of reallocation $c_0 > 0$, the problem of firm $i$ can be formulated as:
\[
\max_{v}\left ( \partial_{i}\pi(t,i) \right) v - c(v),
\]  
but, unlike in Section \ref{sub:dynamics}, the cost function for reallocation appears as:  
\begin{equation}
\label{eq:reallocationCostFixedCost}
    c(v) = \left \{
\begin{array}{l}
0 \qquad \text{if }v = 0\\
\frac{1}{2} v^2 + c_0 \qquad \text{if }v \neq 0\\
\end{array}
\right.,
\end{equation}

and therefore the optimal choice is:  
\[
v(t,i) = 
\left \{
\begin{array}{l}
0 \qquad \text{if $\left | \partial_{i}\pi(t,i) \right | \leq \sqrt{2c_0}$}\\
\partial_{i}\pi(t,i)
\qquad \text{otherwise.}\\
\end{array}
\right.
\]  
Consequently, the PDE describing the dynamics of firms reallocation is: 
\begin{align} 
    \label{eq:PDEFixedCostreallocation}
\partial_{t}\mu(t,i) &= 
\left\{
\begin{array}{ll}
0 & \text{ if } \left | \partial_{i}\pi(t,i) \right | \leq \sqrt{2c_0}\\
-(1-\beta) \left[\dfrac{\partial_i \left( \mu(t,i) \partial_i \left( A(i)^{\frac{\sigma -1}{\sigma(1-\beta) + \beta}} \mu(t,i)^{-\frac{1}{\sigma(1-\beta) + \beta}} \right) \right)}{\displaystyle  \int_0^n \left[ A(i) \mu(i)^{1-\beta} \right]^{\frac{\sigma -1}{\sigma(1-\beta) + \beta}} \, di } \right] & \text{otherwise}\\
\end{array}
\right.
\end{align}
The presence of a ``discontinuity'' in the PDE impedes to have a gradient flow representation (and therefore also the JKO scheme) and therefore to reach general conclusions. 
However, our intuition is that multiple (typically infinitely many) long-run competitive equilibria exist when the profit distribution is sufficiently close to the long-run equilibrium of the problem without fixed costs. Yet, when fixed costs are large enough, there may exist long-run equilibria whose pointwise distance from the equilibrium without fixed costs is not bounded, as the following example illustrates.

As an example, consider the simplest economy where $A(i) = \gamma(i) =1$ for all $i$, $\eta=0$ and $n=1$; if $c_0=0$, Theorem \ref{teo:convergence} states that the only equilibrium firms' distribution (also asymptotically stable) for the problem without fixed cost is uniform. Consider instead the case with $c_0>0$ and the following firms' distribution: 
\begin{equation}
\label{eq:distribution-equilibrium-example-2}
\mu(i) := 
\begin{cases} 
\frac{C_{\alpha}}{i^{\alpha}} & \text{for } i \in (0, 1/2] \\
\frac{C_{\alpha}}{(1-i)^{\alpha}} & \text{for } i \in (1/2, 1),\\
\end{cases}
\end{equation}
where $\alpha \in (0,1)$ and $C_\alpha = (1-\alpha)/2^\alpha$ in order to have $\int_0^1 \mu(j) dj =1$. It diverges at $i=0$ (which is identified with $i=1$). In Subsection~\ref{app:proofsfixedcosts} we prove that if $\sigma(1-\beta) + \beta<\alpha$, and the fixed cost is sufficiently high, i.e. 
\begin{equation}
\label{eq:condizione-su-c0}
c_0 \geq \left(\dfrac{1}{2}\right) \left\{ \frac{\alpha(1-\beta)(1-\alpha)^{\frac{(\sigma-1)(1-\beta)-1}{\sigma(1-\beta) + \beta}}}{2^{\alpha-1+\frac{\alpha}{\sigma(1-\beta) + \beta}}[(1-\alpha)(\sigma(1-\beta) + \beta)+ \alpha]} \right\}^{2}
\end{equation}
then \eqref{eq:distribution-equilibrium-example-2} is a long-run competitive equilibrium distinct from $\mu^{EQ}(i)$. The non-uniformity of this distribution relative to $\mu^{EQ}(i)$ directly implies inequality in profit rates, causing the equilibrium to be inefficient. The magnitude of this inefficiency is directly related to the heterogeneity of profit rates, i.e. to $\alpha$.

However, a minor role of fixed costs could be explained by the potentially restrictive assumption of comparing the fixed cost, paid once and for all,  with instantaneous profits rather than with the expected stream of future profits generated by operating in the sector. Accounting for future profits introduces an intertemporal dimension to firms' decision-making that is difficult to incorporate explicitly into our framework. However, a plausible justification for using current profits as a reference point is that they may serve as the best available proxy for expected future profit flows.
Specifically, consider a discrete-time model with static expectations for profits, i.e. $\pi^a(t,i)=\pi(0,i) \; \forall t >0$. Under this assumption, the expected profits in sector $i$ are given by $\sum_{t=0}^{\infty} (1-r)^t \pi^a(t,i)= \pi(0,i)/r$, where $r$ is the intertemporal discount (real interest) rate.
A distribution of profits is stable if, for every $i$ and $j$,
\[
 \left| \frac{\pi(0,i)}{r} - \frac{\pi(0,j)}{r} \right| < |i-j|^2 +c_0,
\]
which is implied by the stronger constraint:
\[
\max_{i,j} \left| \frac{\pi(0,i)}{r} - \frac{\pi(0,j)}{r} \right| < c_0,
\]
i.e.
\[
\max_{i,j} \left| \pi(0,i) - \pi(0,j) \right| < r c_0.
\]
Therefore, from an intertemporal perspective, the fixed reallocation cost $c_0$ should be scaled by the discount rate $r$ to enable comparison with actual profits, highlighting that fixed costs might be even less relevant for firms' reallocation decisions than initially suggested. Specifically, recalling that $c_0$ represents the fixed cost of moving one unit of capital between sectors, a value of $c_0=0.2$ corresponds to a reallocation cost equivalent to 20\% in terms of instantaneous profit rate differences (i.e., the alternative sector must offer profit rates at least 20\% higher). However, once adjusted by a discount rate $r=10\%$, the effective threshold considered by firms becomes just $0.2 \times 0.1=2\%$. Consequently, profit rate differences smaller than 2\% would be insufficient to incentivize capital reallocation, implying that, in long-run equilibrium, profit rates across sectors could differ by up to 2\%.
To conclude, the empirical evidence discussed in Section~\ref{sec:converegenceROE} suggests that fixed costs play, at most, a minor role in the overall reallocation process, at least at the sectoral level.

\subsection{Proof of the fixed-cost equilibrium}
\label{app:proofsfixedcosts}

\begin{proof}[Proof of the fixed-cost equilibrium condition]
As already mentioned, to ensure that the distribution has a total mass equal to one, we need to set
\begin{equation*}
C_{\alpha} = \frac{1-\alpha}{2^{\alpha}}.
\end{equation*}  
We also compute  
\begin{equation}
\label{eq:distribution-equilibrium-example-3}
\partial_{i} \mu(i) = 
\begin{cases} 
-\frac{\alpha (1-\alpha)}{2^{\alpha}} \frac{1}{i^{\alpha+1}} & \text{for } i \in (0, 1/2), \\
+\frac{\alpha (1-\alpha)}{2^{\alpha}} \frac{1}{(1-i)^{\alpha+1}} & \text{for } i \in (1/2, 1),
\end{cases}
\end{equation}
The expression for profits, defining $\rho = \sigma(1-\beta)+\beta$, is given by
\begin{multline*}
\pi(i) = (1-\beta) \left( \int_0^n \left[ A(i)^{1/(1-\beta)} \mu(i) \right]^{1-\frac{1}{\sigma+\beta-\sigma\beta}} di \right) A(i)^{\frac{1}{1-\beta} \left(1-\frac{1}{\sigma+\beta-\sigma\beta} \right)} \mu(i)^{-1/(\sigma(1-\beta)+\beta)} \\
= (1-\beta) \int_{0}^{1} \mu(i)^{\frac{\rho-1}{\rho}} di \, \mu(i)^{-\frac{1}{\rho}},
\end{multline*}  
which leads to  
\begin{equation}
\label{eq:profit-derivative-example}
\partial_{i} \pi(i) = (1-\beta) \int_{0}^{1} \mu(i)^{\frac{\rho-1}{\rho}} di \left(-\frac{1}{\rho} \mu(i)^{-\frac{1}{\rho}-1}\right) \partial_{i} \mu(i).
\end{equation}  
By symmetry,  
\begin{multline*}
\int_{0}^{1} \mu(i)^{\frac{\rho-1}{\rho}} di = 2 \int_{0}^{1/2} \mu(i)^{\frac{\rho-1}{\rho}} di = 2 \left(\frac{1-\alpha}{2^{\alpha}}\right)^{\frac{\rho-1}{\rho}} \int_{0}^{1/2} \frac{1}{i^{\alpha(\frac{\rho-1}{\rho})}} di \\
= 2 \left(\frac{1-\alpha}{2^{\alpha}}\right)^{\frac{\rho-1}{\rho}} \frac{2^{\alpha(\frac{\rho-1}{\rho})-1}}{1-\alpha(\frac{\rho-1}{\rho})} = \frac{\rho(1-\alpha)^{\frac{\rho-1}{\rho}}}{\alpha(1-\rho)+\rho}.
\end{multline*}  
The expression for the derivative of profits \eqref{eq:profit-derivative-example} then becomes  
\begin{equation*}
\partial_{i} \pi(i) = (1-\beta) \underbrace{\frac{\rho(1-\alpha)^{\frac{\rho-1}{\rho}}}{\alpha(1-\rho)+\rho}}_{\text{integral}} 
\underbrace{\left(-\frac{1}{\rho}\right) \left(\frac{1-\alpha}{2^{\alpha}}\right)^{-\frac{1}{\rho}-1} \frac{1}{i^{\alpha\left(-\frac{1}{\rho}-1\right)}}}_{\mu \text{ exponent}} 
\underbrace{(-\alpha) \left(\frac{1-\alpha}{2^{\alpha}}\right) \frac{1}{i^{\alpha+1}}}_{\partial_{i} \mu(i)} = 
\end{equation*}  
\begin{equation*}
= (1-\beta) \frac{(1-\alpha)^{\frac{\rho-1}{\rho}}}{\alpha(1-\rho)+\rho} \left(\frac{1-\alpha}{2^{\alpha}}\right)^{-\frac{1}{\rho}} \alpha \frac{1}{i^{1-\frac{\alpha}{\rho}}} = \frac{\alpha(1-\beta)}{2^{\alpha}} \frac{(1-\alpha)^{\frac{\rho-2}{\rho}}}{\alpha(1-\rho)+\rho} i^{\frac{\alpha}{\rho}-1}.
\end{equation*}  
If we take \( \rho \in (0,1) \) (i.e., \( \sigma < 1 \)) and \( \alpha > \rho \), then \( \frac{\alpha}{\rho}-1 > 0 \), and the maximum of \( \partial_{i} \pi(i) \) occurs at \( i = 1/2 \), with value  
\begin{equation*}
\frac{\alpha(1-\beta)}{2^{\alpha+\frac{\alpha}{\rho}-1}} \frac{(1-\alpha)^{\frac{\rho-2}{\rho}}}{\alpha(1-\rho)+\rho} = \frac{\alpha(1-\beta)}{2^{\alpha+\frac{\alpha}{\rho}-1}} \frac{1}{\alpha(1-\rho)+\rho} \frac{1}{(1-\alpha)^{\frac{2-\rho}{\rho}}}.
\end{equation*}  
Therefore, when $0<\rho<\alpha<1$, if 
\[
\frac{\alpha(1-\beta)(1-\alpha)^{1-\frac{2}{\rho}}}{2^{\alpha+\frac{\alpha}{\rho}-1}[\alpha(1-\rho)+\rho]} \leq \sqrt{2 c_0},
\]  
the distribution \eqref{eq:distribution-equilibrium-example-2} is an equilibrium distribution for the model with the fixed cost of reallocation.
\end{proof}


\section{Extension: Non-symmetric preferences}\label{sec:nonsymmetric preferences}

In this section, we examine what happens when preferences are not symmetric over the range of goods $[0,n]$. Specifically, assume that the preferences take the form:
\begin{equation}
U = u(X), \qquad \text{with} \quad  X:= \left (\int_0^n \gamma(i)x(i)^{\frac{\sigma-1}{\sigma}} \, di \right )^{\frac{\sigma}{\sigma-1}},
\end{equation}
for some strictly positive (and continuous and bounded) function $\gamma \colon S \to \mathbb{R}$, where  $\gamma(i) > 0$, $\gamma \in C^2(S)$ and $\int_0^{n} \gamma(j) dj =n$. $\gamma(i)$ represents an index of the consumer’s specific preference for good $i$, whose magnitude positively affects its demand and, consequently, the share of total expenditure allocated to the purchase of that good.
As shown in the technical derivation below, the consumer's problem can be transformed into the same of Proposition \ref{pr:shor-run-equilibrium} by an appropriate rescaling of quantity of consumed good $x(i)$ and technology $A(i)$ (see Proposition~\ref{pr:shor-run-equilibrium-nonsymmetric} and its proof in Subsection~\ref{app:non-symmetric}).

The key distinguishing feature of this extension is the positive relationship between the share of income spent on purchasing good $i$ $Y(i)$ and $\gamma(i)$, i.e.:
\begin{equation}
\frac{p(i)x(i)}{y}= Y(i) =\gamma(i)^\sigma \left[\frac{p(i)}{P} \right]^{1-\sigma},
\label{eq:shareExpenditureGoodiHeterogeneousPreferences}
\end{equation}
which reflects in a positive relationship also between profits and $\gamma(i)$:
\begin{equation}
\label{eq:profitti_heterogeneity_preferences}
\pi(i) = \left(1-\beta \right) \left\{ \dfrac{\gamma(i)^{\frac{\sigma}{\sigma(1-\beta)+\beta}} A(i)^{\frac{\sigma -1}{\sigma(1-\beta)+\beta}}\mu(i)^{-\frac{1}{\sigma(1 - \beta) + \beta}}}{\displaystyle \int_0^n  \left[ \gamma(j)^{\frac{\sigma}{\sigma-1}} A(j) \mu(j)^{1-\beta} \right]^{\frac{\sigma - 1}{\sigma(1-\beta) + \beta}} \, dj} \right\},
\end{equation}
The resulting long-run equilibrium distribution of $\mu^{EQ}(i)$, $Y(i)^{EQ}$ and $L(i)^{EQ}$, i.e.
\begin{equation}
\mu^{EQ}(i) = Y(i)^{EQ} = L(i)^{EQ}=
\frac{\gamma(i)^\sigma A(i)^{\sigma-1}}{\displaystyle \int_0^n \gamma(j)^{\sigma} A(j)^{\sigma-1} \, dj},
\label{eq:shareExpenditureGoodiHeterogeneousPreferencesLongRun}
\end{equation}
highlights that also preferences contributes to determine the relative size of different sectors but still remains that labour productivity should converge to one in each sector.

If we rescale each variety by its taste weight
\[
  \tilde{x}(i)=x(i)\,\gamma(i)^{\frac{\sigma}{\sigma-1}},
  \qquad
  \tilde{A}(i)=A(i)\,\gamma(i)^{\frac{\sigma}{\sigma-1}},
\]
then the CES aggregator becomes identical to the symmetric case treated in Proposition
\ref{pr:shor-run-equilibrium}:
\[
  X \;=\;
  \Bigl(\int_{0}^{n}
        \tilde{x}(i)^{\frac{\sigma-1}{\sigma}}\,
        \mathrm{d}i\Bigr)^{\frac{\sigma}{\sigma-1}} .
\]
Production expressed in the new unit of measure is simply
\begin{equation}
  \tilde{q}(i)=\tilde{A}(i)\,l(i)^{\beta},
  \label{eq:productionFunctionTilde}
\end{equation}
so the entire equilibrium of Proposition \ref{pr:shor-run-equilibrium} applies \emph{verbatim}
after the single substitution
\[
  A(i)\quad\longrightarrow\quad\tilde{A}(i).
\]

When the solution is translated back to the original units one finds
\[
  p(i)=\tilde{p}(i)\,\gamma(i)^{\frac{\sigma}{\sigma-1}},
  \qquad
  Q(i)=\tilde{Q}(i)\,\gamma(i)^{-\frac{\sigma}{\sigma-1}},
\]
so that the price re‑introduces the factor $\gamma(i)$ while the physical
quantity is rescaled in the opposite direction.  
Because these two factors cancel in \(Y(i)=p(i)Q(i)\),
all other expressions (revenue, profits, the wage and welfare) remain exactly
those obtained under the substitution.  
This observation, together with Proposition \ref{pr:shor-run-equilibrium}, yields the following result; its full proof is provided in Subsection \ref{app:non-symmetric}.

For compactness, write
\[
\rho=\sigma(1-\beta)+\beta,\qquad
I_\gamma=\int_0^n[\tilde A(j)\mu(j)^{1-\beta}]^{(\sigma-1)/\rho}\,dj.
\]
\begin{Proposition}[The short-run equilibrium with non-symmetric preferences]
\label{pr:shor-run-equilibrium-nonsymmetric}
In the short-run equilibrium of economy with non-symmetric preferences:
\[
w = \beta,
\]
\[
P=\left( \int_0^n  \left[\left(\gamma(i)^{\frac{\sigma}{\sigma-1}} A(i)\right)^{\frac{1}{1-\beta}} \mu(i) \right]^{1-\frac{1}{\sigma(1-\beta) + \beta}} \, di \right)^{\frac{\sigma(1-\beta) + \beta}{1-\sigma}},
\]
\begingroup\footnotesize
\begin{equation}
\label{eq:value-of-pi-spillovers}
p(i)=\frac{\gamma(i)^{\sigma/(\sigma-1)}
                \tilde A(i)^{-1/\rho}\mu(i)^{-(1-\beta)/\rho}}
               {I_\gamma^{1-\beta}}.
\end{equation}
\begin{multline}
\label{eq:profitti_spillovers}
\pi(i) = \left(1-\beta \right)\frac{\left(\gamma(i)^{\frac{\sigma}{\sigma-1}} A(i)\right)^{\frac{1}{1-\beta}\left(1-\frac{1}{\sigma(1-\beta)+\beta}\right)}\mu(i)^{-\frac{1}{\sigma(1 - \beta) + \beta}}}{\displaystyle \int_0^n  \left[\left(\gamma(j)^{\frac{\sigma}{\sigma-1}} A(j)\right)^{\frac{1}{1-\beta}} \mu(j) \right]^{1-\frac{1}{\sigma(1-\beta) + \beta}} \, dj} = \\
=\left(1-\beta \right)\dfrac{\left(\gamma(i)\right)^{\frac{\sigma}{\sigma(1-\beta)+\beta}} \left(A(i)\right)^{\frac{\sigma -1}{\sigma(1-\beta)+\beta}}\mu(i)^{-\frac{1}{\sigma(1 - \beta) + \beta}}}{\displaystyle \int_0^n  \left[\left(\gamma(j)^{\frac{\sigma}{\sigma-1}} A(j)\right)^{\frac{1}{1-\beta}} \mu(j) \right]^{1-\frac{1}{\sigma(1-\beta) + \beta}} \, dj},
\end{multline}
\endgroup
the production in sector $i$ is 
\begin{equation}
\label{eq:Qi-nonsymmetric}
Q(i) = q(i) \mu(i) = \gamma(i)^{-\frac{\sigma}{\sigma-1}} \frac{\left[ \left(\gamma(i)^{\frac{\sigma}{\sigma-1}} A(i)\right)^{\frac{1}{1-\beta}} \mu(i) \right]^{\frac{\sigma(1-\beta)}{\sigma(1 - \beta) + \beta}}}{\displaystyle \left( \int_0^n  \left[\left(\gamma(j)^{\frac{\sigma}{\sigma-1}} A(j)\right)^{\frac{1}{1-\beta}} \mu(j) \right]^{1-\frac{1}{\sigma(1-\beta) + \beta}} \, dj \right)^{\beta}},
\end{equation}
the income per sector is
\begin{equation}
\label{eq:Yinonsymmetric}
Y(i) = p(i) Q(i) = \frac{\left[\gamma(i)^{\frac{\sigma}{\sigma-1}} A(i)  \mu(i)^{(1 - \beta)} \right]^{\frac{\sigma-1}{\sigma(1 - \beta) + \beta}} }{\displaystyle \int_0^n  \left[\left(\gamma(j)^{\frac{\sigma}{\sigma-1}} A(j)\right)^{\frac{1}{1-\beta}} \mu(j) \right]^{1-\frac{1}{\sigma(1-\beta) + \beta}} \, dj },
\end{equation}
while the number of employed workers per sector is
\begin{equation}
\label{eq:L-nonsymmetric}
L(i) =\mu(i) l(i) = \frac{\left[ \gamma(i)^{\frac{\sigma}{\sigma-1}} A(i)  \mu(i)^{1-\beta} \right]^{\frac{\sigma-1}{\sigma(1-\beta)+\beta}}}
{\displaystyle \int_0^n  \left[\left(\gamma(j)^{\frac{\sigma}{\sigma-1}} A(j)\right)^{\frac{1}{1-\beta}} \mu(j) \right]^{1-\frac{1}{\sigma(1-\beta) + \beta}} \, dj }.
\end{equation}
As before $\Pi = (1-\beta)$.
\end{Proposition}

In this case, the arguments from Section \ref{sub:dynamics}, using the profit expression from Proposition~\ref{pr:shor-run-equilibrium-nonsymmetric} we obtain that the distribution equation for $\mu$ follows the equation:
\begin{equation}
\label{eq:state-nonsymmetric}
\left \{
\begin{array}{l}
\displaystyle \partial_t \mu(t,i) = -(1-\beta) \frac{\partial_i \left( \mu(t,i) \partial_i \left(\left(\gamma(i)^{\frac{\sigma}{\sigma-1}} A(i)\right)^{\frac{1}{1-\beta}\left(1-\frac{1}{\sigma(1-\beta)+\beta}\right)}\mu(t,i) 
^{-\frac{1}{\sigma(1 - \beta) + \beta}}\right) \right)}{\displaystyle \int_0^n  \left[\left(\gamma(j)^{\frac{\sigma}{\sigma-1}} A(j)\right)^{\frac{1}{1-\beta}} \mu(t,j) \right]^{1-\frac{1}{\sigma(1-\beta) + \beta}} \, dj}\\[7pt]
\displaystyle \mu(t,0) = \mu(t,n)\\[7pt]
\displaystyle \mu(0,i) = \mu_0(i).
\end{array}
\right .
\end{equation}

This is the same equation (\ref{eq:stato-di-nuovo}) where $\tilde A(i) := \gamma(i)^{\frac{\sigma(1-\beta)}{\sigma(1 - \beta) + \beta-1}} A(i)$ appears in place of $A(i)$. Therefore, Theorems \ref{th:riscritturaGradientFlow} and \ref{teo:convergence} also apply to this case. The limiting distribution in this case will thus be:
\[
\mu^{EQ}(i) = \frac{\tilde A(i)^{\frac{\sigma(1 - \beta) + \beta-1}{1-\beta}}}{\displaystyle \int_0^n \tilde A(j)^{\frac{\sigma(1 - \beta) + \beta-1}{1-\beta}} \, dj} = \frac{\gamma(i)^{\sigma} A(i)^{\sigma-1}}{\displaystyle \int_0^n \gamma(j)^{\sigma} A(j)^{\sigma-1} \, dj}.
\]
By substituting this expression into (\ref{eq:Yi-app-nonsymmetric}) we get the $Y(i)^{EQ}$ described above in this appendix.

\subsection{Proof of the short-run equilibrium}
\label{app:non-symmetric}

\begin{proof}
Following the same approach as in Section \ref{sec:shortRunEquilibrium}, one begins by finding the demand $x(i)$ as a function of $p(i)$ and $P$. In this case, we have:
\[
x(i) = \frac{y}{\left ( \frac{p(i)}{\gamma(i)} \right )^{\sigma}}\frac{1}{P^{1-\sigma}}
\]
where $P$ is defined as
\[
P = \left ( \int_0^n \gamma(i)^{\sigma} p(i)^{1-\sigma} di \right )^{1/(1-\sigma)}.
\]
On the supply side, things remain as in the symmetric case, and the supply of each firm in sector $i$ is obtained from the profit-maximization first-order condition. Thus, the market equilibrium for good $i$ gives, for $Y=y=1$:
\[
\frac{1}{\left ( \frac{p(i)}{\gamma(i)} \right )^{\sigma}}\frac{1}{P^{1-\sigma}} = x(i) = \mu(i) q(i) = \mu(i)A(i)^{1/\left(1-\beta\right)} p(i)^{\beta/\left(1-\beta\right)}
\]
from which, as in the proof of Proposition \ref{pr:shor-run-equilibrium},
\[
p(i) = \left[  P^{1-\sigma}  A(i)^{1/\left(1-\beta\right)} \mu(i) \gamma(i)^{-\sigma} \right]^{-\left(1-\beta\right)/\left(\sigma\left(1-\beta\right)+\beta\right)}.
\]
Thus, we have:
\begin{multline}
P^{1-\sigma} = \left( \int_0^n \gamma(i)^{\sigma} p(i)^{1-\sigma} \, di \right)\\ 
= P^{-\frac{(1-\sigma)^2(1-\beta)}{\sigma(1-\beta) + \beta}} \left( \int_0^n \gamma(i)^{\sigma} \left[ A(i)^{1/(1-\beta)} \mu(i) \gamma(i)^{-\sigma}\right]^{-(1-\beta)(1-\sigma)/(\sigma(1-\beta) + \beta)} \, di \right).
\end{multline}
so that
\begin{multline*}P = \left( \int_0^n \gamma(i)^{\frac{\sigma}{\sigma(1-\beta) + \beta}} \left[ A(i)^{1/(1-\beta)} \mu(i) \right]^{1-\frac{1}{\sigma(1-\beta) + \beta}} \, di \right)^{\frac{\sigma(1-\beta) + \beta}{1-\sigma}} \\
= \left( \int_0^n  \left[\left(\gamma(i)^{\frac{\sigma}{\sigma-1}} A(i)\right)^{\frac{1}{1-\beta}} \mu(i) \right]^{1-\frac{1}{\sigma(1-\beta) + \beta}} \, di \right)^{\frac{\sigma(1-\beta) + \beta}{1-\sigma}} \\
= \left( \int_0^n  \left(\gamma(i)^{\sigma} A(i)^{\sigma -1} \mu(i)^{(1-\beta)(\sigma-1)}\right)^{\dfrac{1}{\sigma(1-\beta)+ \beta}} \, di \right)^{\frac{\sigma(1-\beta) + \beta}{1-\sigma}},
\end{multline*}
and thus, using the definition of $I_\gamma$,
\begin{align*}
p(i)
 &=P^{-(1-\sigma)(1-\beta)/\rho}
   A(i)^{-1/\rho}\mu(i)^{-(1-\beta)/\rho}
   \gamma(i)^{\sigma(1-\beta)/\rho}\\
 &=I_\gamma^{-(1-\beta)}
   A(i)^{-1/\rho}\mu(i)^{-(1-\beta)/\rho}
   \gamma(i)^{\sigma(1-\beta)/\rho}\\
 &=\frac{\gamma(i)^{\sigma/(\sigma-1)}
          \tilde A(i)^{-1/\rho}\mu(i)^{-(1-\beta)/\rho}}
         {I_\gamma^{1-\beta}}.
\end{align*}
Therefore, using $\rho=\sigma(1-\beta)+\beta$, the profits and sectoral quantities can be written as
\begin{equation}
\label{eq:pi-app-nonsymmetric}
\begin{aligned}
\pi(i)&=(1-\beta)A(i)^{1/(1-\beta)}p(i)^{1/(1-\beta)}\\
 &=\frac{(1-\beta)\gamma(i)^{\sigma/\rho}
 A(i)^{(\sigma-1)/\rho}\mu(i)^{-1/\rho}}{I_\gamma}
 =\frac{(1-\beta)\tilde A(i)^{(\sigma-1)/\rho}\mu(i)^{-1/\rho}}{I_\gamma},
\end{aligned}
\end{equation}
\begin{equation}
\label{eq:xi-app-nonsymmetric}
\begin{aligned}
Q(i)=x(i)&=\mu(i)[A(i)p(i)^\beta]^{1/(1-\beta)}\\
 &=\frac{A(i)^{\sigma/\rho}\gamma(i)^{\beta\sigma/\rho}
          \mu(i)^{\sigma(1-\beta)/\rho}}{I_\gamma^\beta}\\
 &=\gamma(i)^{-\sigma/(\sigma-1)}
   \frac{[\tilde A(i)^{1/(1-\beta)}\mu(i)]^{\sigma(1-\beta)/\rho}}
        {I_\gamma^\beta},
\end{aligned}
\end{equation}
\begin{equation}
\label{eq:Yi-app-nonsymmetric}
\begin{aligned}
Y(i)&=p(i)Q(i)=A(i)^{1/(1-\beta)}p(i)^{1/(1-\beta)}\mu(i)\\
 &=\frac{\gamma(i)^{\sigma/\rho}A(i)^{(\sigma-1)/\rho}
          \mu(i)^{(\rho-1)/\rho}}{I_\gamma}
 =\frac{[\tilde A(i)\mu(i)^{1-\beta}]^{(\sigma-1)/\rho}}{I_\gamma}.
\end{aligned}
\end{equation}
Finally, the labour first-order condition is $w l(i)=\beta p(i)q(i)$.
Since $w=\beta$, multiplying by $\mu(i)$ gives
\begin{equation}
\label{eq:Li-app-nonsymmetric}
L(i)=\mu(i)l(i)=p(i)Q(i)=Y(i),
\end{equation}
which proves the employment expression stated above.
\end{proof}


\section{Extension: Immobile labour}  
\label{sec:immobilework}

Figure \ref{fig:labourProductivitycrossSectoralDistribution} shows that the cross‑sector distribution of labour productivity remains non‑uniform and strongly persistent over a five‑year window, suggesting that sectoral wages do not converge in that timeframe. Meanwhile, a substantial body of literature, notably \citet{lee2006intersectoral}, documents significant and persistent inter‑sectoral \textit{labour} reallocation costs, driven by geographic frictions, institutional rigidities, and skill‑mismatch barriers. Both evidence point a large inertia in labour reallocation.
In this section, we therefore examine the polar opposite of the perfectly mobile‑labour case, namely an economy with completely immobile labour, in which workers are assumed to remain locked into specific sectors and are unable to respond to wage differentials by switching employment.

We denote the sectoral distribution of labour by $L(i)$, representing the total mass of workers available in sector $i$ fixed over time. Wages are therefore heterogeneous across sectors and determined by the equality between supply and demand of labour, as stated in Proposition \ref{pr:shor-run-equilibrium-immobilework} reporting the characteristics of the short-run competitive equilibrium.

\begin{Proposition}[The short-run competitive equilibrium with immobile labour]
\label{pr:shor-run-equilibrium-immobilework}
Let $L(i)$ be the stock of labour in sector $i$ and set as the numeraire of economy the aggregate value of production $Y$, i.e. $Y = 1$; then, in
the short-run equilibrium with immobile labour:\begin{eqnarray}
w(i) = \beta \left\{ \dfrac{A_0(i)^{1-\frac{1}{\sigma}} L(i)^{\frac{\sigma(\beta-1) -\beta }{\sigma}} \mu(i)^{(1 + \eta -\beta)(\frac{\sigma - 1}{\sigma})} }{\displaystyle \int_0^n \left[A_0(j)L(j)^\beta\mu(j)^{1+\eta -\beta}\right]^{\frac{\sigma - 1}{\sigma}} dj }\right\},
\end{eqnarray}
\begin{equation}
p(i) =  \frac{\left[ A_0(i)L(i)^\beta\mu(i)^{1+\eta -\beta}\right]^{-\frac{1}{\sigma}}}{\displaystyle \int_0^n \left[A_0(j)L(j)^\beta\mu(j)^{1+\eta -\beta}\right]^{\frac{\sigma - 1}{\sigma}} dj}
\end{equation}
\begin{eqnarray}
\label{eq:profitsCapitalStockLabourImmobility}
\pi(i) = (1-\beta) \left\{\dfrac{\left[A_0(i) L(i)^{\beta}\right]^{ \frac{\sigma - 1}{\sigma}} \mu(i)^{ -\frac{ \left(\beta -\eta  \right) \left(\sigma -1 \right) + 1}{\sigma} }} {\displaystyle \int_0^n \left[A_0(j)L(j)^\beta\mu(j)^{1+\eta -\beta}\right]^{\frac{\sigma - 1}{\sigma}} dj } \right\}
\end{eqnarray}
\begin{equation}
\label{eq:YimmobileLabour}
Y(i) = \dfrac{\left[A_0(i) L(i)^{\beta}\mu(i)^{1 + \eta -\beta}\right]^{\frac{\sigma-1}{\sigma}}   }{\displaystyle \int_0^n \left[A_0(j)L(j)^\beta\mu(j)^{1+\eta -\beta}\right]^{\frac{\sigma-1}{\sigma}} dj }.
\end{equation}
\begin{equation}
P  = \left(\int_0^n \left[ A_0(i)L(i)^\beta\mu(i)^{1+\eta -\beta}\right]^{\frac{\sigma-1}{\sigma}} di \right)^{\frac{\sigma}{1-\sigma}}
\end{equation}
and
\begin{equation}
    \Pi = 1-\beta
\end{equation}
\end{Proposition}

\begin{Remark}
In particular, using \eqref{eq:profitsCapitalStockLabourImmobility} and \eqref{eq:YimmobileLabour} the following relation holds between $Y(i)$ and $\pi(i)$:
\begin{equation}\label{eq:relYpi}
Y(i) =  \frac{\pi(i)^{\phi}\,\bigl[A_0(i)L(i)^{\beta}\bigr]^{\,(1-\tfrac{1}{\sigma})(1-\phi)}}{\displaystyle\int_{0}^{n}\pi(j)^{\phi}\,\bigl[A_0(j)L(j)^{\beta}\bigr]^{\,(1-\tfrac{1}{\sigma})(1-\phi)}\, dj}.    
\end{equation}
where $\phi \equiv \frac{(1+\eta-\beta)(\sigma-1)}{(\eta-\beta)(\sigma-1)-1}$.
\end{Remark}

The proof is provided in Subsection~\ref{app:proofsimmobilework}.
Both wages and labour productivity appears now sector specific, with the later given by:
\begin{equation}
\label{eq:labourProductivity}
\frac{Y(i)}{L(i)} = \dfrac{A_0(i)^{\frac{\sigma-1}{\sigma}} L(i)^{\frac{\sigma(\beta-1) - \beta}{\sigma}} \mu(i)^{\frac{(1 + \eta -\beta)(\sigma-1)}{\sigma}} }{\displaystyle \int_0^n \left[A_0(j)L(j)^\beta\mu(j)^{1+\eta -\beta}\right]^{\frac{\sigma-1}{\sigma}} dj },
\end{equation}
and both results decreasing in $L(i)$ as expected, while the relationship with technology and the mass of firm crucially depends on $\sigma$. 

Following the same steps of Section \ref{sub:dynamics}, we can get the PDE describing the dynamics of firms' distribution: 
\begin{equation}
\label{eq:stato-di-nuovoImmobileSpillover}
\begin{array}{l}
\displaystyle 
\partial_{t}\mu(t,i) = 	- \left(1-\beta \right) \frac{\partial_{i} \left ( \mu(t,i) \partial_{i} \left (\left[A_0(i) L(i)^{\beta}\right]^{\frac{\sigma - 1}{\sigma}} \mu(t,i)^{ -\frac{ \left(\beta -\eta  \right) \left(\sigma -1 \right) + 1}{\sigma} } \right ) \right )}{\displaystyle \int_0^n \left[A_0(j)L(j)^\beta\mu(t,j)^{1+\eta -\beta}\right]^{\frac{\sigma - 1}{\sigma}} dj } 
\end{array}
\end{equation}
Under sectoral externalities and complete labour immobility, ensuring existence, uniqueness, and boundedness of the industry‑wide firms' distribution, as established in Theorem \ref{th:esistenzaeunicita}, and thus enabling the PDE \eqref{eq:stato-di-nuovoImmobileSpillover} to be formulated as a gradient flow, requires the exponent of $\mu(i)$ to be negative, i.e.:
\begin{equation}
\label{eq:conditionGradientFlowImmobilelabour}
   - \left(\beta -\eta  \right) \left(\sigma -1 \right) - 1 = \eta\left(\sigma -1 \right) - \beta\left(\sigma -1 \right) - 1  < 0.
\end{equation}
When no externalities are present, i.e. $\eta=0$, Condition \eqref{eq:conditionGradientFlowImmobilelabour} always holds as in the baseline economy with mobile labour; instead, for $\eta \neq 0$, if $\sigma  > 1$, Condition \eqref{eq:conditionGradientFlowImmobilelabour} holds when $\eta \in \left( - \infty, \beta + 1/\left(\sigma -1 \right) \right)$, while, if $\sigma \in (0,1)$, when $\eta \in \left( \beta + 1/\left(\sigma -1\right), \infty \right)$. The explanation for these upper and lower bounds strictly follows the one discussed in Appendix \ref{sub:spillovers}. Under Condition~\eqref{eq:conditionGradientFlowImmobilelabour}, the following theorem makes the extension precise.

\begin{theorem}
\label{th:riscritturaGradientFlowImmobileSpillover}
Let $A_0,L\in C^2(S)$ be strictly positive, with $\int_S L(i)\,di=1$, and let $\mu_0\in C^2(S)$ be strictly positive with unit mass. Assume Equation~\eqref{eq:conditionGradientFlowImmobilelabour}. Put
\[
s=1+\eta-\beta,\qquad
m=\frac{s(\sigma-1)}{\sigma}<1,\qquad
a(i)=[A_0(i)L(i)^\beta]^{(\sigma-1)/\sigma}.
\]
The periodic equation~\eqref{eq:stato-di-nuovoImmobileSpillover}, with $\mu(0,\cdot)=\mu_0$, has a unique global strictly positive solution. Its potential is
\begin{equation}
\label{eq:longRunEQintraSpilloversImmobileLabour}
\mathcal F(\mu)=
\begin{cases}
\displaystyle\frac{1-\beta}{m}\log\!\left(\int_S a(i)\mu(i)^m\,di\right)
 =\frac{1-\beta}{s}\log X(\mu),&s\ne0,\\[6pt]
\displaystyle\frac{1-\beta}{\int_S a(j)\,dj}
 \int_S a(i)\log\mu(i)\,di,&s=0.
\end{cases}
\end{equation}
With the increasing-potential convention of the main text,
\begin{equation}
\label{eq:gradientflowImmobileSpillover}
\frac{d\mu}{dt}=\nabla_{W_2}\mathcal F(\mu),
\qquad
\frac{d}{dt}\mathcal F(\mu(t))
 =\int_S\mu(t,i)|\partial_i\pi(t,i)|^2\,di\geq0.
\end{equation}
The density converges to the unique stationary distribution
\begin{equation}
\label{eq:muEQImmobileLabour}
\mu^{EQ}(i)=Y^{EQ}(i)
 =\frac{[A_0(i)L(i)^\beta]^{\frac{\sigma-1}{1+(\beta-\eta)(\sigma-1)}}}
 {\displaystyle\int_S[A_0(j)L(j)^\beta]^{\frac{\sigma-1}{1+(\beta-\eta)(\sigma-1)}}\,dj}.
\end{equation}
Consequently,
\begin{equation}
\label{eq:labourProductivityEQImmobileLabour}
\left(\frac{Y(i)}{L(i)}\right)^{EQ}
 =\frac{A_0(i)^{\frac{\sigma-1}{1+(\beta-\eta)(\sigma-1)}}
 L(i)^{\frac{\eta(\sigma-1)-1}{1+(\beta-\eta)(\sigma-1)}}}
 {\displaystyle\int_S[A_0(j)L(j)^\beta]^{\frac{\sigma-1}{1+(\beta-\eta)(\sigma-1)}}\,dj}.
\end{equation}
\end{theorem}

The allocation of labour therefore affect the firms and value of production distribution also in the long run; in particular, $\mu^{EQ}(i)$ and $Y^{EQ}(i)$ are positively (negatively) related to the mass of workers $L(i)$ if $\sigma >1 \,(<1)$. The relationship between the labour productivity $\left(\frac{Y(i)}{L(i)}\right)^{EQ}$ (and wages) and $L(i)$ is more complex: if $\sigma <1$, it is negative; while if $\sigma >1$, for $\eta \in (1/(\sigma-1), \beta + 1/(\sigma-1))$ it is positive, while it is negative for $\eta \in (-\infty, 1/(\sigma-1))$.

We conclude by noting that labour immobility results in a loss of efficiency. In particular, measuring this loss by the difference of consumption in the long-run equilibrium, Proposition \ref{prop:functionalLossEfficiency} states that the economy with immobile labour reaches the most efficient equilibrium when labour is allocated as in the mobile labour economy.

\begin{Proposition}
\label{prop:functionalLossEfficiency}
Assume both Conditions~\eqref{eq:condizione-spillovers} and~\eqref{eq:conditionGradientFlowImmobilelabour}, with strictly positive $A_0$ and unit total labour. Let $\Delta X^{EQ}$ denote the loss in long-run equilibrium consumption caused by labour immobility in the long-run competitive equilibrium, measured by the resulting reduction in consumption:
\begin{equation}
\label{eq:functionalLossEfficiency}
    \Delta X^{EQ} \equiv \left( \int_0^n {A}_0(i)^{\frac{\sigma - 1}{1 - \eta (\sigma - 1)}} \, di \right)^ {\frac{1- \eta(\sigma - 1)}{ \sigma -1}} - \left(\int_0^n \left[A_0(i)L(i)^\beta\right]^{\frac{\sigma-1}{1+(\beta-\eta)(\sigma-1)}}\,\mathrm{d}i\right)^{\frac{1+(\beta-\eta)(\sigma-1) }{\sigma-1}}.
\end{equation}
Then, subject to the constraints $L(i) \geq 0$ and $\int_0^n L(i) \di = 1$, $\Delta X^{EQ}$ is always non-negative, with a minimum value of 0 corresponding to the labour allocation of the mobile labour economy. 
\end{Proposition}
The proof is provided in Subsection~\ref{app:proofsimmobilework}.

Proposition~\ref{prop:functionalLossEfficiency} compares long-run equilibrium consumption across labour allocations. When $1+\eta-\beta>0$, the mobile-labour equilibrium also maximizes consumption over firm allocations, as shown in Appendix~\ref{sub:spillovers}. With sufficiently negative externalities this further welfare interpretation does not follow, because the decentralized firm dynamics can instead reduce consumption.

\subsection{Technical derivations and proofs}
\label{app:proofsimmobilework}

\begin{proof}[Proof of Proposition \ref{pr:shor-run-equilibrium-immobilework}]
We denoted by $L$ the distribution of labour across different sectors so that in particular, $L(i)$ denotes the mass of workers in sector $i$.
The demand of labour of \textit{representative firm} of sector $i$ is again: 
\begin{equation}
l(i) = \left[\dfrac{\beta p(i) A(i)}{w(i)}\right]^{1/\left(1-\beta\right)} = \left[\dfrac{\beta p(i) A_0(i) \mu(i)^\eta}{w(i)}\right]^{1/\left(1-\beta\right)}
\label{eq:firmLabourDemand-nomobile}
\end{equation}
the equilibrium wage $w(i)$ is therefore decided by the equality between total demand and total supply of labour, i.e.:
\begin{equation}
l(i) \mu(i) =  \left[\dfrac{\beta p(i) A_0(i) \mu(i)^\eta}{w(i)}\right]^{1/\left(1-\beta\right)} \mu(i) = L(i);
\end{equation}
from which:
\begin{equation}
w(i) =  \beta p(i) A_0(i) \mu(i)^\eta \left(\dfrac{\mu(i)}{L(i)}\right)^{1-\beta}	
\end{equation}
The equilibrium price $p(i)$ is decided by the equality between total demand (see (\ref{eq:demandGoodi})) and total supply, i.e.:
\begin{equation}
x(i) = \dfrac{y}{p(i)^\sigma} \dfrac{1}{P^{1-\sigma}} = A_0(i) \mu(i)^\eta \left(\dfrac{L(i)}{\mu(i)}\right)^\beta \mu(i), 
\end{equation}
i.e.
\begin{equation}
p(i) = \left(\dfrac{y}{P^{1-\sigma}}\right)^{1/\sigma} \left(A_0(i)L(i)^\beta\mu(i)^{1+\eta-\beta}\right)^{-1/\sigma}
\end{equation}
Using (\ref{eq:preferences}) and the normalization $Y=y=1$, from the previous equation, it is also possible to calculate the level of price index in equilibrium:
\begin{align*}
P^{1-\sigma}
 &=\int_0^n p(i)^{1-\sigma}\,di\\
 &=P^{-(1-\sigma)^2/\sigma}
   \int_0^n[A_0(i)L(i)^\beta\mu(i)^{1+\eta-\beta}]^{(\sigma-1)/\sigma}\,di.
\end{align*}
and then:
\begin{equation}
P  = \left(\int_0^n \left( A_0(i)L(i)^\beta\mu(i)^{1+\eta -\beta}\right)^{1-1/\sigma} di \right)^{\sigma/\left(1-\sigma\right)}
\end{equation}
and
\begin{equation}
p(i) =  \frac{\left( A_0(i)L(i)^\beta\mu(i)^{1+\eta -\beta}\right)^{-1/\sigma}}{\int_0^n \left(A_0(i)L(i)^\beta\mu(i)^{1+\eta -\beta}\right)^{1-1/\sigma} di}
\end{equation}

Finally, we can calculate the equilibrium profits per firm:
\begin{equation}
\label{eq:formadipiappeC}
\begin{aligned}
\pi(i)&=(1-\beta)p(i)A_0(i)\mu(i)^\eta
                    \left(\frac{L(i)}{\mu(i)}\right)^\beta\\
&=\frac{(1-\beta)A_0(i)^{1-1/\sigma}
 L(i)^{\beta(1-1/\sigma)}
 \mu(i)^{(\eta-\beta)(1-1/\sigma)-1/\sigma}}
 {\displaystyle\int_0^n[A_0(j)L(j)^\beta
       \mu(j)^{1+\eta-\beta}]^{1-1/\sigma}\,dj}.
\end{aligned}
\end{equation}
and the equilibrium wage:
\begin{eqnarray}
w(i) &=& \dfrac{\beta A_0(i)^{1-1/\sigma} L(i)^{\beta\left(1-1/\sigma\right)-1} \mu(i)^{(1 + \eta -\beta)(1-1/\sigma)} }{\left(\int_0^n \left(A_0(i)L(i)^\beta\mu(i)^{1+\eta-\beta}\right)^{1-1/\sigma} di \right)}
\end{eqnarray}
At sectoral level, we have that the value of production is:
\begin{equation}
\label{eq:formadiYappeC}
Y(i) = \dfrac{A_0(i)^{1-1/\sigma} L(i)^{\beta\left(1-1/\sigma\right)} \mu(i)^{(1+\eta -\beta)(1-1/\sigma)}} {\left(\int_0^n \left(A_0(i)L(i)^\beta\mu(i)^{1+\eta -\beta}\right)^{1-1/\sigma} di \right)}.
\end{equation}
The total production in each sector in real terms is
\[
Q(i) = Y(i)/p(i) = A_0(i)L(i)^\beta\mu(i)^{1+\eta -\beta}
\]
\end{proof}

\begin{proof}[Proof of Theorem~\ref{th:riscritturaGradientFlowImmobileSpillover}]
With the definitions of $a$ and $m$ in the theorem, Equation~\eqref{eq:profitsCapitalStockLabourImmobility} reads
\[
\pi(i)=\frac{(1-\beta)a(i)\mu(i)^{m-1}}
             {\int_S a(j)\mu(j)^m\,dj}.
\]
Moreover, $m-1=[(\eta-\beta)(\sigma-1)-1]/\sigma<0$ is exactly Condition~\eqref{eq:conditionGradientFlowImmobilelabour}. All hypotheses of Lemma~\ref{lem:appendix-power-profit} therefore hold. Its variational derivative, maximum-principle bounds and weighted energy estimate prove the asserted gradient representation, global positive solution and convergence. Substituting $a$ into its stationary density gives Equation~\eqref{eq:muEQImmobileLabour}. At stationarity $\pi(i)=1-\beta$, so Equation~\eqref{eq:relYpi}, or directly $Y=\mu\pi/(1-\beta)$, gives $Y^{EQ}=\mu^{EQ}$. Division by $L(i)$ gives Equation~\eqref{eq:labourProductivityEQImmobileLabour}. The $m=0$ branch handles $1+\eta-\beta=0$ without dividing by zero.
\end{proof}

\begin{proof}[Proof of Proposition~\ref{prop:functionalLossEfficiency}]
Define
\[
\gA=\frac{\sigma-1}{1-\eta(\sigma-1)},\qquad
\gB=\frac{\sigma-1}{1+(\beta-\eta)(\sigma-1)}.
\]
The two stated conditions make their denominators positive, so $\gA$ and $\gB$ have the same sign and
$\gB/\gA+\beta\gB=1$.
Set
\[
Z=\int_S A_0(i)^{\gA}\,di,\qquad
w_*(i)=\frac{A_0(i)^{\gA}}{Z},\qquad
q=\beta\gB.
\]
Then $\int_S w_*=1$ and
\begin{equation}
\label{eq:inequality_to_prove}
\int_S A_0(i)^{\gB}L(i)^{\beta\gB}\,di
 =Z^{\gB/\gA}\int_S w_*(i)
                 \left(\frac{L(i)}{w_*(i)}\right)^q\,di.
\end{equation}
Since $\int_S L=1$, Jensen's inequality makes the last integral at most one if $\gB>0$, because then $0<q<1$. If $\gB<0$, then $q<0$ and convexity makes it at least one. Raising to the power $1/\gB$ therefore gives, in both cases,
\begin{equation}
\label{eq:efficiencyProofStepSegno}
\left(\int_S A_0(i)^{\gB}L(i)^{\beta\gB}\,di\right)^{1/\gB}
 \leq Z^{1/\gA}.
\end{equation}
This is precisely $\Delta X^{EQ}\geq0$. Strictness in Jensen's inequality gives equality if and only if $L/w_*$ is constant almost everywhere. Unit mass then requires
\[
L^*(i)=w_*(i)=\frac{A_0(i)^{\gA}}{\int_S A_0(j)^{\gA}\,dj},
\]
the labour allocation in the mobile-labour long-run equilibrium. If labour is allowed to vanish, the inequality extends by a positive approximation, with consumption interpreted by its extended limit when the CES integral diverges.
\end{proof}

\graphicspath{{./}}

\section{A primer quantitative evaluation of the model}
\label{app:quantitative}
In Section \ref{sec:backData}, we use a large sample of firms to estimate the model’s parameters. These parameter values are then employed in Section \ref{sec:numericalExperiments} to conduct numerical experiments.

\subsection{Back to the data}
\label{sec:backData}
This section presents an empirical evaluation of the model. We analyse a sample of over 750,000 firms from 14 EU countries, drawn from the ORBIS dataset for 2018 and 2023, and aggregated into 680 four-digit NACE sectors (ranging from 01 to 82). In particular, Section \ref{sec:converegenceROE} provides empirical evidence supporting the fundamental mechanism behind the reallocation of firms across sectors, specifically the convergence dynamics of sectoral profit rates, as measured by \textit{returns on equity} (ROE) between 2018 and 2023. Moreover, it suggests that any potential fixed costs of reallocation do not have a significant impact on the overall dynamics. Section \ref{sec:parametersModelCalibration}, on the other hand, discusses a method for assessing the degree of labour mobility, concluding that, at least within a 5-year horizon, labour appears to be relatively immobile. Taking the extreme case of labour immobility, we then proceed to calculate the model parameters.
A significant limitation of the ORBIS data is the challenge of obtaining a reliable proxy for $\mu(i)$, the stock of physical capital allocated to sector $i$.\footnote{The only variable in the ORBIS dataset that could potentially serve as a proxy for $\mu$ is the total assets reported on balance sheets; however, this includes both physical capital and financial assets, the mix of which is highly heterogeneous across sectors. Since the empirical analysis is not central to the scope of this paper, and for the sake of brevity, we discuss in Appendix \ref{app:identificationTechnologicalProgress} a procedure for estimating $\mu(i)$, based on some assumptions regarding the distribution of $A(i)$.} To address this issue, we focus in the empirical analysis below on ROE as a proxy for profit rates $\pi$, the \textit{value of production} as a proxy for $Y$, and \textit{employment} as a proxy for $L$.

\subsubsection{Convergence in the sectoral profit rates} \label{sec:converegenceROE}

Theorem \ref{teo:convergence}, one of our key results, demonstrates the convergence to the long-run competitive equilibrium, where sectoral profit rates are equalised. This result holds even in the presence of non-symmetric preferences, intrasectoral externalities and labour immobility, but does not apply when there are fixed costs of reallocation. A suitable (econometric) model for the dynamics of profit rates is, therefore, the following:
\begin{equation}
\pi(i,t) - \pi(i,0) \equiv  \Delta \pi(i,t) = m(\pi(i,0)) + \epsilon(i), \text{ with } m^\prime < 0, \; m(-\infty)>0 \text{ and } m(\infty)<0,
\label{eq:dynamicsProfitsRate}
\end{equation}
and $\epsilon(i)$ is random disturbance term, assumed to be i.i.d. with zero mean and finite variance. Equation \eqref{eq:dynamicsProfitsRate} implies a global monotonic convergence to a unique long-term equilibrium profit rate $\pi(i)^{EQ}=\bar{\pi}^{EQ}$ defined by $m(\bar{\pi}^{EQ})=0$. Theorem \ref{teo:convergence}, moreover, states that $\beta=1-\bar{\pi}^{EQ}$. Figure \ref{fig:ROEregressionToTheMean} reports the Nadaraya-Watson estimation of $m(\cdot)$ of Equation \eqref{eq:dynamicsProfitsRate} with its 95\% confidence bands (dashed lines), where profit rate $\pi(i)$ is proxied by the average return on equity $ROE(i)$ in the sector $i$.
\begin{figure}[!htbp]
\centering
\caption{The change of ROE between 2018 and 2023 versus ROE in 2018 for 680 four-digit NACE sectors (from 01 to 82). The ticked lines is a nonparametric estimate by the Nadaraya-Watson estimation with its 95\% confidence bands (dashed lines). The dotted line is the linear regression reported in Table \ref{tab:convergenceInSectoralROE}.}
\label{fig:ROEregressionToTheMean}
\vspace{-0.5cm}
\includegraphics[width=0.5\textwidth]{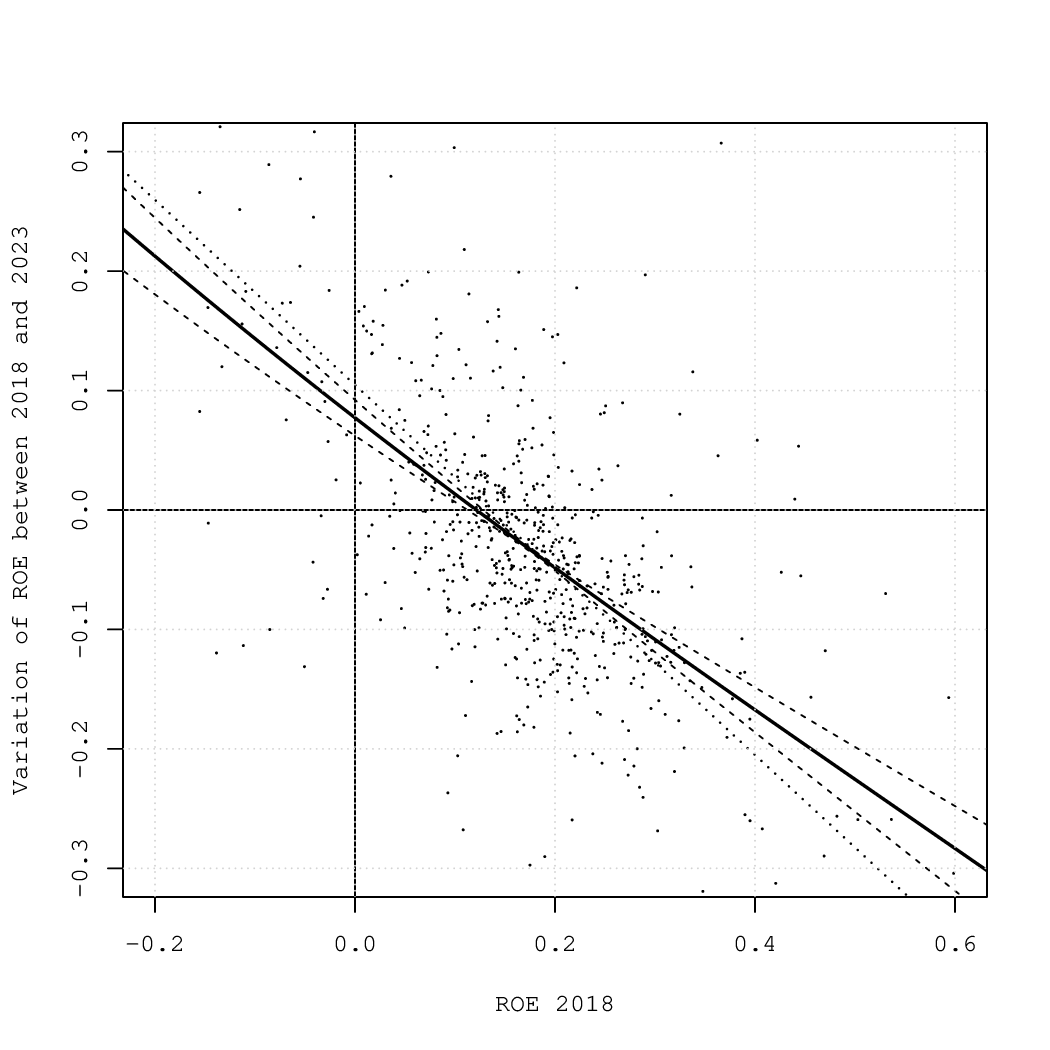}
\begin{flushleft}
    {\small \textit{Source: our elaboration on ORBIS data.}}
\end{flushleft}
\end{figure}
Figure \ref{fig:ROEregressionToTheMean} reports the The Nadaraya-Watson estimation of $m(\cdot)$ with its 95\% confidence bands (dashed lines). The estimation suggests the existence of a linear relationship, the estimate of which is, in turn, given by:\footnote{Table \ref{tab:convergenceInSectoralROE} in Appendix \ref{app:tableEstimation} reports the standard diagnostics of the estimate.} 
\begin{equation}
    \Delta ROE(i,2023) = \underset{(0.008)}{0.105} + \underset{(0.023)}{-0.775} ROE(i,2018);
     \label{eq:dynamicsProfitsRateParametricModel}
\end{equation}
the outcome of this estimate is reported by a dotted line in Figure \ref{fig:ROEregressionToTheMean}. From \eqref{eq:dynamicsProfitsRateParametricModel} we have an estimate of $\bar{\pi}^*= 0.105/0.775 \approx 0.14$, i.e. of $\beta=0.86$. A profit rate of 14\% is plausible taking into account that is \textit{nominal} (in that period inflation rate was about 2\% in EU) and is \textit{net of depreciation rate} of physical capital. A share of labour on total production equal to 0.86 would be instead overestimated, being in EU in 2022 the ratio between domestic production and compensation of employees around 0.24.\footnote{\label{foot:19}See \url{https://ec.europa.eu/eurostat/statistics-explained/index.php?title=Supply_and_use_tables_for_the_European_Union_and_the_euro_area}.} However, we should consider that in our framework two key factors are missing and both inflate the estimate of $\beta$, the \textit{depreciation rate} of physical capital, which amounts to about 10\% \citep{gupta2014efficiency}, and the \textit{intermediate consumption}, which absorbs about 54\% of domestic production in EU in 2022\footnote{As in footnote \ref{foot:19}.} Given a production including the intermediate consumption $c$ as $q=c^\alpha l^\beta$ (the unit of capital is set to one), and assuming that the shares of each input are decided by their marginal productivity, we have that the net profit rate is given by $(1-\alpha-\beta) - \delta$. where $\delta$ is the depreciation rate. Taken $\alpha=0.54$, $\delta=0.1$ and a net profit rate in equilibrium equal to 0.14, then $\beta=0.22$, which is very close to 0.24. The presence of intermediate production does not affect our theoretical results, as it appears as a rescaling of total production. The same holds for the depreciation rate, which does not affect the differential sectoral profit rates. 
Finally, the smoothness of convergence path to the long-run equilibrium in Figure \ref{fig:ROEregressionToTheMean} suggests that the possible presence of fixed costs of reallocation has not significant impact on the sectoral firm reallocation, at least at aggregate level. 

\subsubsection{Divergence in the sectoral labour productivity}

When we focus on labour productivity, Figure \ref{fig:labourProductivityRegressionToTheMean} reveals a divergent trend, i.e. labour productivity across sectors diverges over time, rather than regressing to the mean. Given the close relationship between productivity and wages, this lack of convergence suggests that labour reallocation, which wage differentials should primarily drive, is either absent or proceeds at a much slower pace than capital reallocation.
\begin{figure}[!htbp]
\centering
\caption{The change of labour productivity between 2018 and 2023 versus labour productivity in 2018 for 680 four-digit NACE sectors (from 01 to 82). The ticked lines is a nonparametric estimate by the Nadaraya-Watson estimation with its 95\% confidence bands (dashed lines). The dotted line is the linear regression reported in Table \ref{tab:convergenceLabourProductivity}.}
\label{fig:labourProductivityRegressionToTheMean}
\vspace{-0.5cm}
\includegraphics[width=0.5\textwidth]{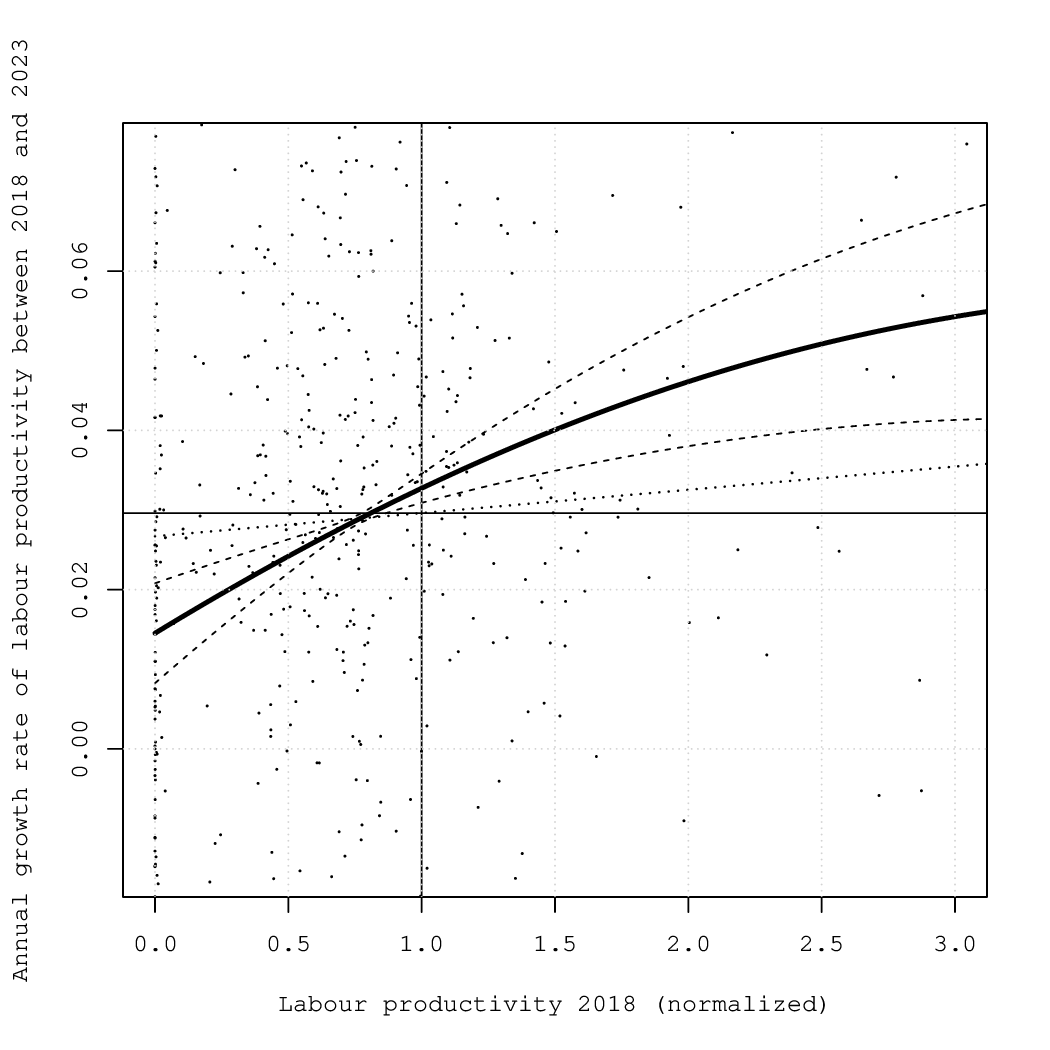}
\begin{flushleft}
    {\small \textit{Source: our elaboration on ORBIS data.}}
\end{flushleft}
\end{figure}
The estimated of the linear econometric model of convergence in labour productivity, i.e.:\footnote{Table \ref{tab:convergenceLabourProductivity} in Appendix \ref{app:tableEstimation} reports the standard diagnostics of the estimate.}
\begin{equation}
    \Delta \log(\text{Labour Productivity}(i,2023)) = \underset{(0.003)}{0.027} + \underset{(0.001)}{0.003} \text{Labour Productivity}(i,2018);
    \label{eq:dynamicsLabourProductivityParametricModel}
\end{equation}
confirms the (marginally significant) divergence in sectoral labour productivity.

\subsubsection{Labour mobility and the calibration of model parameters}
\label{sec:parametersModelCalibration}

Taken together, the evidence to date indicates that the immobile‑labour model provides the closest approximation for our sample of firms. A direct comparison of Propositions \ref{pr:shor-run-equilibrium} and \ref{pr:shor-run-equilibrium-immobilework} makes clear that the critical difference between the two models lies in the functional relationship between sector‑level output value and employment. With perfect mobile labour, the value of production $Y(i)$ and employment $L(i)$ are the same and profit rate does not play any role; on the contrary, with labour immobility, assuming $\eta=0$, the relationship is less than proportional, to become indeterminate when $\eta \neq 0$. In particular, from Equation \eqref{eq:relYpi} in Appendix \ref{app:proofsimmobilework}, we get: 
\begin{equation}
    \log \left( \dfrac{Y(i,t)}{Y(i,0)}\right) = g(i) +  \left[\frac{\beta(\sigma-1)}{1-(\eta-\beta)(\sigma-1)} \right]\log \left( \dfrac{L(i,t)}{L(i,0)}\right) + \left[\frac{(1+\eta-\beta)(\sigma-1)}{1-(\eta-\beta)(\sigma-1)} \right] \log \left( \dfrac{\pi(i,t)}{\pi(i,0)}\right),
    \label{eq:changeValueProductionImmobileLabour}
\end{equation}
where $g(i)$ reflects the time change both of $A(i)$ and $P$.\footnote{Equation \eqref{eq:changeValueProductionImmobileLabour} is robust to the presence of heterogeneous preferences. The expression for $g(i)$ in this case takes into account also the possible time change of the factor $\gamma(i)$.}
The estimate of the corresponding econometric model of
Equation \eqref{eq:changeValueProductionImmobileLabour} is:\footnote{Table \ref{tab:estimateValueProductionVsEmploymentProfitRate} in Appendix \ref{app:tableEstimation} reports the standard diagnostics of the estimate.}
\begin{equation}
    \log\left( \dfrac{Y(i,t)}{Y(i,0)}\right) = \underset{(0.011)}{0.243} +  \underset{(0.027)}{0.214} \log \left( \dfrac{L(i,t)}{L(i,0)}\right) +  \underset{(0.012)}{0.040} \log \left( \dfrac{ROE(i,t)}{ROE(i,0)}\right),
\label{eq:changeValueProductionImmobileLabourEconometricModel}
\end{equation}
where the random disturbance of the econometric model takes the possible (exogenous) heterogeneity in $g(i)$, while any other aggregate change is in the estimated constant.

The estimate in \eqref{eq:changeValueProductionImmobileLabourEconometricModel} point to a some degree of labour immobility, given the low elasticity of the value of production to employment (should be 1 in the case of perfect mobility) and the significant elasticity to profit rates (should be 0 in the case of perfect mobility).
Assuming the extreme scenario of labour immobility, the estimate in \eqref{eq:changeValueProductionImmobileLabourEconometricModel}, together with their theoretical counterparts in Equation \eqref{eq:changeValueProductionImmobileLabour}, allow to calculate the values of $\eta$ and $\sigma$, i.e. 
\begin{eqnarray}
    \eta &=& \beta \left(1+\rho_2/\rho_1\right) -1 \approx 0.0222, \text{ and }\\
    \sigma &=& 1+ \rho_1/\left[ \beta \left(1+\rho_2\right) - \rho_1 \right] \approx 1.3144, 
\end{eqnarray}
taking into account a value of $\beta=0.86$ estimated in Section \ref{sec:converegenceROE}.
Therefore, our economy appears to be characterised by positive but small intrasectoral externalities, while goods and services are moderately substitutable for one another.

\subsection{Diagnostics of the estimates of Section \ref{sec:backData} 
}
\label{app:tableEstimation}

\begin{table}[!htbp] \centering 
  \caption{Estimate of the convergence econometric model for ROE based on 680 four-digit NACE sectors (01–82) over the period 2018–2023 (s.e. in brackets).} 
  \label{tab:convergenceInSectoralROE} 
\begin{tabular}{@{\extracolsep{5pt}}lc} 
\\[-1.8ex]\hline 
\hline \\[-1.8ex] 
 & \multicolumn{1}{c}{\textit{Dependent variable:}} \\ 
\cline{2-2} 
\\[-1.8ex] & $\Delta ROE(2023)$ \\ 
\hline \\[-1.8ex] 
 $ROE(2018)$ & $-$0.775$^{***}$ \\ 
  & (0.023) \\ 
 Constant & $0.105^{***}$ \\ 
  & (0.008) \\ 
\hline \\[-1.8ex] 
Observations & 680 \\ 
R$^{2}$ & 0.618 \\ 
Adjusted R$^{2}$ & 0.617 \\ 
Residual Std. Error & 0.185 (df = 678) \\ 
F Statistic & 1,096.980$^{***}$ (df = 1; 678) \\ 
\hline 
\hline 
\textit{Note:}  & \multicolumn{1}{r}{$^{*}$p$<$0.1; $^{**}$p$<$0.05; $^{***}$p$<$0.01} \\ 
\end{tabular} 
\begin{flushleft}
   {\small \textit{Source: our elaboration on ORBIS data.}} 
\end{flushleft}
\end{table} 

\begin{table}[!htbp] \centering 
  \caption{Estimate of the convergence econometric model for labour productivity across 680 four-digit NACE sectors (01–82) in 2018 and 2023 (s.e. in brackets).} 
  \label{tab:convergenceLabourProductivity} 
\begin{tabular}{@{\extracolsep{5pt}}lc} 
\\[-1.8ex]\hline 
\hline \\[-1.8ex] 
 & \multicolumn{1}{c}{\textit{Dependent variable:}} \\ 
\cline{2-2} 
\\[-1.8ex] & $\Delta \log(\text{Labour Productivity}(i,2023))$ \\ 
\hline \\[-1.8ex] 
  $\log(\text{Labour Productivity}(i,2018))$ & 0.003$^{*}$ \\ 
  & (0.001) \\ 
  & \\ 
 Constant & 0.027$^{***}$ \\ 
  & (0.003) \\ 
  & \\ 
\hline \\[-1.8ex] 
Observations & 679 \\ 
R$^{2}$ & 0.006 \\ 
Adjusted R$^{2}$ & 0.004 \\ 
Residual Std. Error & 0.079 (df = 677) \\ 
F Statistic & 3.819$^{*}$ (df = 1; 677) \\ 
\hline 
\hline \\[-1.8ex] 
\textit{Note:}  & \multicolumn{1}{r}{$^{*}$p$<$0.1; $^{**}$p$<$0.05; $^{***}$p$<$0.01} \\ 
\end{tabular} 
\end{table} 

\begin{table}[!htbp] \centering 
  \caption{Estimate of the growth rate of the value of production for 680 four-digit NACE sectors (01–82) in 2018 and 2023 (s.e. in brackets).}
\label{tab:estimateValueProductionVsEmploymentProfitRate} 
\begin{tabular}{@{\extracolsep{5pt}}lc} 
\\[-1.8ex]\hline 
\hline \\[-1.8ex] 
 & \multicolumn{1}{c}{\textit{Dependent variable:}} \\ 
\cline{2-2} 
\\[-1.8ex] & $\log \left( Y(i,2023)/Y(i,2018)\right)$ \\ 
\hline \\[-1.8ex] 
 $\log \left( \dfrac{L(i,t)}{L(i,0)}\right)$ & 0.214$^{***}$ \\ 
  & (0.027) \\ 
  & \\ 
 $\log \left( \dfrac{ROE(i,t)}{ROE(i,0)}\right)$  & 0.040$^{***}$ \\ 
  & (0.012) \\ 
  & \\ 
 Constant & 0.243$^{***}$ \\ 
  & (0.011) \\ 
\hline \\[-1.8ex] 
Observations & 612 \\ 
R$^{2}$ & 0.111 \\ 
Adjusted R$^{2}$ & 0.108 \\ 
Residual Std. Error & 0.241 (df = 609) \\ 
F Statistic & 38.000$^{***}$ (df = 2; 609) \\ 
\hline 
\hline 
\textit{Note:}  & \multicolumn{1}{r}{$^{*}$p$<$0.1; $^{**}$p$<$0.05; $^{***}$p$<$0.01} \\ 
\end{tabular}
\begin{flushleft}
   {\small \textit{Source: our elaboration on ORBIS data.}} 
\end{flushleft}
\end{table} 

\subsection{Numerical exploration}
\label{sec:numericalExperiments}

In this section, we present a series of numerical explorations, using as a baseline the labour immobile economy, with values of $\beta = 0.86$, $\sigma = 1.3144$, and $\eta = 0.0222$, as calculated in Section \ref{sec:backData}. 
We assume that $L(i) =  [\sin (2\pi(i -1/4)) + 3]/2$ (in short, the labour distribution on $S$ presents a peak at $i=0.5$), no heterogeneity in technological progress, i.e., $A(i) \equiv 1$, and symmetric preferences. Furthermore, we assume $n=1$, meaning that the circle $S$, which represents the variety space of goods, has unit length. Given our assumptions, in the long-run competitive equilibrium, the profit rate $\pi^{EQ}(i) \equiv 0.14$, while the firms' density $\mu^{EQ}(i) \equiv 1$.

\subsubsection{The dynamics of baseline economy}

Figures \ref{fig:convergenceFirmsDistribution} and \ref{fig:convergenceValueProduction} display the dynamics of firms' distribution and the value of production in our baseline economy with immobile labour, where firms' distribution evolves according to PDE \eqref{eq:stato-di-nuovoImmobileSpillover}. The $x$-axis represents the space of goods, while the $y$-axis reports firm density $\mu$ and value of production $Y$. Since sectors are arranged on a circle, the values at positions 0 and 1 coincide.
The initial firms' distribution at $t = 0$ features a minimum density at position 0.5, mirroring the inverse of the labour distribution. The distribution of value of production in Figure \ref{fig:convergenceValueProduction} more closely resembles the labour distribution (Proposition \ref{pr:shor-run-equilibrium-immobilework}).
The evolution of firm density, driven by the distribution of profit rates (Figure \ref{fig:convergenceROE}), appears non-linear across certain types of goods. In some sectors, the mass of firms initially increases (or decreases) before reversing direction, ultimately resulting in a decline (or rise). As expected, in the long-run competitive equilibrium, the distribution of firms tends to adjust to the distribution of labour.
\begin{figure}[!htbp]
\centering
\caption{The dynamics of firms and production in the baseline economy}
\label{fig:baselineCaseFirmProductionDistribution}

\begin{subfigure}[t]{0.35\textwidth}
\centering
\includegraphics[width=\linewidth]{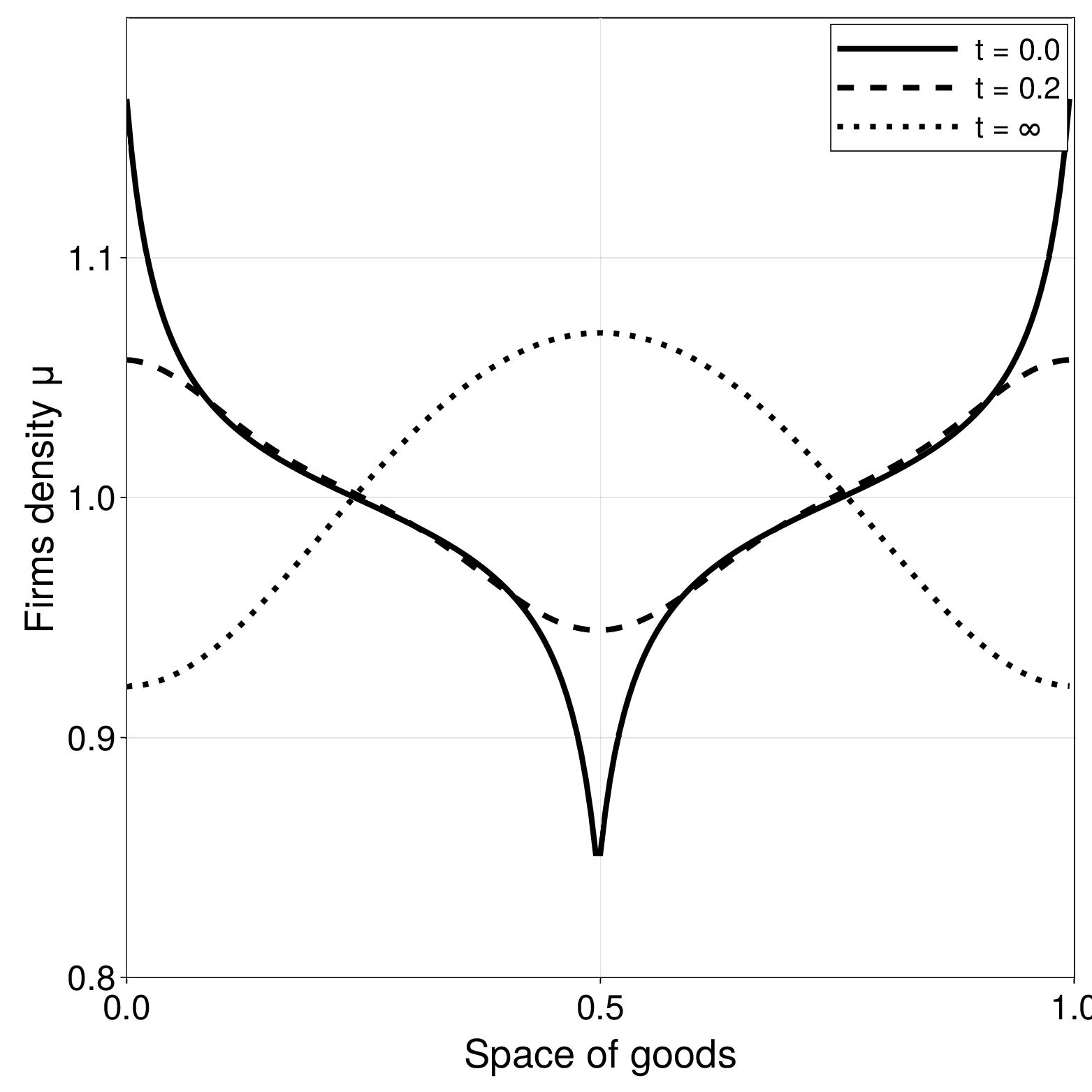}
\captionsetup{width=1.3\linewidth} 
\caption{The convergence to the long-run competitive equilibrium of the firms' distribution. On $x$-axis are reported the space of goods, while on $y$-axis the firm density.}
\label{fig:convergenceFirmsDistribution}
\end{subfigure}
\hspace{0.15\textwidth} 
\begin{subfigure}[t]{0.35\textwidth}
\centering
\includegraphics[width=\linewidth]{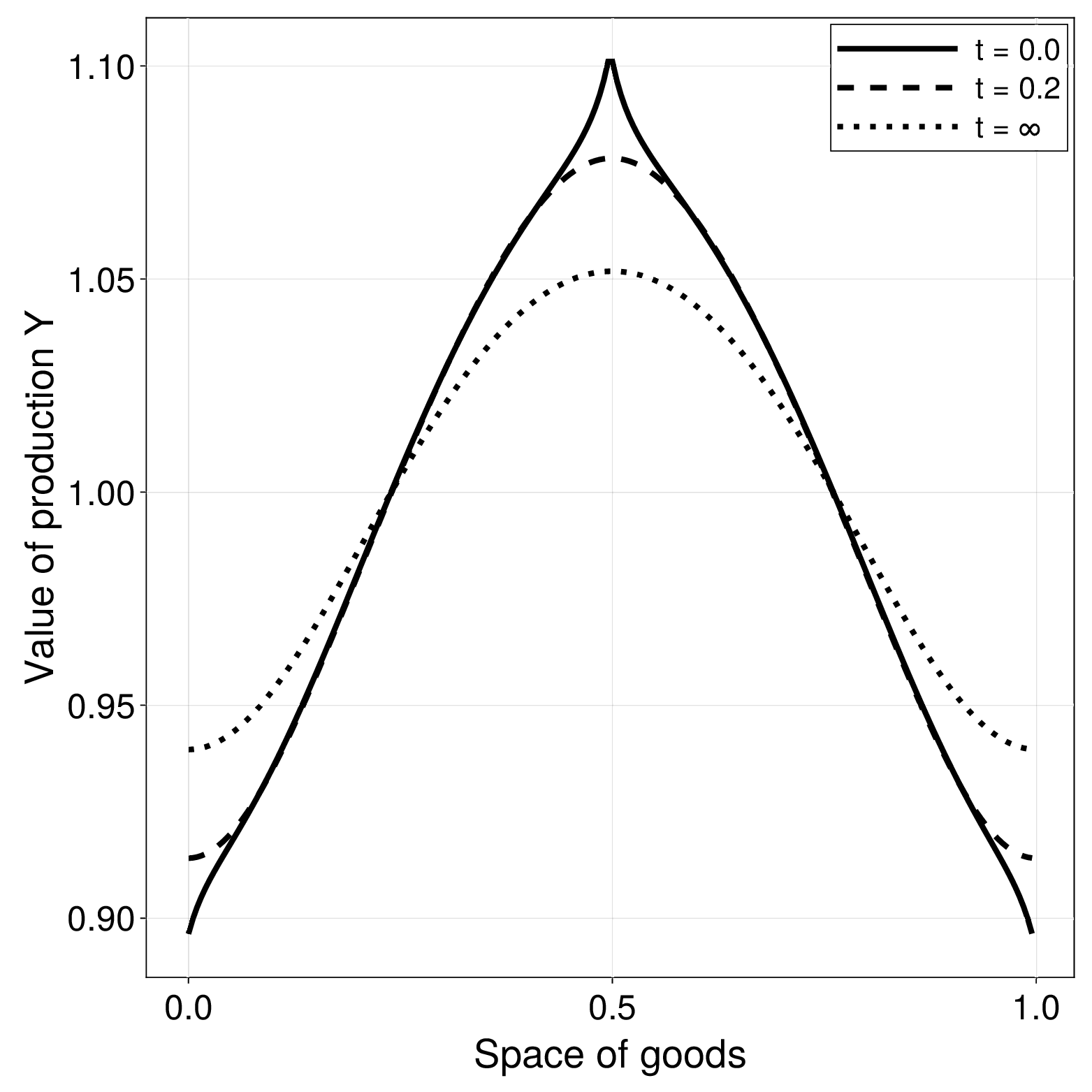}
\captionsetup{width=1.3\linewidth}
\caption{The convergence to the long-run competitive equilibrium of the production distribution. On $x$-axis are reported the space of goods, while on $y$-axis the value of production.}
\label{fig:convergenceValueProduction}
\end{subfigure}
\end{figure}

Turning our attention to profit rates, these converge non-monotonically towards a uniform equilibrium of $\pi(i)= 0.14 \; \forall i$, as expected (Figure \ref{fig:convergenceROE}). These short-run non-linear dynamics are also evident in the relationship between the initial profit rate and its subsequent change over time across different intervals, as illustrated in Figure \ref{fig:ROERegressionToTheMean}. However, when observed over a sufficiently long time horizon, the relationship becomes linear, resembling the pattern estimated in Figure \ref{fig:ROEregressionToTheMean}.
By contrast, the distribution of labour productivity remains relatively stable over time (Figure \ref{fig:convergenceLabourProductivity}), yet exhibits a clear divergent pattern similar to that observed in our firm-level sample (compare Figures \ref{fig:labourProductivityRegressionToTheMean} and \ref{fig:LabourProductivityRegressionToTheMean}).

\begin{figure}[!htbp]
\centering
\caption{The dynamics of profit rates and labour productivity in the baseline economy}
\label{fig:baselineCaseConvergenceROELabourProductivity}

\begin{subfigure}[t]{0.35\textwidth}
\centering
\includegraphics[width=\linewidth]{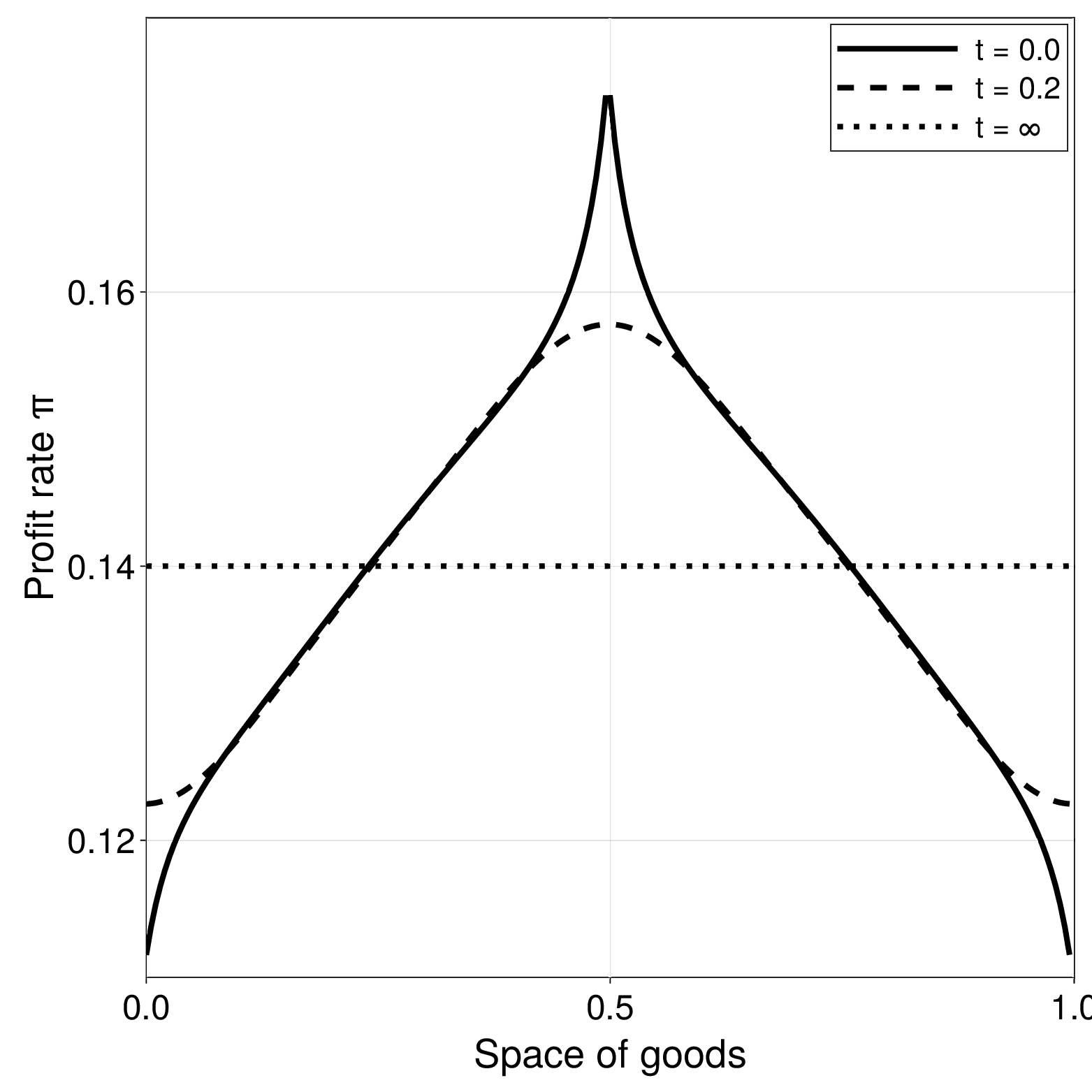}
\captionsetup{width=1.3\linewidth}
\caption{The convergence to the long-run competitive equilibrium of the profit rate distribution. On $x$-axis are reported the space of goods, while on $y$-axis the profit rates.}
\label{fig:convergenceROE}
\end{subfigure}
\hspace{0.15\textwidth}
\begin{subfigure}[t]{0.35\textwidth}
\centering
\includegraphics[width=\linewidth]{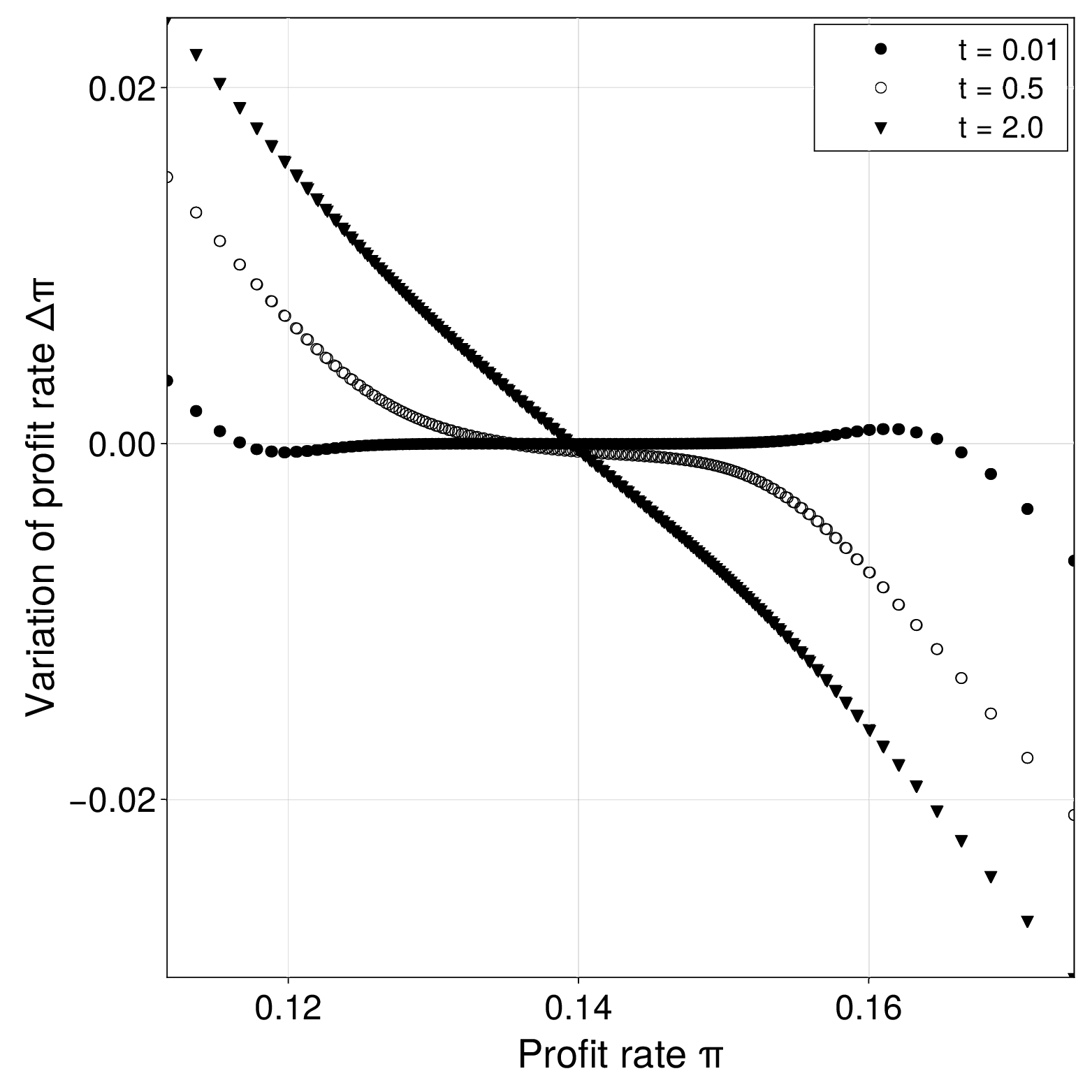}
\captionsetup{width=1.3\linewidth}
\caption{The relationship between the sectoral profit rates at the time 0 ($x$-axis) and their change over time for different time intervals (0.01, 0.5, 2) of the same ($y$-axis).}
\label{fig:ROERegressionToTheMean}
\end{subfigure}

\vspace{0.5cm} 

\begin{subfigure}[t]{0.35\textwidth}
\centering
\includegraphics[width=\linewidth]{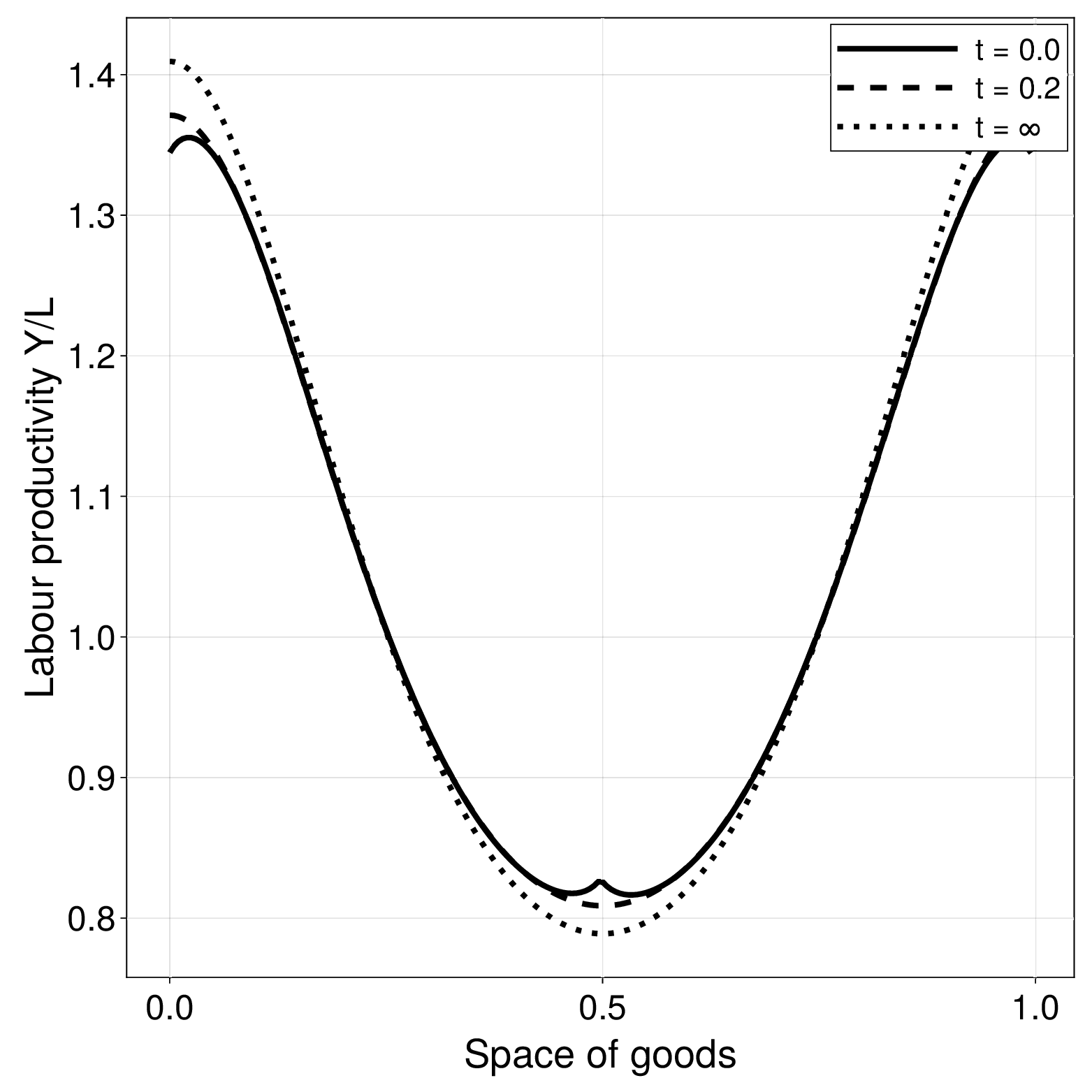}
\captionsetup{width=1.3\linewidth}
\caption{The convergence to the long-run competitive equilibrium of the labour productivity distribution. On $x$-axis are reported the space of goods, while on $y$-axis the profit rates.}
\label{fig:convergenceLabourProductivity}
\end{subfigure}
\hspace{0.15\textwidth}
\begin{subfigure}[t]{0.35\textwidth}
\centering
\includegraphics[width=\linewidth]{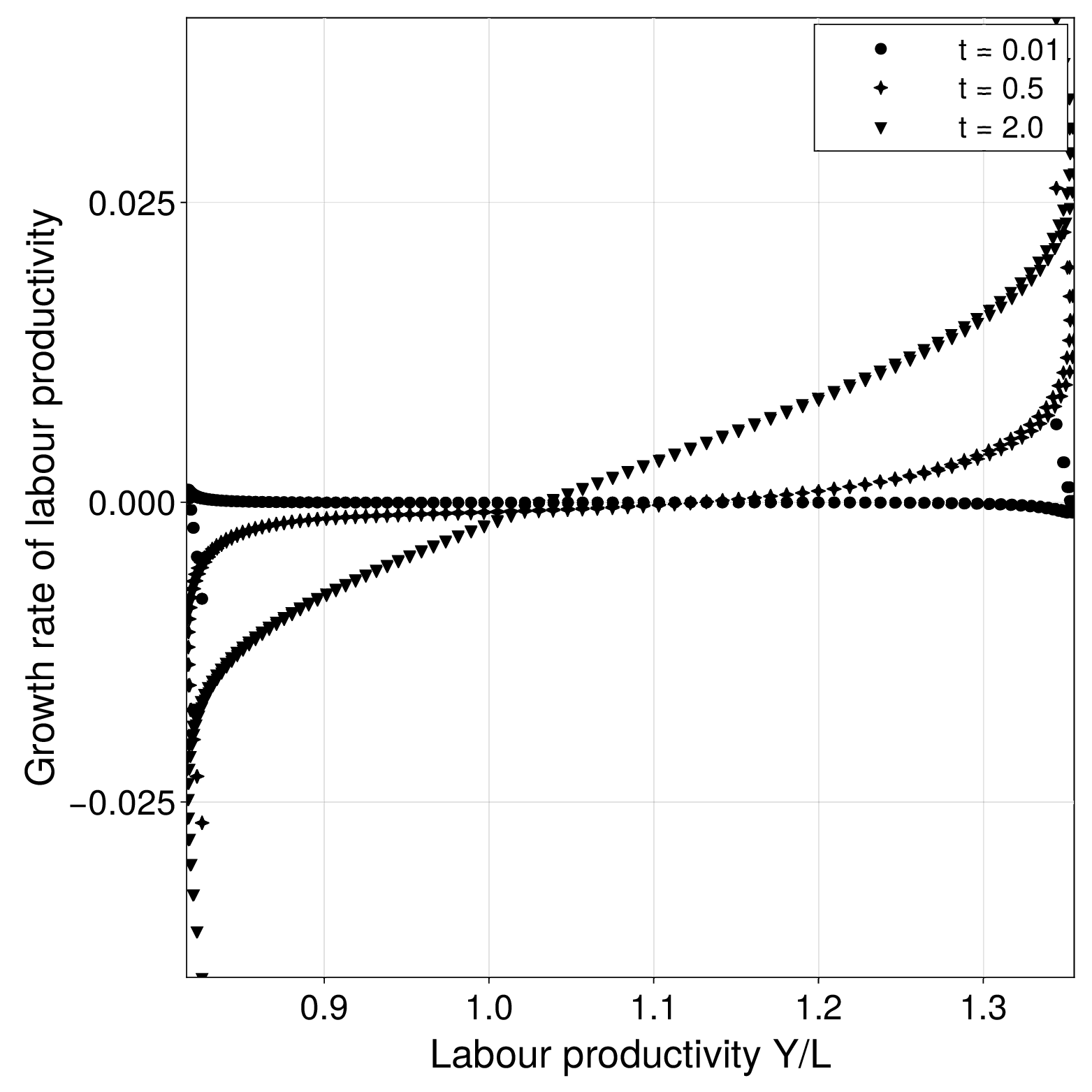}
\captionsetup{width=1.3\linewidth}
\caption{The relationship between the sectoral labour productivities at the time 0 ($x$-axis) and their change over time of the same ($y$-axis).}
\label{fig:LabourProductivityRegressionToTheMean}
\end{subfigure}

\end{figure}

\subsubsection{The convergence to the long-run competitive equilibrium}

Figure \ref{fig:speedOfConvergenceLongRunEquilibrium} illustrates the varying speed of convergence to the long-run distribution of firms, $\mu^{EQ}$, as defined in Equation \eqref{eq:muEQImmobileLabour}. This also implicitly confirms that the economy is converging toward the long-run competitive equilibrium. The $x$-axis reports time periods, while the $y$-axis shows the $L^2$ distance, which serves as a summary statistic indicating how far the economy is from equilibrium at each point in time.\footnote{The $L^2$ distance is defined as the square root of the integral of the squared differences between the current distribution of firms and the long-run equilibrium distribution.}
The convergence pattern appears to be exponential over time. It is faster for lower values of $\beta$ and slower for higher values of $\eta$.
\begin{figure}[!htbp]
\centering
\caption{The convergence to the long-run equilibrium for different configuration of parameters.}
\label{fig:convergenceLongRunEquilibrium}

\begin{subfigure}[t]{0.35\textwidth}
\centering
\includegraphics[width=\linewidth]{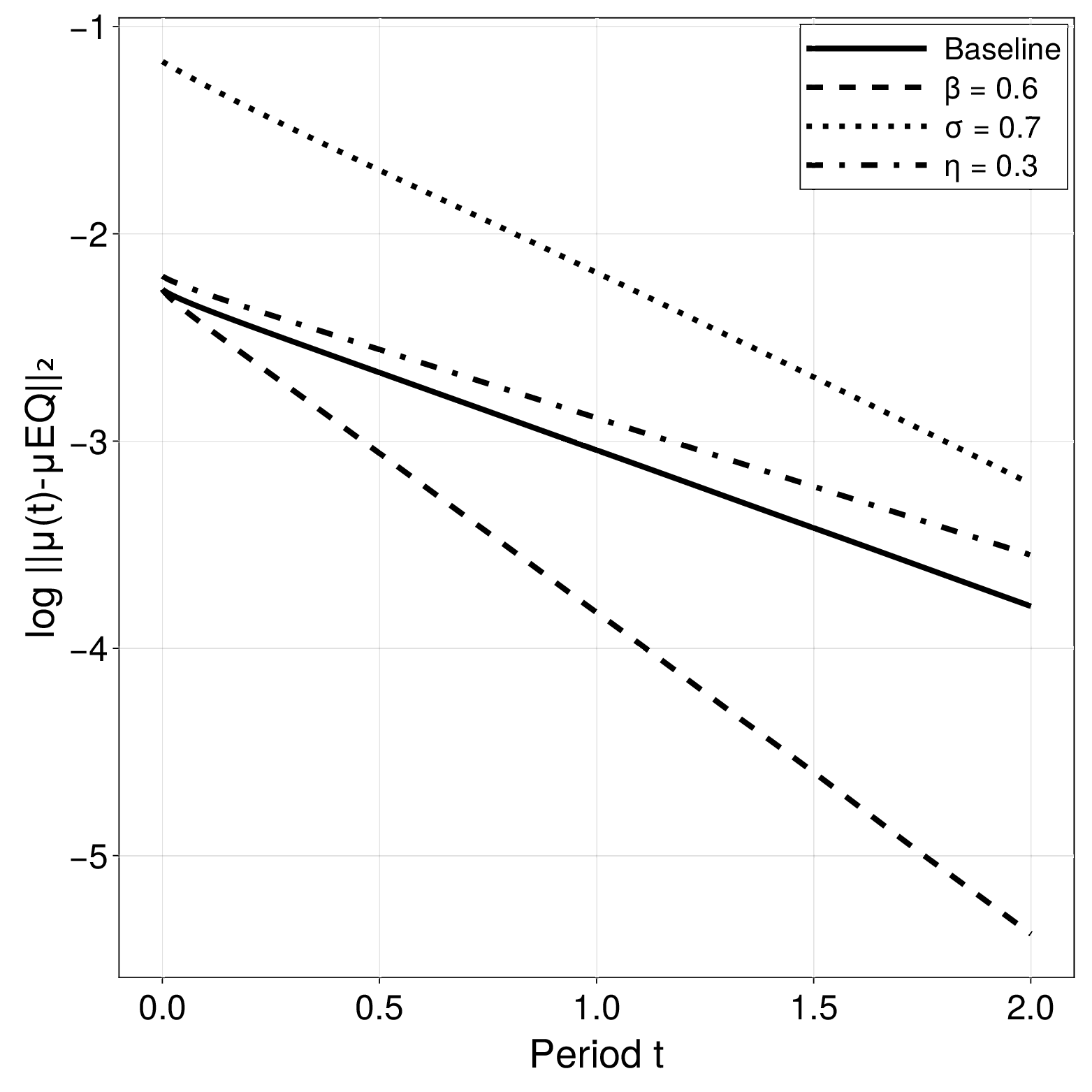}
\captionsetup{width=1.3\linewidth}
\caption{The speed of convergence to the long-run competitive equilibrium for different values of the parameters $\beta$, $\sigma$, and $\eta$, relative to the baseline economy. The $x$-axis reports the time period, while the $y$-axis shows the logarithm of the $L^2$ distance, defined as the square root of the integral of the squared differences between the current distribution of firms and the long-run equilibrium distribution.}
\label{fig:speedOfConvergenceLongRunEquilibrium}
\end{subfigure}
\hspace{0.15\textwidth}
\begin{subfigure}[t]{0.35\textwidth}
\centering
\includegraphics[width=\linewidth]{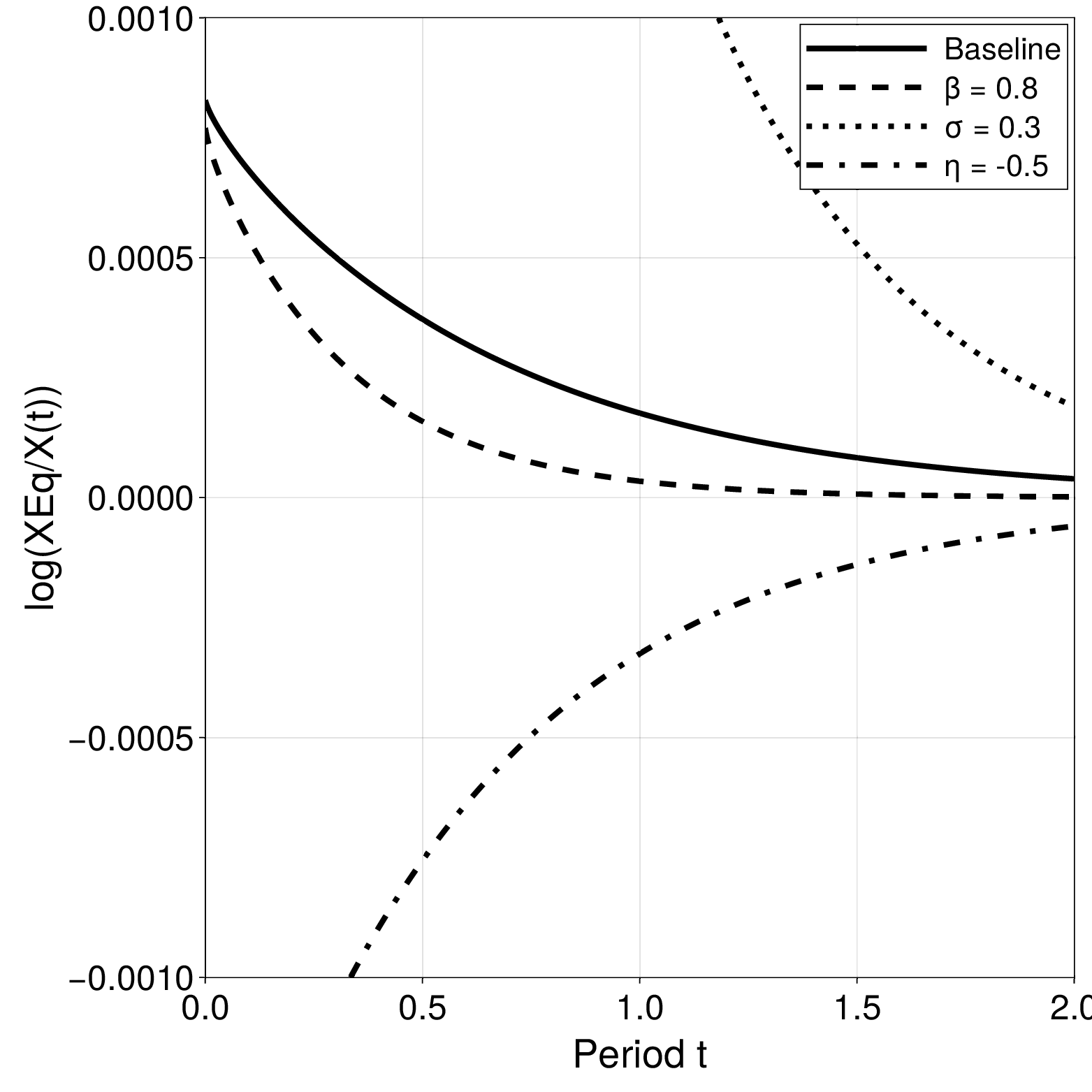}
\captionsetup{width=1.3\linewidth}
\caption{The convergence of aggregate consumption for different values of the parameters $\beta$, $\sigma$, and $\eta$, relative to the baseline economy. The $x$-axis reports the time period, while the $y$-axis shows the logarithm of the ratio between consumption in the long-run competitive equilibrium and that observed in the economy at each period.}
\label{fig:convergenceConsumption}
\end{subfigure}

\end{figure}

Figure \ref{fig:convergenceConsumption} supports the validity of the gradient flow approach in characterizing the dynamics of aggregate consumption $X$. The $x$-axis reports the time period, while the $y$-axis shows the logarithm of the ratio between consumption in the long-run competitive equilibrium and that observed in the economy at each period, $X^{EQ}/X$.
The downward trend observed for the baseline economy confirms that consumption increases over time, as expected, given that Condition \eqref{eq:condizione-spillovers} holds and $\eta > \beta - 1$ (see Equation \eqref{eq:espressionediF-coneta}). In contrast, for the economy with $\eta = -0.5$, we have $\eta < \beta - 1$, and correspondingly, the ratio exhibits an upward trend indicating that, in this case, consumption tends to be minimized in the long-run competitive equilibrium.

\subsubsection{The economy with a fixed cost of reallocation}

Figure \ref{fig:fixedCostreallocation} illustrates the main effects of introducing a fixed cost of reallocation (we set $c_0 = 0.1$ in Equation \eqref{eq:reallocationCostFixedCost}, while keeping all other parameters at their baseline values). A comparison between Figures \ref{fig:convergenceROE} and \ref{fig:Fig6} reveals that the presence of fixed reallocation costs prevents profit rates from converging to a uniform distribution.
In our numerical example, the long-run competitive distribution of profit rates remains markedly dispersed, ranging from 0.12 to 0.165. Figure \ref{fig:Fig5} provides further insight into the dynamics of profit rates: even over an extended time horizon, a wide range of profit rates persists around the long-run equilibrium value of 0.14. This reflects the existence of a plateau where reallocation does not occur, due to insufficient incentives for firms to move across sectors. This is evidenced by the flat slope of the profit rate distribution around 0.14 at the final period of the simulation ($t = 2$) in Figure \ref{fig:Fig6}.
\begin{figure}[!htbp]
\centering
\caption{The economy with fixed costs of reallocation.}
\label{fig:fixedCostreallocation}

\begin{subfigure}[t]{0.35\textwidth}
\centering
\includegraphics[width=\linewidth]{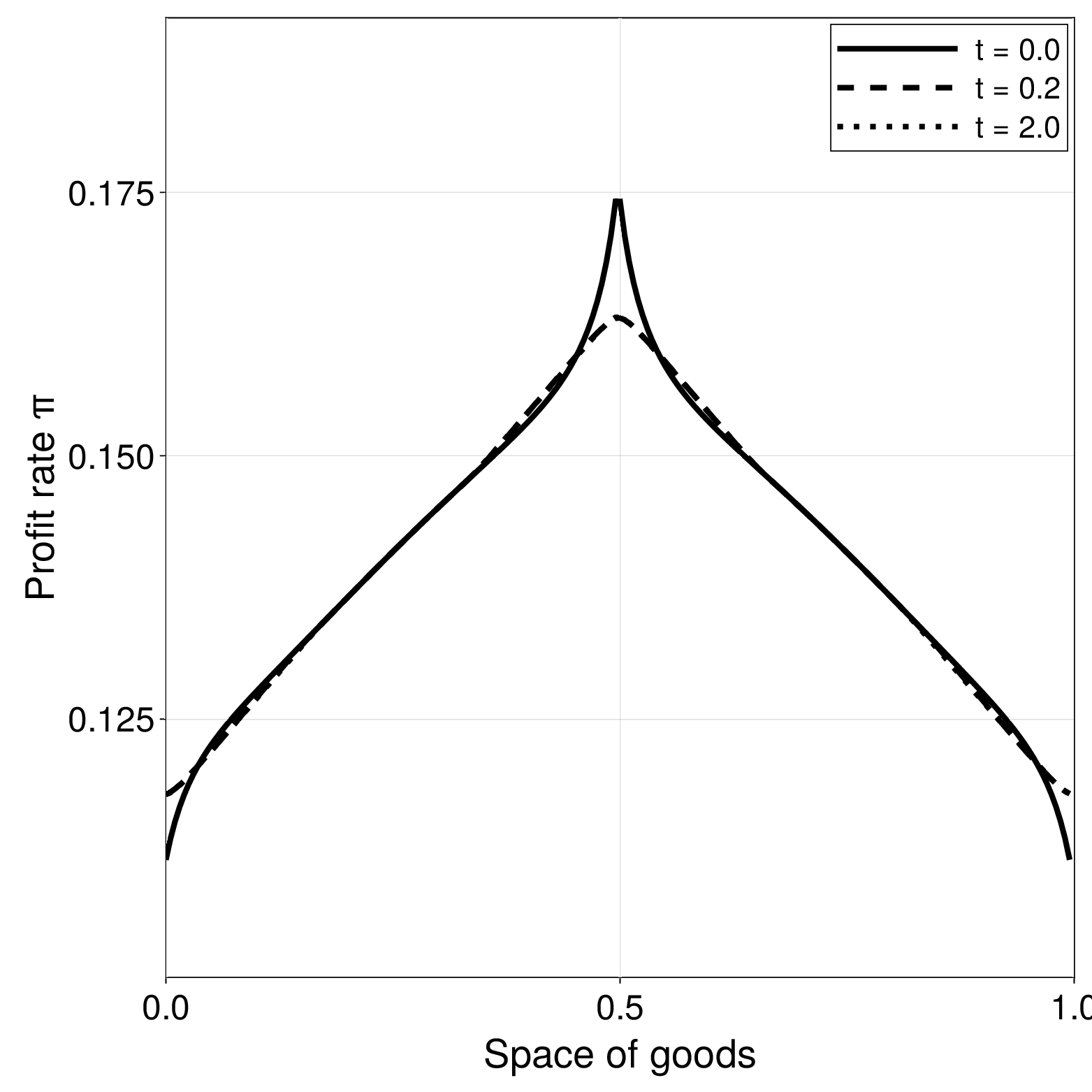}
\captionsetup{width=1.3\linewidth}
\caption{The profit rate distribution at different periods in presence of a fixed cost of reallocation. On $x$-axis are reported the space of goods, while on $y$-axis the profit rates.}
\label{fig:Fig6}
\end{subfigure}
\hspace{0.15\textwidth}
\begin{subfigure}[t]{0.35\textwidth}
\centering
\includegraphics[width=\linewidth]{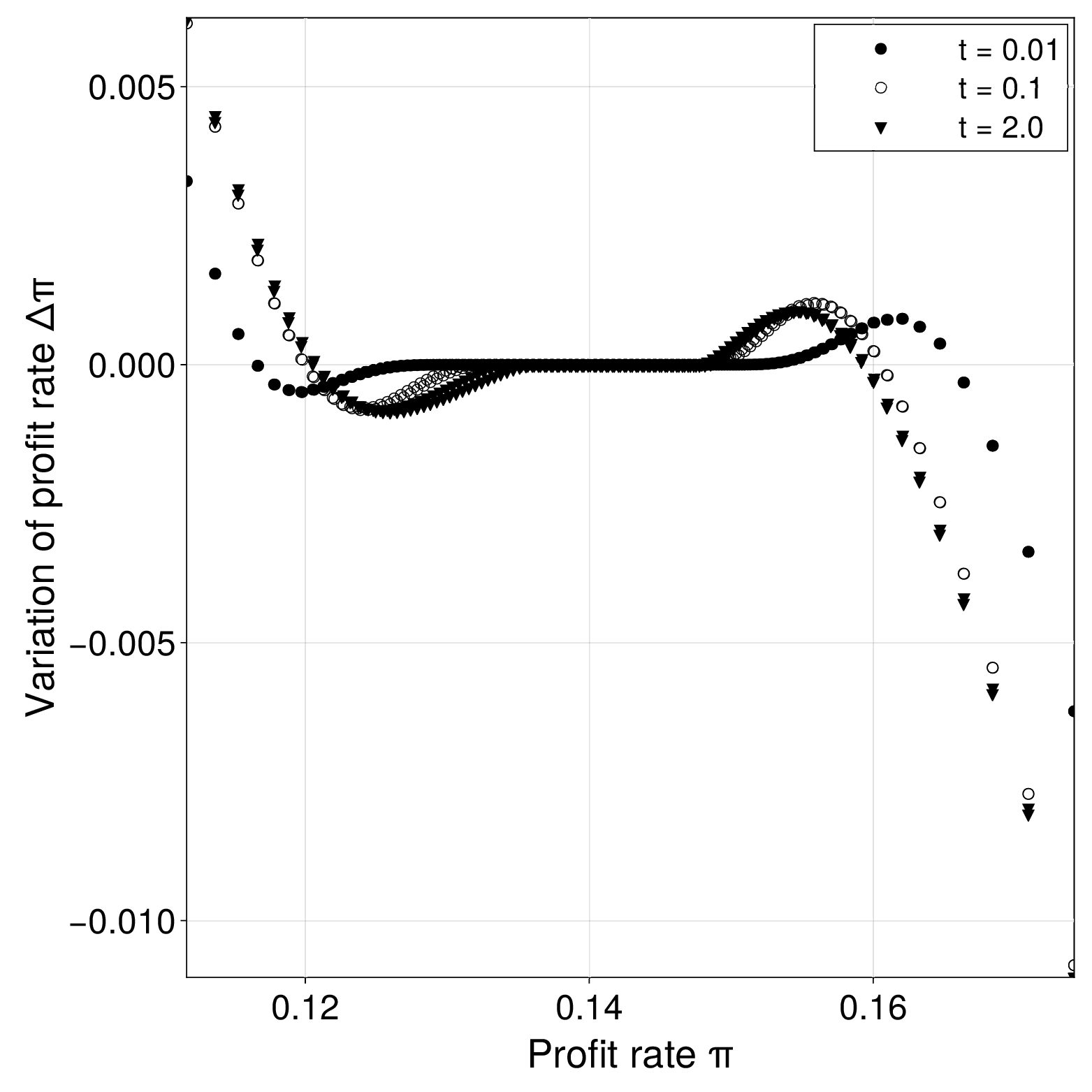}
\captionsetup{width=1.3\linewidth}
\caption{The relationship between the sectoral profit rates at the time 0 ($x$-axis) and their change over time of the same ($y$-axis) in presence of fixed costs of reallocation.}
\label{fig:Fig5}
\end{subfigure}

\end{figure}

\subsubsection{The cost of labour immobility}
\label{sec:costLabourImmobility}
We conclude our numerical experiments by quantifying the efficiency loss, proxied by a lower consumption, resulting from labour immobility  (see Proposition \ref{prop:functionalLossEfficiency}). In Figure \ref{fig:lossEfficiencyLabourImmobility}, this loss is greatest at $t = 0$, as the absence of instantaneous labour reallocation contributes maximally to the differential efficiency between mobile and immobile labour scenarios. The magnitude of this efficiency loss is approximately 2\% relative to the long-run competitive equilibrium with mobile labour (given that $A(i) \equiv 1$, we have $X^{EQ} = 1$, see Proposition \ref{prop:functionalLossEfficiency}).
Although gains from firm reallocation are higher in the less efficient (immobile labour) economy, leading to a gradual reduction in the efficiency gap, the loss remains positive and exceeds 1.92\% even in the long run.
\begin{figure}[!htbp]
\caption{The loss in consumption due to the labour immobility}
\label{fig:lossEfficiencyLabourImmobility}
\centering
\includegraphics[width=0.35\textwidth]{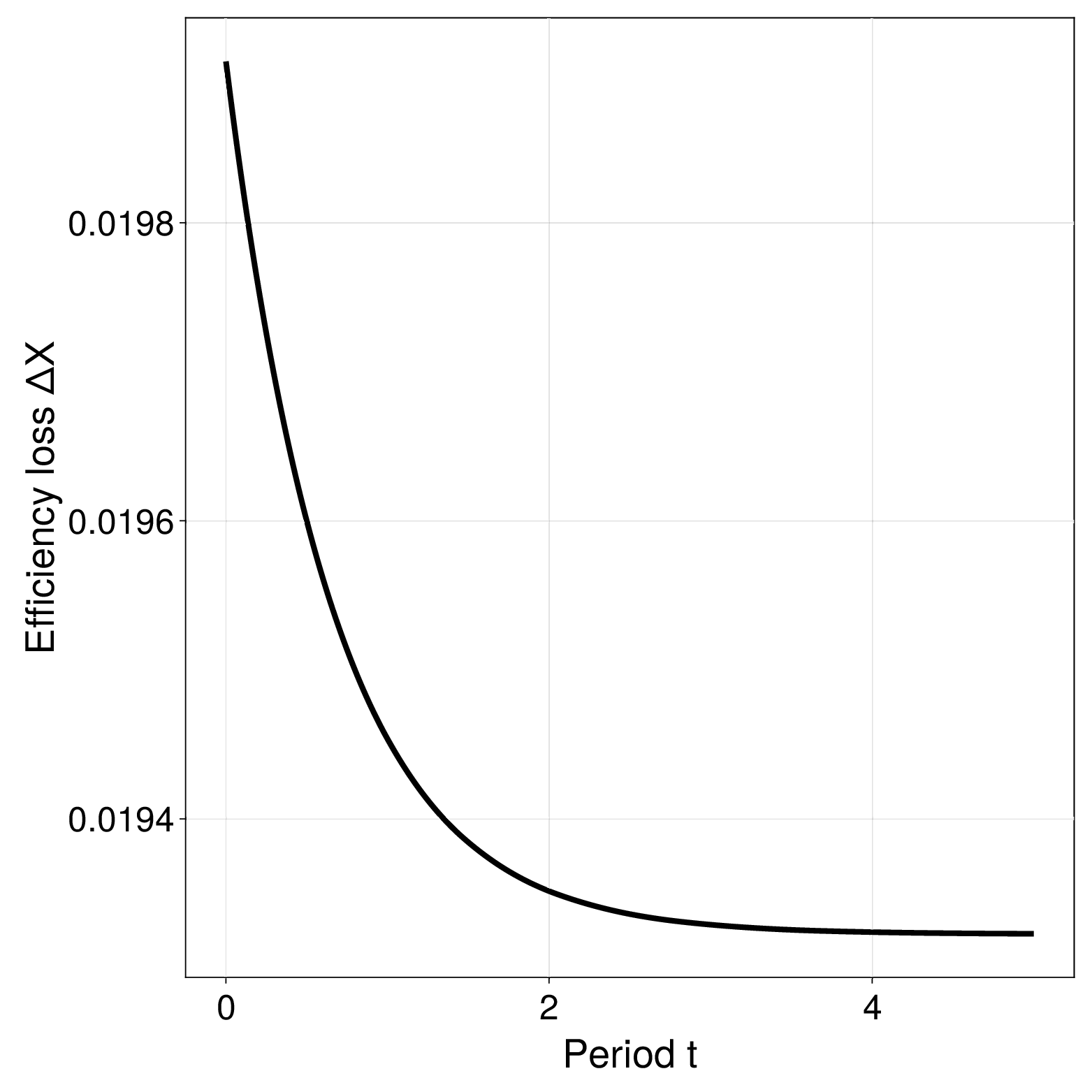}
\caption*{The dynamics of the difference in consumption between an economy with mobile labour and one with immobile labour. The $x$-axis represents the time period, while the $y$-axis shows the difference in consumption $X$ between the two economies.}
\end{figure}

\clearpage
\subsection{The identification of firm density}
\label{app:identificationTechnologicalProgress}

As we discussed, unfortunately, we have not a good proxy for the firm density in the sector $\mu(i)$. This also poses a limit to the calculation of good-specific technological parameter $\tilde{A}(i)$. The same consideration can be made for the identification of $\gamma(i)$ in the case of non-symmetric preference. Below, we discuss a possible procedure for the identification of $\mu(i)$ and $\tilde{A}(i)$ based on some assumptions on the share of distribution of $\tilde{A}(i)$ in the economy with labour immobility. The case of non-symmetric preferences is leave to future research.

From Equation \eqref{eq:profitsCapitalStockLabourImmobility}, we can observe that:
\begin{equation}
\label{eq:eq:profitsCapitalStockLabourImmobility_II}
\begin{aligned}
& \left[ \dfrac{\pi(i)}{1-\beta}\right]^{\dfrac{\sigma}{\left(\eta - \beta\right)\left( \sigma -1 \right) - 1}}   L(i)^{\dfrac{\beta \left(1-\sigma\right)}{\left(\eta - \beta\right)\left( \sigma -1 \right) - 1}} \\
& \qquad = \mu(i) \left[ \tilde{A}(i)^{\dfrac{\sigma -1 }{\left(\eta-\sigma \right)\left(\sigma-1\right) -1}} P^{\dfrac{\sigma -1 }{\left(\eta-\sigma \right)\left(\sigma-1\right) -1}}\right],
\end{aligned}
\end{equation}
where the left hand side is observable (given the value of parameters calculated above), while the right hand side is not observable since:
\begin{equation}
P  = \left(\int_0^n \left[ \tilde A(i)L(i)^\beta\mu(i)^{1+\eta -\beta}\right]^{-\frac{1-\sigma}{\sigma}} di \right)^{\frac{\sigma}{1-\sigma}}.
\end{equation}
The distribution of $\mu$ therefore appears as a \textit{coinvolution} of the distribution of $\pi$, $L$ and $\tilde{A}$, where the latter is unobservable. A feasible way could be to impose that:
\begin{equation}
    \tilde{A}(i) \sim \mathcal{LN}(m_{\tilde{A}}, \sigma_{\tilde{A}}^2)
\end{equation}
from which
\begin{equation}
    \tilde{A}(i)^{{\dfrac{\sigma -1 }{\left(\eta-\sigma \right)\left(\sigma-1\right) -1}}} \sim \mathcal{LN}\left(m_{\tilde{A}} \left[{\dfrac{\sigma -1 }{\left(\eta-\sigma \right)\left(\sigma-1\right) -1}}\right], \sigma_{\tilde{A}}^2\left[{\dfrac{\sigma -1 }{\left(\eta-\sigma \right)\left(\sigma-1\right) -1}}\right]^2\right).
\end{equation}
Therefore, taking the log of both sides of Equation \eqref{eq:eq:profitsCapitalStockLabourImmobility_II}:
\begin{equation}
    z(i) = y(i) + a(i),
    \label{eq:equationfForTweedies}
\end{equation}
where:
\begin{eqnarray}
    z(i) &\equiv& \log \left( \left[ \dfrac{\pi(i)}{1-\beta}\right]^{\dfrac{\sigma}{\left(\eta - \beta\right)\left( \sigma -1 \right) - 1}}   L(i)^{\dfrac{\beta \left(1-\sigma\right)}{\left(\eta - \beta\right)\left( \sigma -1 \right) - 1}} \right); \\
    y(i) &\equiv& \log \left( \mu(i) \right); \text{ and}\\
    a(i) &\equiv& \log \left( \tilde{A}(i)^{\dfrac{\sigma -1 }{\left(\eta-\sigma \right)\left(\sigma-1\right) -1}} \right) + \log \left( P^{\dfrac{\sigma -1 }{\left(\eta-\sigma \right)\left(\sigma-1\right) -1}} \right).
\end{eqnarray}
Summarizing, we observe $z$, we assume that \begin{equation}
    a \sim \mathcal{N}(\bar{a}, \sigma_a^2),
\end{equation}
where $\bar{a}$ and $\sigma_a^2$ are to be \textit{consistently} calculated considering that:
\begin{equation}
    \bar{a} = \left[{\dfrac{\sigma -1 }{\left(\eta-\sigma \right)\left(\sigma-1\right) -1}}\right] \left[ m_{\tilde{A}}  + \log\left( \left(\int_0^n \left[ \tilde A(i)L(i)^\beta\mu(i)^{1+\eta -\beta}\right]^{-\frac{1-\sigma}{\sigma}} di \right)^{\frac{\sigma}{1-\sigma}} \right) \right],
    \end{equation}
    and
    \begin{equation}
    \sigma_a^2 = \sigma_{\tilde{A}}^2\left[{\dfrac{\sigma -1 }{\left(\eta-\sigma \right)\left(\sigma-1\right) -1}}\right]^2,
\end{equation}
and we are interested to the distribution of $y \equiv \log(\mu)$. 

The Tweedies's Formula \citep{efron2011tweedie} gives one possible solution to the estimate of $\mu$ starting from Equation \eqref{eq:equationfForTweedies}, stating that:
\begin{equation}
    \mathbb{E}\left[y(i) \mid z(i)\right] = z(i) - \bar{a} + \sigma_a^2 \frac{d}{dz} \log f(z)|_{z=z(i)},
\end{equation}
where $f(z)$ is the marginal density of $z$.
In the estimate, $\bar{a}$ and $\sigma_a^2$ are two parameters to be estimated under the constraint that $\exp(y)$ is a marginal distribution, i.e. $\int \exp(y(i)) di =\int \mu(i) di=1$. This procedure produces an estimate of $\bar{a}$ and $\sigma_a^2$. From the latter,  it is trivial to calculate $\sigma_{\tilde{A}^2}$, i.e.
\begin{equation}
    \sigma_{\tilde{A}}^2 = \left(\sigma_a^2\right)^{\dfrac{\left(\eta-\sigma \right)\left(\sigma-1\right) -1}{ \sigma -1}}.
\end{equation}
Instead, as $m_A$, we should consider that $\mathbb{E}\left[a(i) \mid z(i)\right] = z(i) -  \mathbb{E}\left[y(i) \mid z(i)\right]$ (indeed, $\mathbb{E}_{z}\left[ \mathbb{E}\left[a(i) \mid z(i)\right] \right]=\bar{a}$), and therefore:
\begin{multline}
    \mathbb{E}\left[\tilde{A}(i)P\mid z(i)\right] = \exp\left( \left[ \dfrac{\left(\eta-\sigma \right)\left(\sigma-1\right) -1}{ \sigma -1} \right] \mathbb{E}\left[a(i) \mid z(i)\right] \right) = \\
    =\exp\left( \left[ \dfrac{\left(\eta-\sigma \right)\left(\sigma-1\right) -1}{ \sigma -1} \right] \left( z(i) -  \mathbb{E}\left[y(i) \mid z(i)\right]\right) \right),
\end{multline}
and, finally,
\begin{multline}
    \int \mathbb{E}\left[\tilde{A}(i)P\mid z(i)\right] di = P \int \mathbb{E}\left[\tilde{A}(i)\mid z(i)\right] di = P = \\
    = \int \exp\left( \left[ \dfrac{\left(\eta-\sigma \right)\left(\sigma-1\right) -1}{ \sigma -1} \right] \left( z(i) -  \mathbb{E}\left[y(i) \mid z(i)\right]\right) \right) di.
\end{multline}
Therefore:
\begin{multline}
    m_{\tilde{A}} = \left[{\dfrac{\left(\eta-\sigma \right)\left(\sigma-1\right) -1}{ \sigma -1}}\right]\bar{a} - \log P = \\  = \left[{\dfrac{\left(\eta-\sigma \right)\left(\sigma-1\right) -1}{ \sigma -1}}\right]\bar{a} - \log \left( \int \exp\left( \left[ \dfrac{\left(\eta-\sigma \right)\left(\sigma-1\right) -1}{ \sigma -1} \right] \left( z(i) -  \mathbb{E}\left[y(i) \mid z(i)\right]\right) \right) di \right),  
\end{multline}
which complete the characterization of the distribution of $\tilde{A}$. Then, using \eqref{eq:equationfForTweedies}, we get the distribution of $\mu$.
\end{footnotesize}

\end{document}